%% file: freezeml.tex
\pdfoutput=1
\newif\ifextended\extendedtrue 
\newif\ifcomments\commentsfalse 
\newif\iffinal\finaltrue       

\ifextended
\documentclass[sigplan,screen]{acmart}
\else
\documentclass[sigplan,screen]{acmart}
\fi

\usepackage{amsmath}     
\usepackage{amssymb}
\usepackage{amsfonts}
\usepackage{paralist}
\usepackage{listings}    
\usepackage{xcolor}      
\usepackage[shortcuts]{extdash} 
\usepackage{mathpartir}
\usepackage{pifont}
\usepackage{bbding}  
\usepackage{xspace}
\usepackage{natbib}
\usepackage{longtable}
\usepackage{enumitem}
\usepackage{centernot} 
\usepackage{environ}
\usepackage{balance}
\usepackage{wasysym}
\usepackage[htt]{hyphenat}

\usepackage{adjustbox} 
\usepackage{array} 

\usepackage{tikz}
\usetikzlibrary{positioning}
\usetikzlibrary{calc}
\usetikzlibrary{arrows}

\usepackage{refcount}

\usepackage{cleveref}

\usepackage{mathtools}
\newtagform{bold}[\bfseries]{\textbf{(}}{\textbf{)}}
\usetagform{bold}

\usepackage{thmtools}
\usepackage{thm-restate}
\declaretheorem[name=Theorem]{thm}

\crefname{equation}{}{}
\Crefname{equation}{Statement}{Statements}
\crefname{lemma}{Lemma}{Lemmas}
\Crefname{lemma}{Lemma}{Lemmas}
\crefname{figure}{Figure}{Figures}
\Crefname{figure}{Figure}{Figures}
\crefname{section}{Section}{Section}
\Crefname{section}{Section}{Section}
\crefname{thm}{Theorem}{Theorems}
\Crefname{thm}{Theorem}{Theorems}
\def\currentprefix{proof:}
\newcommand*{\proofContext}[1]{\def\currentprefix{proof:#1}}
\newcommand*{\localeqlabel}[1]{\label[equation]{\currentprefix:#1}}
\newcommand*{\locallabel}[1]{\label{\currentprefix:#1}}

\newcommand{\labelNum}[1]{\refstepcounter{equation}\localeqlabel{#1}\textbf{(\arabic{equation})}}
\newcommand{\mathlabelNum}[1]{\localeqlabel{#1}}

\newcommand{\refNum}[1]{%
  \renewcommand{\do}[1]{\currentprefix:##1,}%
  \def\mylist{\docsvlist{#1}}%
  \cref*{\mylist}%
}

\NewEnviron{killcontents}{}

\newcommand{\resetNum}{\setcounter{equation}{0}}
\newcommand{\fresh}{\mathrel{\#}}

\ifextended
\newcommand{\auxref}[1]{Appendix~\ref{#1}}
\newcommand{\Auxref}[1]{Appendix~\ref{#1}}
\else
\newcommand{\auxref}[1]{the extended version of this paper~\cite{extended-version}}
\newcommand{\Auxref}[1]{The extended version of this paper~\cite{extended-version}}
\fi

\iffinal
\newcommand{\citeanon}[1]{\cite{#1}\xspace}
\else
\newcommand{\citeanon}[1]{[citation anonymised]\xspace}
\fi

\iffinal
\newcommand{\expertname}{Didier R{\'{e}}my\xspace}
\else
\newcommand{\expertname}{an expert on this topic (anonymised for the purpose of this submission)\xspace}
\fi

\iffinal
\newcommand{\linkslink}{\url{https://github.com/links-lang/links}}
\else
\newcommand{\linkslink}{Link removed for double blind review.}
\fi

\iffinal
\newcommand{\sysname}{Links\xspace}
\else
\newcommand{\sysname}{Chain\xspace}
\fi

\iffinal
\newcommand{\sysdesc}{the Links\xspace}
\else
\newcommand{\sysdesc}{an ML-like typed functional\xspace}
\fi

\iffinal
\newcommand{\pseudonym}{}
\else
\newcommand{\pseudonym}{(A pseudonym for an
ML-like functional programming language.)}
\fi

\usepackage{float}
\floatstyle{boxed}
\restylefloat{figure}

\ifcomments
\newcommand{\note}[1]{{\color{red} {\textsc{NOTE:}} #1}}
\newcommand{\todo}[1]{{\color{red} {\textsc{TODO:}} #1}}
\newcommand{\draft}[1]{{\color{black!50} #1}}
\newcommand{\jrc}[1]{{\color{red} {\textsc{JRC:}} #1}}
\newcommand{\fe}[1]{{\color{red}  {\textsc{FE:}}  #1}}
\newcommand{\sam}[1]{{\color{red} {\textsc{SL:}}  #1}}
\newcommand{\js}[1]{{\color{red}  {\textsc{JS:}}  #1}}
\else
\newcommand{\note}[1]{}
\newcommand{\todo}[1]{}
\newcommand{\draft}[1]{}
\newcommand{\jrc}[1]{}
\newcommand{\fe}[1]{}
\newcommand{\sam}[1]{}
\newcommand{\js}[1]{}
\fi

\definecolor{grayhighlight}{RGB}{191,191,191}

\newcommand{\freezeml}{FreezeML\xspace}

\newcommand{\mlf}{$\mathrm{ML}^\mathrm{F}$\xspace} 
\newcommand{\xmark}{\ding{53}} 
\newcommand{\newty}{\ensuremath{\star}\xspace}   
\newcommand{\altty}{\ensuremath{\bullet}\xspace} 
\newcommand{\cons}{\mathbin{::}}
\newcommand{\append}{\mathbin{+\!\!+}}

\ifextended
\setcopyright{none}
\renewcommand\footnotetextcopyrightpermission[1]{}
\else
\setcopyright{rightsretained}
\acmPrice{}
\acmDOI{10.1145/3385412.3386003}
\acmYear{2020}
\copyrightyear{2020}
\acmISBN{978-1-4503-7613-6/20/06}
\acmConference[PLDI '20]{Proceedings of the 41st ACM SIGPLAN International Conference on Programming Language Design and Implementation}{June 15--20, 2020}{London, UK}
\acmBooktitle{Proceedings of the 41st ACM SIGPLAN International Conference on Programming Language Design and Implementation (PLDI '20), June 15--20, 2020, London, UK}
\fi

\newcommand{\ba}{\begin{array}}
\newcommand{\ea}{\end{array}}

\newcommand{\bl}{\ba[t]{@{}l@{}}}
\newcommand{\el}{\ea}

\newcommand{\bp}{\begin{center}\begin{tabular}{p{6cm}p{5cm}l}}
\newcommand{\ep}{\end{tabular}\end{center}}

\makeatletter
\newcommand{\Ginormouser}{\bBigg@{7.5}}
\newcommand{\Ginormous}{\bBigg@{6}}
\newcommand{\mediumginormous}{\bBigg@{5}}
\newcommand{\ginormous}{\bBigg@{4}}

\newcommand{\mathmodelabel}[1]{\ltx@label{#1}}

\def\aligncr{\math@cr}  
\def\arraycr{\@arraycr} 

\newcommand{\xdashrightarrow}[2][]{\ext@arrow 0359\rightarrowfill@@{#1}{#2}}
\def\rightarrowfill@@{\arrowfill@@\relax\relbar\rightarrow}
\def\arrowfill@@#1#2#3#4{%
  $\m@th\thickmuskip0mu\medmuskip\thickmuskip\thinmuskip\thickmuskip
   \relax#4#1
   \xleaders\hbox{$#4#2$}\hfill
   #3$%
}

\makeatother

\newenvironment{syntax}{\begin{displaymath}\ba{@{}l@{\quad}r@{~}c@{~}l@{}}}{\ea\end{displaymath}\ignorespacesafterend}

\newenvironment{equations}{\begin{displaymath}\ba{@{}r@{~}c@{~}l@{}}}{\ea\end{displaymath}\ignorespacesafterend}

\newenvironment{eqs}{\ba[t]{@{}r@{~}c@{~}l@{}}}{\ea}

\newcommand{\reason}[1]{\quad (\text{#1})}

\newcommand{\poly}{\star}
\newcommand{\mono}{{\mathord{\bullet}}}

\newcommand{\idsubst}{\iota}                  

\newcommand{\typ}[2]{#1 \vdash #2}

\newcommand{\wfctx}[2]{#1 \vdash #2}

\newcommand{\termwf}[2]{#1 \Vdash #2}

\newcommand{\transftofreezeml}[2][]{ \mathcal{E}#1\llbracket #2 #1\rrbracket }
\newcommand{\sem}[1]{\llbracket #1 \rrbracket}
\newcommand{\transfreezemlfof}[2][]{ \mathcal{C}#1\llbracket #2 #1\rrbracket  }
\newcommand{\transmltof}[2][]{ \mathcal{C}#1\llbracket #2 #1\rrbracket }
\newcommand{\J}[1]{\mathcal{J}\sem{#1}}

\newcommand{\slab}[1]{\textrm{#1}}            
\newcommand{\fLab}[1]{\text{\scshape{F-#1}}}  %
\newcommand{\mlLab}[1]{\text{\scshape{ML-#1}}}  %
\newcommand{\freezemlLab}[1]{\text{\scshape{#1}}}  %

\newcommand{\many}[1]{\overline{#1}}

\newcommand{\var}[1]{\mathit{#1}}    
\newcommand{\dec}[1]{\mathsf{#1}}    
\newcommand{\key}[1]{{\rm {\bf #1}}} 

\newcommand{\Int}{\dec{Int}}
\newcommand{\Bool}{\dec{Bool}}
\newcommand{\List}{\dec{List}}

\newcommand{\id}{\dec{id}}
\newcommand{\ids}{\dec{ids}}

\newcommand{\Let}{\key{let}}

\newcommand{\In}{\key{in}}

\newcommand{\gen}{\mathord{\$}}
\newcommand{\inst}{\mathord{@}}
\newcommand{\vinst}{V\!\inst}   
\newcommand{\freeze}[1]{\lceil{#1}\rceil}

\newcommand{\meta}{\mathsf}
\newcommand{\mgen}{\meta{gen}}
\newcommand{\mgend}[3][\Delta]{\mgen(#1, #2, #3)}

\newcommand{\msplit}{\meta{split}}

\newcommand{\Infer}{\meta{infer}}
\newcommand{\unify}{\meta{unify}}
\newcommand{\ftv}{\meta{ftv}}
\newcommand{\arity}{\meta{arity}}

\newcommand{\mlet}{\text{let}}
\newcommand{\mreturn}{\text{return}}
\newcommand{\massert}{\text{assert}}
\newcommand{\mfresh}{\text{assume fresh}}

\newcommand{\rsubst}{\delta}  
\newcommand{\fsubst}{\theta}  

\newcommand{\tc}{D}

\newcommand{\permutation}{\approx}

\newcommand{\conv}{\mathbin{~\simeq~}}

\newcommand{\lst}[1]{\List\;#1}
\newcommand{\lstp}[1]{\List\;(#1)}

\begin{document}

\ifextended
\fancyhead[LE]{}
\fancyhead[RO]{}
\fi

\title{\freezeml}

\subtitle{Complete and Easy Type Inference for First-Class Polymorphism}



\author{Frank Emrich}
\affiliation{\institution{The University of Edinburgh}}
\email{frank.emrich@ed.ac.uk}

\author{Sam Lindley}
\affiliation{\institution{The University of Edinburgh}}
\affiliation{\institution{Imperial College London}}
\email{sam.lindley@ed.ac.uk}

\author{Jan Stolarek}
\affiliation{\institution{The University of Edinburgh}}
\affiliation{\institution{Lodz University of Technology}}
\email{jan.stolarek@ed.ac.uk}

\author{James Cheney}
\affiliation{\institution{The University of Edinburgh}}
\affiliation{\institution{The Alan Turing Institute}}
\email{jcheney@inf.ed.ac.uk}

\author{Jonathan Coates}
\affiliation{\institution{The University of Edinburgh}}
\email{s1627856@sms.ed.ac.uk}

\begin{abstract}
ML is remarkable in providing statically typed polymorphism without the
programmer ever having to write any type annotations.
The cost of this parsimony is that the programmer is limited to a form of
polymorphism in which quantifiers can occur only at the outermost level of a
type and type variables can be instantiated only with monomorphic types.

Type inference for unrestricted System F-style
polymorphism is undecidable in general. Nevertheless, the literature abounds
with a range of proposals to bridge the gap between ML and System F.

We put forth a new proposal, \freezeml, a conservative extension of ML with two
new features.  First, let- and lambda-binders may be annotated with arbitrary
System F types. Second, variable occurrences may be \emph{frozen}, explicitly
disabling instantiation.  \freezeml is equipped with type-preserving
translations back and forth between System F and admits a type inference
algorithm, an extension of algorithm W, that is sound and complete and which
yields principal types.

%
%
%
%
%
%

\end{abstract}



\maketitle
\renewcommand{\shortauthors}{Emrich, Lindley, Stolarek, Cheney, and Coates}

\section{Introduction}
\label{sec:introduction}


%
The design of ML~\citep{MilnerTHM97} hits a sweet spot in providing statically
typed polymorphism without the programmer ever having to write type
annotations. The Hindley-Milner type inference algorithm on which ML relies is
\emph{sound} (it only yields correct types) and \emph{complete} (if a program
has a type then it will be inferred). Moreover, inferred types are
\emph{principal}, that is, most general.
Alas, this sweet spot is rather narrow, depending on a delicate balance of
features; it still appears to be an open question how best to extend ML type
inference to support first-class polymorphism as found in System F.


Nevertheless, ML has unquestionable strengths as the basis for high-level
programming languages. Its implicit polymorphism is extremely convenient for
writing concise programs.
Functional programming languages such as Haskell and OCaml employ algorithms
based on Hindley-Milner type inference and go to great efforts to reduce the
need to write type annotations on programs.
Whereas the plain Hindley-Milner algorithm supports a limited form of
polymorphism in which quantifiers must be top-level and may only be instantiated
with monomorphic types, advanced programming techniques often rely on
first-class polymorphism, where quantifiers may appear anywhere and may be
instantiated with arbitrary polymorphic types, as in System F.
However, working directly in System F is painful due to the need
for \emph{explicit} type abstraction and application.
Alas, type inference, and indeed type checking, is undecidable for System F
without type annotations~\cite{Pfenning93, Wells94}.

%
%




%
The primary difficulty in extending ML to support first-class polymorphism is
with implicit instantiation of polymorphic types: whenever a variable occurrence
is typechecked, any quantified type variables are immediately instantiated with
(monomorphic) types. Whereas with plain ML there is no harm in greedily
instantiating type variables, with first-class polymorphism there is sometimes a
non-trivial choice to be made over whether to instantiate or not.


The basic Hindley-Milner algorithm~\cite{DamasM82} restricts the use of
polymorphism in types to \emph{type schemes} of the form $\forall \vec{a}. A$
where $A$ does not contain any further polymorphism.
This means that, for example, given a function $\dec{single} : \forall
a.a \to \lst{a}$, that constructs a list of one element, and a polymorphic
function choosing its first argument $\dec{choose} : \forall a.a \to a \to a$,
the expression $\dec{single}\ \dec{choose}$ is assigned the type $\lstp{a \to
a \to a}$, for some fixed type $a$ determined by the context in which the
expression is used.
%
%
The type $\lstp{a \to a \to a}$ arises from instantiating the
quantifier of $\dec{single}$ with $a \to a \to a$.
But what if instead of constructing a list of choice functions at a fixed type,
a programmer wishes to construct a list of polymorphic choice functions of type
$\lstp{\forall a.a \to a \to a}$?
This requires instantiating the quantifier of $\dec{single}$ with a polymorphic
type $\forall a.a \to a \to a$, which is forbidden in ML, and indeed the
resulting System F type is not even an ML type scheme.
%
%
%
However, in a richer language such as System F, the expression
$\dec{single}\ \dec{choose}$ could be annotated as appropriate in order to
obtain either the type $\lstp{a \to a \to a}$ or the type $\lstp{\forall a.a \to
a \to a}$.
%

All is not lost. By adding a sprinkling of explicit type annotations, in
combination with other extensions, it is possible to retain much of the
convenience of ML alongside the expressiveness of System F.
Indeed, there is a plethora of techniques bridging the expressiveness gap
between ML and System F without sacrificing desirable type inference properties
of
ML~\citep{BotlanR03,Leijen07,Leijen08,Leijen09,RussoV09,SerranoHVJ18,VytiniotisWJ06,VytiniotisWJ08,GarrigueR99}.

However, there is still not widespread consensus on what constitutes a good
design for a language combining ML-style type inference with System F-style
first-class polymorphism, beyond the typical criteria of decidability,
soundness, completeness, and principal typing.
As \citet{SerranoHVJ18} put it in their PLDI 2018 paper, type inference in the
presence of first-class polymorphism is still ``a deep, deep swamp'' and ``no
solution (...) with a good benefit-to-weight ratio has been presented to date''.
While previous proposals offer considerable expressive power, we nevertheless
consider the following combination of design goals to be both compelling and not
yet achieved by any prior work:
\begin{asparaitem}
  \item \textbf{Familiar System F types} Our ideal solution would use exactly
  the type language of System F. Systems such as HML~\citep{Leijen08},
  MLF~\citep{BotlanR03}, Poly-ML\footnote{The name Poly-ML does not appear in
  the original~\cite{GarrigueR99} paper, but was introduced
  retrospectively~\cite{BotlanR03}.}~\citep{GarrigueR99}, and
  QML~\citep{RussoV09}, capture (or exceed) the power of System F, but employ a
  strict superset of System F's type language. Whilst in some cases this
  difference is superficial, we consider that it does increase the burden on the
  programmer to understand and use these systems effectively, and may also
  contribute to increasing the syntactic overhead and decreasing the clarity of
  programs.

\item \textbf{Close to ML type inference} Our ideal solution would
  conservatively extend ML and standard Hindley-Milner type inference, including
  the (now-standard) \emph{value restriction}~\cite{Wright95}, without being
  tied to one particular type inference algorithm. Systems such as MLF and Boxy
  Types have relied on much more sophisticated type inference techniques than
  needed in classical Hindley-Milner type inference, and proven difficult to
  implement or extend further because of their complexity.  Other systems, such
  as GI, are relatively straightforward to implement atop an OutsideIn(X)-style
  constraint-based type inference algorithm, but would be much more work to add
  to a standard Hindley-Milner implementation.

\item \textbf{Low syntactic overhead} Our ideal solution would provide first-class
  polymorphism without significant departures from ordinary ML-style
  programming. Early systems~\cite{Remy94,LauferO94,OderskyL96,Jones97} showed
  how to accommodate System F-style polymorphism by associating it with nominal
  datatype constructors, but this imposes a significant syntactic overhead to
  make use of these capabilities, which can also affect the readability and
  maintainability of programs. All previous systems necessarily involve some
  type annotations as well, which we also desire to minimise as much as
  possible.

\item \textbf{Predictable behaviour} Our ideal solution would avoid guessing
  polymorphism and be specified so that programmers can anticipate where type
  annotations will be needed.  More recent systems, such as HMF~\citep{Leijen07}
  and GI~\cite{SerranoHVJ18}, use System F types, and are relatively easy to
  implement, but employ heuristics to guess one of several different polymorphic
  types, and require programmer annotations if the default heuristic behaviour
  is not what is needed.
\end{asparaitem}

In short, we consider that the problem of extending ML-style type inference with
the power of System F is solved as a \emph{technical} problem by several
existing systems, but there remains a significant \emph{design} challenge to
develop a system that uses \emph{familiar System F types}, is \emph{close to ML
type inference}, has \emph{low syntactic overhead}, and has \emph{predictable
behaviour}.  Of course, these desiderata represent our (considered, but
subjective) views as language designers, and others may (and likely will)
disagree.  We welcome such debate.

\paragraph{Our contribution: \freezeml}
In this paper, we introduce \freezeml, a core language extending ML
with two System F-like features:
\begin{itemize}
\item ``frozen'' variable occurrences for which polymorphic instantiation is
  inhibited (written $\freeze{x}$ to distinguish them from ordinary
  variables $x$ whose polymorphic types are implicitly instantiated); and
\item type-annotated lambda abstractions $\lambda (x : A).M$.
\end{itemize}
\freezeml also refines the typing rule for let by:
\begin{itemize}
\item restricting let-bindings to have principal types; and
\item allowing type annotations on let-bindings.
\end{itemize}
In \freezeml explicit type annotations are only required on lambda binders used
in a polymorphic way, and on let-bindings that assign a non-principal type to a
let-bound term; annotations are not required (or allowed) anywhere else.
As we shall see in Section~\ref{sec:essence}, the introduction of type\-/annotated
let-bindings and frozen variables allows us to macro-express explicit versions
of generalisation and instantiation (the two features that are implicit in plain
ML).
Thus, unlike ML, although \freezeml still has ML-like variables and
let-binding it \emph{also} enjoys explicit encodings of all of the underlying
System F features.
Correspondingly, frozen variables and type-annotated let-bindings are also
central to encoding type abstraction and type application of System F
(Section~\ref{sec:f-to-freezeml}).
%
Although, as we explain later, our approach is similar in expressiveness to
existing proposals such as Poly-ML, we believe its close alignment with System F
types and ML type inference are important benefits, and we argue via examples
that its syntactic overhead and predictability compare favourably with the state
of the art.  Nevertheless, further work would need to be done to systematically
compare the syntactic overhead and predictability of our approach with existing
systems --- this criticism, however, also applies to most previous work on new
language design ideas.

A secondary technical
contribution we make is to repair technical problem faced by FreezeML and some
previous systems.  In FreezeML, we restrict generalisation to principal types.
However, directly incorporating this constraint into the type system results in
rules that are syntactically not well-founded. We clarify that the typing
relation can still be defined and inductive reasoning about it is still sound.
This observation may also apply to other systems, such as HMF~\cite{Leijen08} and Poly-ML~\cite{GarrigueR99},
where the same issue arises but was not previously addressed.

\paragraph{Contributions}
This paper is a programming language design paper. Though we have an
implementation on top of \sysdesc programming
language~\citeanon{CooperLWY06}\footnote{\linkslink} implementation is not the
primary focus.
The paper makes the following main contributions:
\begin{itemize}
\item A high-level introduction to \freezeml (Section~\ref{sec:essence}).

\item A type system for \freezeml as a conservative extension of ML with the
  expressive power of System F (Section~\ref{sec:calculi}).

\item Local type-preserving translations back and forth between System F and
  \freezeml, and a discussion of the equational theory of \freezeml (Section~\ref{sec:translations}).

\item A type inference algorithm for \freezeml as an extension of algorithm
  W~\citep{DamasM82}, which is sound, complete, and yields principal types
  (Section~\ref{sec:inference}).
\end{itemize}

Section~\ref{sec:implementation} discusses implementation,
Section~\ref{sec:related-work} presents related work and
Section~\ref{sec:conclusion} concludes.

\section{An Overview of \freezeml}
\label{sec:essence}

We begin with an informal overview of \freezeml. Recall that the types of
\freezeml are exactly those of System F.

\paragraph{Implicit Instantiation}
In \freezeml (as in plain ML), when variable occurrences are typechecked, the
outer universally quantified type variables in the variable's type are
instantiated implicitly.
Suppose a programmer writes $\dec{choose}\ \id$, where $\dec{choose} : \forall
a. a \to a \to a$ and $\id : \forall a.a \to a$.
The quantifier in the type of $\id$ is implicitly instantiated with an as yet
unknown type $a$, yielding the type $a \to a$.
The type $a \to a$ is then used to instantiate the quantifier in the type of
$\dec{choose}$, yielding $\dec{choose}\ \id : (a \to a) \to (a
\to a)$.
The concrete type of $a$ depends on the context in which the expression is
used. For instance, if we were to apply $\dec{choose}\ \id$ to an increment
function then $a$ would be unified with $\Int$.
(For the formal treatment of type inference in Section~\ref{sec:inference} we
will be careful to explicitly distinguish between \emph{rigid} type variables,
like those bound by the quantifiers in the types of $\dec{choose}$ and $\id$,
and \emph{flexible} type variables, like the $a$ in the type inferred for the
expression $\dec{choose}\ \id$.)

\paragraph{Explicit Freezing \emph{($\freeze{x}$)}}

The programmer may explicitly prevent a variable from having its already existing quantifiers
instantiated by using the freeze operator $\freeze{-}$.
Whereas each ordinary occurrence of $\dec{choose}$ has type $a \to a \to a$ for
some type $a$, a frozen occurrence $\freeze{\dec{choose}}$ has type $\forall
a.a \to a \to a$.
More interestingly, whereas the term $\dec{single}\ \dec{choose}$ has type
$\lstp{a \to a \to a}$, the term $\dec{single}\ \freeze{\dec{choose}}$
has type $\lstp{\forall a.a \to a \to a}$.
%
This makes it possible to pass polymorphic arguments to functions that expect
them.
Consider a function $\dec{auto} : (\forall a.a \to a) \to (\forall a.a \to a)$.
Whereas the term $\dec{auto}\ \id$ does not typecheck (because $\id$ is
implicitly instantiated to type $a \to a$ which does not match the argument type
$\forall a.a \to a$ of $\dec{auto}$) the term $\dec{auto}\ \freeze{\id}$ does.


\paragraph{Explicit Generalisation \emph{($\gen{V}$)}}

We can generalise an expression to its principal polymorphic type by binding it to a variable and then freezing it, for instance: $\Let\;
id = \lambda x.x \;\In\; \dec{poly}\,\freeze{id}$, where $\dec{poly} : (\forall
a. a \to a) \to \Int \times \Bool$.
The explicit generalisation operator $\gen$ generalises the type of any value.
Whereas the term $\lambda x.x$ has type $a \to a$, the term
$\gen(\lambda x.x)$ has type $\forall a. a \to a$, allowing us to write
$\dec{poly}\,\gen(\lambda x.x)$.
Explicit generalisation is macro-expressible~\cite{Felleisen91} in \freezeml.
%
\[
\gen V \equiv \Let\; x = V \;\In\; \freeze{x}
\]
We can also define a type-annotated variant:
\[
\gen^A V \equiv \Let\; (x : A) = V \;\In\; \freeze{x}
\]
Note that \freezeml adopts the ML value restriction~\cite{Wright95}; hence let
generalisation only applies to syntactic values.

\paragraph{Explicit Instantiation \emph{($\inst{M}$)}}

As in ML, the polymorphic types of variables are implicitly instantiated when
typechecking each variable occurrence. Unlike in ML,
other terms can have polymorphic types, which are \emph{not} implicitly instantiated.
Nevertheless, we can instantiate a term by binding it to a variable: $\Let\; x =
\dec{head}\,\dec{ids} \;\In\; x\,\,42$, where $\dec{head} : \forall a. \lstp{a}
\to a$ returns the first element in a list and $\dec{ids} : \lstp{\forall a. a \to
a}$ is a list of polymorphic identity functions.
The explicit instantiation operator $@$ supports instantiation of a term without
having to explicitly bind it to a variable.
For instance, whereas the term $\dec{head}\,\ids$ has type $\forall a.a \to a$
the term $(\dec{head}\,\dec{ids})\inst$ in the context of application to 42 has type $\Int \to \Int$, so
$(\dec{head}\,\dec{ids})\inst\,42$ is well-formed. Explicit instantiation is
macro-expressible in \freezeml:
\[
M\inst \equiv \Let\; x = M \;\In\; x
\]

\paragraph{Ordered Quantifiers}

Like in System F, but unlike in ML, the order of quantifiers matters.
Quantifiers introduced through generalisation are ordered by the sequence in
which they first appear in a type.
Type annotations allow us to specify a different quantifier order, but
variable instantiation followed by generalisation restores the canonical order.
For example, if we have functions $\var{f} : (\forall a\ b.a \to b \to a \times
b) \to \Int$, $\dec{pair} : \forall a\ b. a \to b \to a \times b$, and
$\dec{pair'} : \forall b\ a. a \to b \to a \times b$, then
$f\ \freeze{\dec{pair}}$, $f\ \gen{\dec{pair}}$, $f\ \gen{\dec{pair'}}$ have
type $\Int$ and behave identically, whereas $f\ \freeze{\dec{pair'}}$ is
ill-typed.
%

%


\paragraph{Monomorphic parameter inference}
\label{sec:monomorphism-restriction}
As in ML, function arguments need not have annotations, but their inferred types
must be monomorphic, i.e. we cannot typecheck $\dec{bad}$:
\[
\bl
\dec{bad} = \lambda f.(f\,42, f\,\dec{True}) \\
\el
\]
 Unlike in ML we can annotate arguments with polymorphic types and use them at different types:
\[
\bl
\dec{poly} = \lambda (f : \forall a.a \to a).(f\,42, f\,\dec{True}) \\
\el
\]
%
%
%
One might hope that it is safe to infer polymorphism by local, compositional
reasoning, but that is not the case.
Consider the following two functions.
\[
\bl
\dec{bad1} = \lambda f.(\dec{poly}\,\freeze{f}, (f\,42)+1) \\
\dec{bad2} = \lambda f.((f\,42)+1, \dec{poly}\,\freeze{f})
\el
\]
We might reasonably expect both to be typeable by assigning the type $\forall
a.a \to a$ to $f$.
Now, assume type inference is left-to-right.
%
%
In $\dec{bad1}$ we first infer that $f$ has type $\forall a.a \to a$ (as
$\freeze{f}$ is the argument to $\dec{poly}$); then we may instantiate $a$ to
$\Int$ when applying $f$ to $42$.
In $\dec{bad2}$ we eagerly infer that $f$ has type $\Int \to \Int$; now when we
pass $\freeze{f}$ to $\dec{poly}$, type inference fails.
To rule out this kind of sensitivity to the order of type inference, and the
resulting incompleteness of our type inference algorithm, we insist that
unannotated $\lambda$-bound variables be monomorphic. This in turn entails
checking monomorphism constraints on type variables and maintaining other
invariants (\cref{sec:design-considerations}).
%
(One can build more sophisticated systems that defer determining whether a term
is polymorphic or not until more information becomes available --- both Poly-ML
and MLF do, for instance --- but we prefer to keep things simple.)




\subsection{\freezeml by Example}
\label{sec:examples}

\setlength{\jot}{0ex}

\input{figures/example_table.tex}
\input{figures/example_signatures.tex}

Figure~\ref{fig:freezeml-examples} presents a collection of \freezeml examples
that showcase how our system works in practice.  We use functions with type
signatures shown in Figure~\ref{fig:function-signatures} (adapted from
\citet{SerranoHVJ18}).
In Figure~\ref{fig:freezeml-examples} well-formed expressions are annotated with a
type inferred in \freezeml, whilst ill-typed expressions are annotated with
\xmark. Sections A-E of the table are taken from~\cite{SerranoHVJ18}. Section F
of the table contains additional examples which further highlight the behaviour
of our system.
Examples F1-F4 show how to define some of the functions and values in
Figure~\ref{fig:function-signatures} in \freezeml.
In \freezeml it is sometimes possible to infer a different type depending on the
presence of freeze, generalisation, and instantiation operators. In such cases
we provide two copies of an example in Figure~\ref{fig:freezeml-examples}, the
one with extra \freezeml annotations being marked with \altty.  Sometimes
explicit instantiation, generalisation, or freezing is mandatory to make an
expression well-formed in \freezeml. In such cases there is only one,
well-formed copy of an example marked with a \newty, e.g. A9\newty.
Example F10$\dagger$ typechecks only in a system without a value restriction due
to generalisation of an application.

\section{\freezeml via System F and ML}
\label{sec:calculi}

In this section we give a syntax-directed presentation of \freezeml and discuss
various design choices that we have made.  We wish for \freezeml to be an
ML-like call-by-value language with the expressive power of System F.  To this
end we rely on a standard call-by-value definition of System F, which
additionally obeys the value restriction (i.e. only values are allowed under
type abstractions). We take mini-ML~\cite{Clement86} as a core representation of
a call-by-value ML language. Unlike System F, ML separates monotypes from
(polymorphic) type schemes and has no explicit type abstraction and application.
Polymorphism in ML is introduced by generalising the body of a let-binding, and
eliminated implicitly when using a variable. Another crucial difference between
System F and ML is that in the former the order of quantifiers in a polymorphic
type matters, whereas in the latter it does not.
Full definitions of System F and ML, including the syntax, kinding and typing
rules, as well as translation from ML to System F, are given
in \auxref{sec:core-calculi}.

\paragraph*{Notations.}
We write $\ftv(A)$ for the sequence of distinct free type variables of a type in
the order in which they first appear in $A$. For example, $\ftv((a \to b) \to (a \to
c)) = a,b,c$.
%
%
%
%
Whenever a kind environment $\Delta$ appears as a domain of a substitution or a
$\forall$ quantifier, it is allowed to be empty.  In such case we identify type
$\forall \Delta. H$ with $H$.
We write $\Delta - \Delta'$ for the restriction of $\Delta$ to those type
variables that do not appear in $\Delta'$.
We write $\Delta \mathbin{\#} \Delta'$ to mean that the type variables in
$\Delta$ and $\Delta'$ are disjoint.
Disjointedness is also implicitly required when concatenating $\Delta$ and
$\Delta'$ to $\Delta, \Delta'$.

\subsection{\freezeml}
\label{sec:freezeml}


\freezeml is an extension of ML with two new features. First, let-bindings and
lambda-bindings may be annotated with arbitrary System F
types. Second, \freezeml adds a new form $\freeze{x}$, called \emph{frozen
variables}, for preventing variables from being instantiated.

The syntax of \freezeml is given in Figure~\ref{fig:freezeml-syntax}.
(We name the syntactic categories for later use in
Section~\ref{sec:inference}.)
The types are the same as in System F.
We explicitly distinguish two kinds of type: a monotype ($S$), is as in ML a
type entirely free of polymorphism, and a guarded type ($H$) is a type with no
top-level quantifier (in which any polymorphism is guarded by a type
constructor).
The terms include all ML terms plus frozen variables ($\freeze{x}$) and lambda-
and let-bindings with type ascriptions.
Values are those terms that may be generalised under the value restriction. They are slightly more general
than the value forms of Standard ML in that they are closed under let binding
(as in OCaml).
Guarded values are those values that can only have guarded types (that is, all
values except those that have a frozen variable in tail position).

\input{figures/freezeml_syntax.tex}

\input{figures/freezeml_kinding_instantiation.tex}

The \freezeml kinding judgement $\Delta \vdash A : K$ states that type $A$ has kind
$K$ in kind environment $\Delta$. The kinding rules are given in
Figure~\ref{fig:freezeml-kinding}.
As in ML we distinguish monomorphic types ($\mono$) from polymorphic types
($\poly$).
Unlike in ML polymorphic types can appear inside data type constructors.

Rules for type instantiation are given in Figure~\ref{fig:freezeml-inst}.  The
judgement $\Delta \vdash \rsubst : \Delta' \Rightarrow \Delta''$ defines a
well-formed finite map from type variables in $\Delta, \Delta'$ into type
variables in $\Delta, \Delta''$, such that $\rsubst(a) = a$ for every
$a \in \Delta$.
As such, it is only well-defined if $\Delta$ and $\Delta'$ are disjoint and
$\Delta$ and $\Delta''$ are disjoint.
Type instantiation accounts for polymorphism by either being restricted to
instantiate type variables with monomorphic kinds only ($\Rightarrow_\mono$) or
permitting polymorphic instantiations ($\Rightarrow_\poly$).  The following rule is
admissible
\begin{mathpar}
\inferrule
  {\Delta, \Delta' \vdash A : K \\ \Delta \vdash \rsubst : \Delta' \Rightarrow_{K'} \Delta''}
  {\Delta, \Delta'' \vdash \rsubst(A) : K \sqcup K'}
\end{mathpar}
where $\mono \sqcup \mono = \mono$ and $\mono \sqcup \poly = \poly \sqcup \mono
= \poly \sqcup \poly = \poly$.  We apply type instantiation in a standard way,
taking care to account for shadowing of type variables
(Figure~\ref{fig:freezeml-substitution-application}).

\input{figures/type_instantiation.tex}

\input{figures/freezeml_typing_instantiation.tex}

\input{figures/freezeml_aux_judgments.tex}

The \freezeml judgement $\Delta; \Gamma \vdash M : A$ states that term $M$ has
type $A$ in kind environment $\Delta$ and type environment $\Gamma$; its rules
are shown in Figure~\ref{fig:freezeml-typing}.
These rules are adjusted with respect to ML to allow full System F types
everywhere except in the types of variables bound by unannotated lambdas, where
only monotypes are permitted.

As in ML, the \freezemlLab{Var} rule implicitly instantiates variables.
The $\poly$ in the judgement $\Delta \vdash \rsubst
: \Delta' \Rightarrow_\poly \cdot$ indicates that the type variables in
$\Delta'$ may be instantiated with polymorphic types.
%
%
The \freezemlLab{Freeze} rule differs from the \freezemlLab{Var} rule only in
that it suppresses instantiation.
In the \freezemlLab{Lam} rule, the restriction to a syntactically monomorphic
argument type ensures that an argument cannot be used at different types inside
the body of a lambda abstraction.
%
%
 However, the type of an unannotated lambda abstraction may subsequently be
generalised.
For example, consider the expression $\dec{poly}\ \gen{(\lambda x. x)}$.
%
The parameter $x$ cannot be typed with a polymorphic type;
giving the syntactic monotype $a$ to $x$ yields type $a \to a$ for the lambda-abstraction.
The $\gen$ operator then generalises this to $\forall a. a \to a$ as the type
of argument passed to $\dec{poly}$.
%
%
The \freezemlLab{Lam-Ascribe} rule allows an argument to be used polymorphically
inside the body of a lambda abstraction.
The \freezemlLab{App} rule is standard.

\paragraph{Let Bindings}

Because we adopt the value restriction, the \freezemlLab{Let} rule behaves
differently depending on whether or not $M$ is a guarded value (cf. $\dec{GVal}$
syntactic category in Figure~\ref{fig:freezeml-syntax}).
The choice of whether to generalise the type of $M$ is delegated to the
judgement $(\Delta, \Delta'', M, A') \Updownarrow A$, where $A'$ is the type of
$M$ and $\Delta''$ are the generalisable type variables of $M$, i.e. $\Delta''
= \ftv(A') - \Delta$.
The $\Updownarrow$ judgement determines $A$, the type given to $x$ while
type-checking $N$.
If $M$ is a guarded value, we generalise and have $A = \forall \Delta''.A'$.
If $M$ is not a guarded value, we have $A = \rsubst(A')$, where $\rsubst$ is an
instantiation with $\Delta \vdash \rsubst : \Delta'' \Rightarrow_\mono \cdot$.
This means that instead of abstracting over the unbound type variables
$\Delta''$ of $A'$, we instantiate them \emph{monomorphically}. We further
discuss the need for this behaviour in \cref{sec:design-considerations}.

The $\dec{gen}$ judgement used in the \freezemlLab{Let} rule may seem surprising
--- its first component is unused whilst the second component is identical in
both cases and corresponds to the generalisable type variables of $A'$.
Indeed, the first component of $\dec{gen}$ is irrelevant for typing but it is
convenient for writing the translation from \freezeml to System F
(Figure~\ref{fig:trans-freezeml-to-f} in Section~\ref{sec:freezeml-to-f}), where
it is used to form a type abstraction, and in the type inference algorithm
(Figure~\ref{fig:inference} in Section~\ref{sec:type-inference}), where it
allows us to collapse two cases into one.

The \freezemlLab{Let} rule requires that $A'$ is the principal type for $M$.
This constraint is necessary to ensure completeness of our type inference
algorithm;
we discuss it further in \cref{part:principal-type-restriction}.
The relation $\meta{principal}$ is defined in
Figure~\ref{fig:freezeml-aux-judgments}.


The \freezemlLab{Let-Ascribe} rule is similar to the \freezemlLab{Let} rule, but instead
of generalising the type of $M$, it uses the type $A$ supplied via an annotation.
As in \freezemlLab{Let}, $A'$ denotes the type of $M$.
However, the annotated case admits non-principal types for $M$.
The $\msplit$ operator enforces the value restriction.  If $M$ is a guarded
value, $A'$ must be a guarded type, i.e. we have $A' = H$ for some $H$.  We then
have $A = \forall \Delta'. H$.  If $M$ is not a guarded value $\msplit$ requires
$A' = A$ and $\Delta' = \cdot$.  This means that \emph{all} toplevel quantifiers
in $A$ must originate from $M$ itself, rather than from generalising it.

\medskip
Every valid typing judgement in ML is also a valid typing judgement in \freezeml.
\begin{restatable}{thm}{freezemlconvservative}
If $\typ{\Delta; \Gamma}{M : S}$ in ML then $\typ{\Delta; \Gamma}{M : S}$ in \freezeml.
\end{restatable}
(The exact derivation can differ due to differences in the kinding rules and the
principality constraint on the \freezemlLab{Let} rule.)

\subsection{Design Considerations}
\label{sec:design-considerations}

\paragraph{Monomorphic instantiation in the \freezemlLab{Let} rule}

Recall that the \freezemlLab{Let} rule enforces the value restriction by
instantiating those type variables that would otherwise be quantified over.
Requiring these type variables to be instantiated with monotypes allows us
to avoid problems similar to the ones outlined in
Section~\ref{sec:monomorphism-restriction}.
Consider the following two functions.
\[
\bl
\dec{bad3} = \lambda (\var{bot} \!:\! \forall a.a).
               \Let\; f = \var{bot}\,\var{bot} \;\In\; (\dec{poly}\,\freeze{f}, (f\,42)+1) \\
\dec{bad4} = \lambda (\var{bot} \!:\! \forall a.a).
               \Let\; f = \var{bot}\,\var{bot} \;\In\; ((f\,42)+1, \dec{poly}\,\freeze{f}) \\
\el
\]
Since we do not generalise non-values in let-bindings due to the value
restriction,
in both of these examples
$f$ is initially assigned the type $a$ rather than the most general type
$\forall a.a$ (because $\var{bot}\,\var{bot}$ is a non-value).
Assuming type inference proceeds from left to right then type inference will
succeed on $\dec{bad3}$ and fail on $\dec{bad4}$ for the same reasons as in
Section~\ref{sec:monomorphism-restriction}.
In order to rule out this class of examples, we insist that non-values are first
generalised and then instantiated with monomorphic types. Thus we constrain $a$
to only unify with monomorphic types, which leads to type inference failing on
both $\dec{bad3}$ and $\dec{bad4}$.

Our guiding principle is ``never guess polymorphism''.
While our system permits instantiation of quantifiers with polymorphic types --
per \freezemlLab{Var} rule -- it does not permit polymorphic instantiations of
type variables inside the type environment.
The high-level invariant
that \freezeml uses to ensure that this principle is not violated is that any (as
yet) unknown types appearing in the type environment (which maps term variables
to their currently inferred types) during type inference must be explicitly
marked as monomorphic.
The only means by which inference can introduce unknown types into the type
environment are through unannotated lambda-binders or through not generalising
let-bound variables. By restricting these cases to be monomorphic we ensure in
turn that any unknown type appearing in the type environment must be explicitly
marked as monomorphic.
%

%
%

\paragraph{Principal Type Restriction}
\label{part:principal-type-restriction}

The \freezemlLab{Let} rule requires that when typing $\Let\; x = M \;\In\; N$, the
type $A'$ given to $M$ must be principal.
Consider the program
\[
\dec{bad5} = \Let\; f = \lambda x.x \;\In\; \freeze{f}\,42
\]
On the one hand, if we infer the type $\forall a.a \to a$ for $f$, then
$\dec{bad5}$ will fail to type check as we cannot apply a term of polymorphic
type (instantiation is only automatic for variables).
However, given a traditional declarative type system one might reasonably
propose $\Int \to \Int$ as a type for $f$, in which case $\dec{bad5}$ would be
typeable --- albeit a conventional type inference algorithm would have
difficulty inferring a type for it.
In order to ensure completeness of our type inference algorithm in the presence
of generalisation and freeze, we
bake principality into the typing rule for let, similarly to
  \cite{GarrigueR99,VytiniotisWJ06,Leijen08,DynamicsML}.
This means that the only legitimate type that $f$ may be assigned is the most
general one, that is $\forall a.a \to a$.

One may think of side-stepping the problem with $\dec{bad5}$ by always
instantiating terms that appear in application position (after all, it is
always a type error for an uninstantiated term of polymorphic type to appear in
application position). But then we can exhibit the same problem with a slightly
more intricate example.
\[
\dec{bad6} = \Let\; f = \lambda x.x \;\In\; \id\,\freeze{f}\,42
\]

The principality condition is also applied in the non\-/generalising case of
the \freezemlLab{Let} rule, meaning that we must instantiate the principal type
for $M$ rather than an arbitrary one.
Otherwise, we could still type $\dec{bad4}$ by assigning $\var{bot}\,\var{bot}$
type $\forall a. a \to a$. In the \freezemlLab{Let} rule $\Delta'$ would be
empty, making instantiation a no-op.


\paragraph{Well-foundedness}
The alert reader may already have noticed a complication resulting from the
principal type restriction: $\meta{principal}(\Delta,\Gamma,M,\Delta',A')$
contains a negative occurrences of the typing relation, in order to express that
$\Delta',A'$ is a ``most general'' solution for $\Delta'',A''$ among all
possible derivations of $\Delta,\Delta'';\Gamma \vdash M : A''$.  This negative
occurrence means that \emph{a priori}, the rules in
Figures~\ref{fig:freezeml-typing} and~\ref{fig:freezeml-aux-judgments} do not
form a proper inductive definition.

This is a potentially serious problem, but it can be resolved easily by
observing that the rules, while not syntactically well-founded, can be
stratified.  Instead of considering the rules in
Figures~\ref{fig:freezeml-typing} and~\ref{fig:freezeml-aux-judgments} as a
\emph{single} inductive definition, we consider them to determine a  \emph{function} $\J{-}$ from
terms $M$ to triples $(\Delta,\Gamma,A)$.
  The typing relation is then defined as $\typ{\Delta;\Gamma}{M:A} \iff
  (\Delta,\Gamma,A) \in \J{M}$.  We can easily prove
by induction on $M$ that $\J{M}$ is well-defined.
Furthermore, we can show that the inference rules in
Figure~\ref{fig:freezeml-typing} hold and are invertible.  When reasoning about
typing judgements, we can proceed by induction on $M$ and use inversion.  It is
also sound to
perform recursion over typing derivations provided the
$\meta{principal}$ assumption is not needed; we indicate this by greying out
this assumption (for example in Figure~\ref{fig:trans-freezeml-to-f}).   We give full details
and explain how this reasoning is performed in \auxref{sec:well-founded-freezeml}.

\paragraph{Type Variable Scoping}

A type annotation in \freezeml may contain type variables that is not bound by
the annotation.
In contrast to many other systems, we do not interpret such variables
existentially, but allow binding type variables across different annotations.
In an expression $\Let\; (x : A) = M \;\In\; N$, we therefore consider
the toplevel quantifiers of $A$ bound in $M$, meaning that they can be used
freely in annotations inside $M$,
rather like GHC's scoped type variables~\citep{PeytonjonesS02},
However, this is only true for the generalising case, when $M$ is a guarded
value.  In the absence of generalisation, any polymorphism in the type $A$
originates from $M$ directly (e.g., because $M$ is a frozen variable).
Hence, if $M$ is not a guarded value no bound variables of $A$ are bound in $M$.

Note that given the $\Let$ binding above, where $A$ has the shape $\forall
\Delta. H$, there is no ambiguity regarding which of the type variables in $\Delta$
result from generalisation and which originate from $M$ itself.
If $M$ is a guarded value, its type is guarded, too, and hence all variables in $\Delta$
result from generalisation.
Conversely, if $M \not \in \dec{GVal}$, then there is no generalisation at all.

Due to the unambiguity of the binding behaviour in our system with the value
restriction, we can define a purely syntax-directed well-formedness judgement
for verifying that types in annotations are well-kinded and respect the intended
scoping of type-variables.
We call this property well-scopedness, and it is a prerequisite for type
inference.
The corresponding judgement is
$\termwf{\Delta}{M}$, checking that
in $M$, the type annotations are well-formed with respect to kind environment
$\Delta$ (Figure~\ref{fig:terms-well-scopedness}).  The main subtlety in this judgement is in how $\Delta$ grows when we
encounter \emph{annotated} let-bindings.  For annotated lambdas, we just check
that the type annotation is well-formed in $\Delta$ but do not add any type
variables in $\Delta$.  For plain let, we just check well-scopedness
recursively.  However, for annotated let-bindings, we check that the type
annotation $A$ is well-formed, and we check that $M$ is well-scoped \emph{after
  extending $\Delta$ with the top-level type variables of $A$}.  This is
sensible because in the \textsc{Let-Ascribe} rule, these type variables (present
in the type annotation) are introduced into the kind environment when type
checking $M$.  In an unannotated let, in contrast, the generalisable type
variables are not mentioned in $M$, so it does not make sense to allow them to
be used in other type annotations inside $M$.

As a concrete example of how this
works, consider an explicitly annotated let-binding of the identity function:
$\Let\; (f: \forall a. a \to a) = \lambda (x: a). x \; \In \; N$, where the $a$
type annotation on $x$ is bound by $\forall a$ in the type annotation on $f$.
However, if we left off the $\forall a. a \to a$ annotation on $f$, then the $a$
annotation on $x$ would be unbound.
This also means that in expressions, we cannot let type annotations
$\alpha$-vary freely; that is, the previous expression is $\alpha$-equivalent to
$\Let\; (f: \forall b. b \to b) = \lambda (x: b). x \; \In\; N$ but not to $\Let\; (f:
\forall b. b \to b) = \lambda (x: a). x \; \In\; N$. This behaviour is similar to other
proposals for scoped type variables~\citep{PeytonjonesS02}.

\paragraph{``Pure'' FreezeML}
In a hypothetical version of \freezeml without the value restriction, a purely
syntactic check on $\Let\; (x : A) = M \;\In\; N$ is not sufficient to determine
which top-level quantifiers of $A$ are bound in $M$.
In the expression
\[\ba{l}
\Let\; (f : \forall a \, b. a \to b \to b ) = \\
\quad   \Let\; (g : \forall b. a \to b \to b) =  \lambda y \, z. z \;\In\; \dec{id} \: \freeze{g} \\
\In\; N
\ea\]
the outer $\Let$ generalises $a$, unlike the subsequent variable $b$, which
arises from the inner $\Let$ binding.
The well-scopedness judgement would require typing information.
Moreover, the \freezemlLab{Let-Asc} rule would have to nondeterministically
split the type annotation $A$ into $\forall \Delta', \Delta''.H$, such that
$\Delta'$ contains those variables to generalise ($a$ in the example), and
$\Delta''$ contains those type variables originating from $M$ directly ($b$ in
the example).
Similarly, type inference would have to take this splitting into account.

\input{figures/freezeml_well_scopedness.tex}

\paragraph{Instantiation strategies}
In \freezeml (and indeed ML) the only terms that are implicitly instantiated are
variables.
%
%
Thus $(\dec{head}\ \dec{ids})\ 42$ is ill-typed and we must insert
the instantiation operator $\inst$ to yield a type-correct expression:
$(\dec{head}\ \dec{ids})\inst\ 42$.
It is possible to extend our approach to perform \emph{eliminator
instantiation}, whereby we implicitly instantiate terms appearing in monomorphic
elimination position (in particular application position), and thus, for
instance, infer a type for $\dec{bad5}$ without compromising completeness.

Another possibility is to instantiate all terms, except those that are
explicitly frozen or generalised. Here, it also makes sense to extend the
$\freeze{-}$ operator to act on arbitrary terms, rather than just variables. We
call this strategy \emph{pervasive instantiation}. Like eliminator
instantiation, pervasive instantiation infers a type for
$(\dec{head}\ \dec{ids})\ 42$. However, pervasive instantiation requires
inserting explicit generalisation where it was previously unnecessary. Moreover,
pervasive instantiation complicates the meta-theory, requiring two mutually
recursive typing judgements instead of just one.

The formalism developed in this paper uses variable instantiation alone, but our
implementation also supports eliminator instantiation. We defer further
theoretical investigation of alternative strategies to future work.

\section{Relating System F and \freezeml}\label{sec:translations}

In this section we present type-preserving translations mapping System F terms
to \freezeml terms and vice versa.
%
%
We also briefly discuss the equational theory induced on \freezeml by these
translations.

\subsection{From System F to \freezeml}
\label{sec:f-to-freezeml}

\begin{figure}[tb]
\[
\ba{@{}r@{~\;}c@{~\;}l@{\quad}l@{}}
  \transftofreezeml{x}                  &=& \freeze{x} \\
  \transftofreezeml{\lambda x^A. M}     &=& \lambda (x:A). \transftofreezeml{M} \\
  \transftofreezeml{M\; N}              &=& \transftofreezeml{M} \;\: \transftofreezeml{N} \\
  \transftofreezeml{\Lambda a.V^B}   &=&
    \Let\; (x : \forall a. B) = (\transftofreezeml{V})@ \; \In \; \freeze{x}\\
  \transftofreezeml{M^{\forall a.B} \: A}  &=&
    \Let\; (x : B[A/a]) = (\transftofreezeml{M})@ \;\In\; \freeze{x} \\
\ea\]
\caption{Translation from System F to \freezeml}
\label{fig:trans-f-to-freezeml}
\end{figure}

\input{figures/translation_freezeml_to_systemf.tex}


Figure~\ref{fig:trans-f-to-freezeml} defines a translation
$\transftofreezeml{-}$ of System F terms into \freezeml.
The translation depends on types of subterms and is thus formally defined on
derivations, but we use a shorthand notation in which subterms are annotated
with their type (e.g., in $\Lambda a.V^B$, $B$ indicates the type of $V$).

Variables are frozen to suppress instantiation. Term abstraction and application
are translated homomorphically.

Type abstraction $\Lambda a. V$ is translated using an annotated let-binding to
perform the necessary generalisation.
However, we cannot bind $x$ to the translation of $V$ directly as only
\emph{guarded} values may be generalised but $\transftofreezeml{V}$ may be
an unguarded value (concretely, a frozen variable).
Hence, we bind $x$ to $(\transftofreezeml{V})@$, which is syntactic sugar for
$\Let\; y = \transftofreezeml{V} \;\In\; y$.
This expression is indeed a guarded value.
We then freeze $x$ to prevent immediate instantiation.
Type application $M \: A$, where $M$ has type $\forall a.B$, is translated
similarly to type abstraction. We bind $x$ to the result of translating $M$, but
only after instantiating it. The variable $x$ is annotated with the intended
return type $B[A/a]$ and returned frozen.

Explicit instantiation is strictly necessary and the following, seemingly easier
translation is incorrect.
\[
  \transftofreezeml{M^{\forall a.B} \: A} \quad \neq \quad
    \Let\; (x : B[A/a]) = \transftofreezeml{M} \;\In\; \freeze{x} \\
\]
The term $\transftofreezeml{M}$ may be a frozen variable or an application,
whose type cannot be implicitly instantiated to type $B[A/a]$.

For any System F value $V$ (i.e., any term other than an application),
$\transftofreezeml{V}$ yields a \freezeml value
(Figure~\ref{fig:freezeml-syntax}).

Each translated term has the same type as the original.
\begin{restatable}[Type preservation]{thm}{ftofreezemltypepreservation}
\label{theorem:f-to-freezeml}
If $\typ{\Delta;\Gamma}{M : A}$ in System F then
$\typ{\Delta;\Gamma}{\transftofreezeml{M} : A}$ in \freezeml.
\end{restatable}

\subsection{From \freezeml to System F}
\label{sec:freezeml-to-f}

Figure~\ref{fig:trans-freezeml-to-f} gives the translation of \freezeml to
System F.
The translation depends on types of subterms and is thus formally defined on
derivations.
%
%
Frozen variables in \freezeml are simply variables in System F.
A plain (i.e., not frozen) variable $x$ is translated to a type application $x
\: \rsubst(\Delta')$, where $\rsubst(\Delta')$ stands for the pointwise application of
$\rsubst$ to $\Delta'$.
Here, $\rsubst$ and $\Delta'$ are obtained from $x$'s type derivation in
\freezeml; $\Delta'$ contains all top-level quantifiers of $x's$ type.
This makes \freezeml's implicit instantiation of non-frozen variables explicit.
Lambda abstractions and applications translate directly.
Let-bindings in \freezeml are translated as generalised let-bindings in System F
where $\Let\, x^A = M \,\In\, N$ is syntactic sugar for $(\lambda x^A.N)\,M$.
Here, generalisation is repeated type abstraction.

Each translated term has the same type as the original.
\begin{restatable}[Type preservation]{thm}{freezemltoftypepreservation}
\label{theorem:freezeml-to-f}
If $\typ{\Delta; \Gamma}{M : A}$ holds in \freezeml
then $\typ{\Delta; \Gamma}{ \transfreezemlfof{M} : A}$ holds in System F.
\end{restatable}

\subsection{Equational reasoning}
\label{sec:reasoning}

%
We can derive and verify equational reasoning principles for \freezeml by
lifting from System F via the translations.
We write $M \conv N$ to mean $M$ is observationally equivalent to $N$ whenever
$\typ{\Delta; \Gamma}{M : A}$ and $\typ{\Delta; \Gamma}{N : A}$.
%
%
At a minimum we expect $\beta$-rules to hold, and indeed they do; the twist is
that they involve substituting a different value depending on whether the
variable being substituted for is frozen or not.
\[
\ba{@{}l@{~}c@{~}l@{~}l@{~}l@{}}
\Let\; x = V \;\In\; N                &\conv& N[\gen V     &/~ \freeze{x},~ (\gen V)\inst &/~ x] \\
\Let\; (x : A) = V \;\In\; N          &\conv& N[\gen^A V   &/~ \freeze{x},~ (\gen^A V)\inst &/~ x] \\
(\lambda x.M)\,V                      &\conv& M[V          &/~ \freeze{x},~ \vinst      &/~ x] \\
(\lambda (x : A).M)\,V                &\conv& M[V          &/~ \freeze{x},~ \vinst &/~ x] \\
\ea
\]
If we perform type-erasure then these rules degenerate to the standard ones.
%
%
%
We can also verify that $\eta$-rules hold.
%
\[
\ba{@{}l@{~}c@{~}l@{}}
\Let\; x = U \;\In\; x                &\conv& U \\
\Let\; (x : A) = U \;\In\; x          &\conv& U \\
\lambda x.M\,x                        &\conv& M \\
\ea\qquad
\ba{@{}l@{~}c@{~}l@{}}
\Let\; x = \freeze{y} \;\In\; x       &\conv& y \\
\Let\; (x : A) = \freeze{y} \;\In\; x &\conv& y \\
\lambda (x : A).M\,\freeze{x}         &\conv& M \\
\ea
\]

\section{Type Inference}
\label{sec:inference}

In this section we present a sound and complete type inference algorithm for
\freezeml. The style of presentation is modelled on that of \citet{Leijen08}.

\subsection{Type Variables and Kinds}
\label{sec:freezeml-tyvars-and-kinds}

When expressing type inference algorithms involving first-class polymorphism, it
is crucial to distinguish between object language type variables, and meta
language type variables that stand for unknown types required to solve the type
inference problem. This distinction is the same as that between
\emph{eigenvariables} and \emph{logic variables} in higher-order logic
programming~\cite{Miller92}. We refer to the former as \emph{rigid} type
variables and the latter as \emph{flexible} type variables.
For the purposes of the algorithm we will explicitly separate the two by placing
them in different kind environments.

As in the rest of the paper, we let $\Delta$ range over fixed kind environments
in which every type variable is monomorphic (kind $\mono$).
In order to support, for instance, applying a function to a polymorphic
argument, we require flexible variables that may be unified with polymorphic
types.
For this purpose we introduce refined kind environments ranged over by $\Theta$.
Type variables in a refined kind environment may be polymorphic (kind $\poly$)
or monomorphic (kind $\mono$).
In our algorithms we place rigid type variables in a fixed environment $\Delta$
and flexible type variables in a refined environment $\Theta$.
Refined kind environments ($\Theta$) are given by the following grammar.
\begin{syntax}
 & \dec{KEnv} \ni \Theta &::=& \cdot \mid \Theta, a : K \\
\end{syntax}
%
%
We often implicitly treat fixed kind environments $\many{a}$ as refined kind
environments $\many{a : \mono}$.
The refined kinding rules are given in Figure~\ref{fig:kinding}.
%







\begin{figure}
\raggedright
$\boxed{\Theta \vdash A : K}$
\vspace{-0.6cm}
\begin{mathpar}
  \inferrule*[Lab=TyVar]
    {a:K \in \Theta}
    {\Theta \vdash a : K}
\hspace{7pt}
  \inferrule*[Lab=Cons]
    {\arity(D) = n \\\\
     \Theta \vdash A_1 : K \\\\ \cdots \\\\ \Theta \vdash A_n : K}
    {\Theta \vdash \tc\,\many{A} : K}
\hspace{7pt}
  \inferrule*[Lab=ForAll]
    {\Theta, a:\mono \vdash A : \poly}
    {\Theta \vdash \forall a.A : \poly}
\hspace{7pt}
  \inferrule*[Lab=Upcast]
    {\Theta \vdash A : \mono}
    {\Theta \vdash A : \poly}
\end{mathpar}

$\boxed{\Theta \vdash \Gamma}$
\vspace{-0.6cm}
\begin{mathpar}
\inferrule*[Lab=Empty]
  { }
  {\Theta \vdash \cdot}

\inferrule*[Lab=Extend]
  {\Theta \vdash \Gamma \\ \Theta \vdash A : \poly \\\\ (\text{for all } a \in \ftv(A) \mid a : \mono \in \Theta )}
  {\Theta \vdash \Gamma, x : A}
\end{mathpar}
\caption{Refined Kinding Rules}
\label{fig:kinding}
\end{figure}

The key difference with respect to the object language kinding rules is that
type variables can now be polymorphic.
Rather than simply defining kinding of type environments point-wise the
\textsc{Extend} rule additionally ensures that all type variables appearing in a
type environment are monomorphic.
This restriction is crucial for avoiding guessing of polymorphism. More
importantly, it is also key to ensuring that typing judgements are stable under
substitution.  Without it it would be possible to substitute monomorphic type
variables with types containing nested polymorphic variables, thus introducing
polymorphism into a monomorphic type.

We generalise typing judgements $\typ{\Delta;\Gamma}{M : A}$ to
$\typ{\Theta;\Gamma}{M : A}$, adopting the convention that
$\wfctx{\Theta}{\Gamma}$ and $\Theta \vdash A$ must hold as preconditions.

\subsection{Type Substitutions}

In order to define the type inference algorithm we will find it useful to define
a judgement for type substitutions $\fsubst$, which operate on flexible type
variables, unlike type instantiations $\rsubst$, which operate on rigid type
variables.
The type substitution rules are given in
Figure~\ref{fig:freezeml-type-subst}.
The rules are as in Figure~\ref{fig:freezeml-typing}, except that the kind
environments on the right of the turnstile are refined kind environments and
rather than the substitution having a fixed kind, the kind of each type variable
must match up with the kind of the type it binds.
\begin{figure}
\raggedright
$\boxed{\Delta \vdash \fsubst : \Theta \Rightarrow \Theta'}$
\vspace{-0.4cm}
\begin{mathpar}
\inferrule
  { }
  {\Delta \vdash \emptyset : \cdot \Rightarrow \Theta}

\inferrule
  {\Delta \vdash \fsubst : \Theta' \Rightarrow \Theta \\
   \Delta, \Theta \vdash A : K}
  {\Delta \vdash \fsubst[a \mapsto A] : (\Theta', a : K) \Rightarrow \Theta}
\end{mathpar}
\caption{Type Substitutions}
\label{fig:freezeml-type-subst}
\end{figure}

We write $\idsubst_\Theta$ for the identity type substitution on $\Theta$,
omitting the subscript when clear from context.
\begin{equations}
\idsubst_\cdot       = \emptyset  \qquad \idsubst_{\Theta, a:K} = \idsubst_{\Theta}[a \mapsto a]
\end{equations}
%
%
%
%
%

Composition of type substitutions is standard.
\begin{equations}
\fsubst \circ \emptyset             = \emptyset \qquad
\fsubst \circ \fsubst'[a \mapsto A] = (\fsubst \circ \fsubst')[a \mapsto \fsubst(A)]
\end{equations}
%
%
%

The rules shown in~Figure~\ref{fig:subst-prop} are admissible and we make use of them freely in our
algorithms and proofs.
\begin{figure}[t]
\begin{mathpar}
\inferrule*[Lab=S-Identity]
  { }
  {\Delta \vdash \idsubst_\Theta : \Theta \Rightarrow \Theta}

\inferrule*[Lab=S-Weaken]
  {\Delta \vdash \fsubst : \Theta \Rightarrow \Theta'}
  {\Delta, \Delta' \vdash \fsubst : \Theta \Rightarrow \Theta', \Theta''}
\\
\inferrule*[Lab=S-Compose]
  {\Delta \vdash \fsubst : \Theta' \Rightarrow \Theta'' \\\\ \Delta \vdash \fsubst' : \Theta \Rightarrow \Theta'}
  {\Delta \vdash \fsubst \circ \fsubst' : \Theta \Rightarrow \Theta''}

\inferrule*[Lab=S-Strengthen]
  {\Delta \vdash \fsubst : \Theta \Rightarrow \Theta' \\\\ \ftv(\fsubst) \mathbin{\#} \Delta', \Theta''}
  {\Delta - \Delta' \vdash \fsubst : \Theta \Rightarrow \Theta' - \Theta''}
\end{mathpar}
\caption{Properties of Substitution}
\label{fig:subst-prop}
\end{figure}

\subsection{Unification}
\label{sec:unification}

\begin{figure}
\raggedright

\[
\bl
\unify : (\meta{PEnv} \times \meta{KEnv} \times \meta{Type} \times \meta{Type})
           \rightharpoonup (\meta{KEnv} \times \meta{Subst}) \medskip \\

\unify(\Delta, \Theta, a, a ) = \\
\quad\mreturn\; (\Theta, \idsubst) \medskip \\

\unify(\Delta, (\Theta, a : K), a, A) = \\
\quad\bl
     \mlet\;\Theta_1 = \dec{demote}(K, \Theta, \ftv(A) - \Delta) \\
     \massert\;\Delta, \Theta_1 \vdash A : K \\
     \mreturn\; (\Theta_1, \idsubst[a \mapsto A]) \\
     \el
     \medskip \\

\unify(\Delta, (\Theta, a : K), A, a) = \\
\quad\bl
     \mlet\;\Theta_1 = \dec{demote}(K, \Theta, \ftv(A) - \Delta) \\
     \massert\;\Delta, \Theta_1 \vdash A : K \\
     \mreturn\; (\Theta_1, \idsubst[a \mapsto A]) \\
     \el
     \medskip \\

\unify(\Delta, \Theta, D\,\many{A}, D\,\many{B}) = \\
\quad\bl
     \mlet\; (\Theta_1, \fsubst_1) = (\Theta, \idsubst) \\
     \mlet\; n = \dec{arity}(D) \\
     \text{for } i \in 1...n \\
     \quad \mlet\; (\Theta_{i+1}, \fsubst_{i+1}) = \\
     \qquad\bl
           \mlet\; (\Theta', \fsubst') = \unify(\Delta, \Theta_i, \fsubst_i(A_i), \fsubst_i(B_i)) \\
           \mreturn\; (\Theta', \fsubst' \circ \fsubst_i) \\
           \el \\
     \mreturn\; (\Theta_{n+1}, \fsubst_{n+1}) \\
     \el \medskip \\

\unify(\Delta, \Theta, \forall a. A, \forall b. B)= \\
\quad\bl
     \mfresh\; c \\
     \mlet\; (\Theta_1, \fsubst') = \unify((\Delta, c), \Theta, A[c/a], B[c/b]) \\
     \massert\;c \notin \ftv(\fsubst') \\
     \mreturn\; (\Theta_1, \fsubst') \\
     \el \medskip \\
\el
\]

\begin{equations}
\dec{demote}(\poly, \Theta, \Delta)      &=& \Theta \\
\dec{demote}(\mono, \cdot, \Delta)      &=& \cdot \\
\dec{demote}(\mono, (\Theta,a:K), \Delta)       &=&
\dec{demote}(\mono,\Theta,\Delta),a:\mono  \quad (a \in \Delta)\\
\dec{demote}(\mono, (\Theta,a:K), \Delta) &=&
\dec{demote}(\mono, \Theta,\Delta),a:K \quad (a \not\in \Delta)  \\
\end{equations}

\caption{Unification Algorithm}
\label{fig:unification}
\end{figure}

A crucial ingredient for type inference is unification.
The unification algorithm is defined in Figure~\ref{fig:unification}. It is
partial in that it either returns a result or fails. Following \citet{Leijen08}
we explicitly indicate the successful return of a result $X$ by writing
$\mreturn\;X$. Failure may be either explicit or implicit (in the case that an
auxiliary function is undefined). The algorithm takes a quadruple ($\Delta,
\Theta, A, B$) of a fixed kind environment $\Delta$, a refined kind environment
$\Theta$, and types $A$ and $B$, such that $\Delta, \Theta \vdash A, B$.
It returns a unifier, that is, a pair $(\Theta', \fsubst)$ of a new refined kind
environment $\Theta'$ and a type substitution $\fsubst$, such that $\Delta
\vdash \fsubst : \Theta \Rightarrow \Theta'$.

A type variable unifies with itself, yielding the identity substitution.
Due to the use of explicit kind environments, there is no need for an explicit
occurs check to avoid unification of a type variable $a$ with a type $A$
including recursive occurrences of $a$.
Unification of a flexible variable $a$ with a type $A$ implicitly performs an
occurs check by checking that the type substituted for $a$ is well-formed in an
environment ($\Delta, \Theta_1$) that does not contain $a$.
A polymorphic flexible variable unifies with any other type, as is standard.
A monomorphic flexible variable only unifies with a type $A$ if $A$ may be
\emph{demoted} to a monomorphic type. The auxiliary $\dec{demote}$ function converts
any polymorphic flexible variables in $A$ to monomorphic flexible variables in
the refined kind environment. This demotion is sufficient to ensure that further
unification cannot subsequently make $A$ polymorphic.
Unification of data types is standard, checking that the data type constructors
match, and recursing on the substructures.
Following \citet{Leijen08}, unification of quantified types ensures that
forall-bound type variables do not escape their scope by introducing a fresh
rigid (skolem) variable and ensuring it does not appear in the free type
variables of the substitution.

\begin{restatable}[Unification is sound]{thm}{thmunifysound}
\label{thm:unification-sound}
If $\Delta, \Theta \vdash A, B : K$ and $\unify(\Delta, \Theta, A, B) = (\Theta',
\fsubst)$ then $\fsubst(A) = \fsubst(B)$ and $\Delta \vdash \fsubst : \Theta
\Rightarrow \Theta'$.
\end{restatable}

\begin{restatable}[Unification is complete and most general]{thm}{thmunifycomplete}
\label{thm:thmunifycomplete}
If $\Delta \vdash \fsubst : \Theta \Rightarrow \Theta'$ and $\Delta, \Theta
\vdash A : K$ and $\Delta, \Theta \vdash B : K$ and $\fsubst(A) = \fsubst(B)$, then
$\unify(\Delta, \Theta, A, B) = (\Theta'', \fsubst')$ where there exists
$\fsubst''$ satisfying $\Delta \vdash \fsubst'' : \Theta'' \Rightarrow \Theta'$
such that $\fsubst = \fsubst'' \circ \fsubst'$.
\end{restatable}

\input{figures/freezeml_inference_algorithm.tex}

\subsection{The Inference Algorithm}
\label{sec:type-inference}

The type inference algorithm is defined in Figure~\ref{fig:inference}. It is
partial in that it either returns a result or fails. The algorithm takes a
quadruple ($\Delta, \Theta, \Gamma, M)$ of a fixed kind environment $\Delta$, a
refined kind environment $\Theta$, a type environment $\Gamma$, and a term $M$,
such that $\Delta; \Theta \vdash \Gamma$.
If successful, it returns a triple $(\Theta', \fsubst, A)$ of a new refined kind
environment $\Theta'$, a type substitution $\fsubst$, such that $\Delta \vdash
\fsubst : \Theta \Rightarrow \Theta'$, and a type $A$ such that $\Delta, \Theta'
\vdash A: \poly$.

The algorithm is an extension of algorithm W~\citep{DamasM82} adapted to use
explicit kind environments $\Delta,\Theta$.
Inferring the type of a frozen variable is just a matter of looking up its type
in the type environment.
As usual, the type of a plain (unfrozen) variable is inferred by instantiating
any polymorphism with fresh type variables. The returned identity
type substitution is weakened accordingly.
Crucially, the argument type inferred for an unannotated lambda abstraction is
monomorphic.
If on the other hand the argument type is annotated with a type, then we just
use that type directly.
For applications we use the unification algorithm to check that the function and
argument match up.
Generalisation is performed for unannotated let-bindings in which the
let-binding is a guarded value.
For unannotated let-bindings in which the let-binding is not a guarded value,
generalisation is suppressed and any ungeneralised flexible type variables are
demoted to be monomorphic.
When a let-binding is annotated with a type then rather than performing
generalisation we use the annotation, taking care to account for any
polymorphism that is already present in the inferred type for $M$ using
$\dec{split}$, and checking that none of the quantifiers escape by inspecting
the codomain of $\fsubst_2$.

\begin{restatable}[Type inference is sound]{thm}{thminferencesound}
\label{thm:inference-sound}
If $\Delta, \Theta \vdash \Gamma$ and $\termwf{\Delta}{M}$ and $\Infer(\Delta,
\Theta, \Gamma, M) = (\Theta', \fsubst, A)$ then $\typ{\Delta, \Theta';
\fsubst(\Gamma)}{M : A}$ and $\Delta \vdash \fsubst : \Theta \Rightarrow \Theta'$.
\end{restatable}


%
\begin{restatable}[Type inference is complete and principal]{thm}{thminferencecomplete}
\label{thm:inference-completeness-mg}
Let $\termwf{\Delta}{M}$ and $\wfctx{\Delta, \Theta}{\Gamma}$.
If
  $\Delta \vdash \fsubst : \Theta \Rightarrow \Theta'$ and $\typ{\Delta,
  \Theta'; \fsubst(\Gamma)}{M : A}$, then $\Infer(\Delta, \Theta, \Gamma, M)
\allowbreak = (\Theta'', \fsubst', A')$ where there exists $\fsubst''$
satisfying $\Delta \vdash \fsubst'' : \Theta'' \Rightarrow \Theta'$ such that
$\fsubst = \fsubst'' \circ \fsubst'$
and $\fsubst''( A') = A$.
\end{restatable}

\section{Implementation}
\label{sec:implementation}

We have implemented FreezeML as an extension of \sysname\pseudonym.  This
exercise was mostly routine.
In the process we addressed several practical concerns and encountered some
non-trivial interactions with other features of \sysname.
In order to keep this paper self-contained we avoid concrete \sysname syntax, but
instead illustrate the ideas of the implementation in terms of extensions to the
core syntax used in the paper.

In ASCII we render $\freeze{x}$ as \verb|~|$x$.
For convenience, \sysname builds in the generalisation $\gen$ and instantiation
operators $\inst$.

In practice (in \sysname and other functional languages), it is often convenient to
include a type signature for a function definition rather than annotations on
arguments.
%
%
Thus
\[
\bl
f : \forall a.A \to B \to C \\
f~x~y = M \\
N \\
\el
\]
is treated as:
\[
\Let\; (f : \forall a.A \to B \to C) = \lambda (x : A).\lambda (y : B).M \;\In\; N
\]
Though $x$ and $y$ are not themselves annotated, $A$ and $B$ may be polymorphic,
and may mention $a$.

Given that FreezeML is explicit about the order of quantifiers, adding support
for explicit type application~\citep{EisenbergWA16} is straightforward.
We have implemented this feature in \sysname.

\sysname has an implicit subkinding system used for various purposes including
classifying base types in order to support language-integrated
query~\citeanon{LindleyC12} and distinguishing between linear and non-linear
types in order to support session typing~\citeanon{LindleyM17}.
In plain FreezeML, if we have $\dec{poly} : (\forall a.a \to
a) \to \Int \times \Bool$ and $\dec{id} : \forall a.a \to a$, then we may write
$\dec{poly}~\freeze{\dec{id}}$.
The equivalent in \sysname also works.
However, the type inferred for the identity function in \sysname is not $\forall
a.a \to a$, but rather $\forall (a : \circ).a \to a$, where the subkinding
constraint $\circ$ captures the property that the argument is used linearly.
Given this more refined type for $\dec{id}$ the term
$\dec{poly}~\freeze{\dec{id}}$ no longer type-checks.
In this particular case one might imagine generating an implicit coercion (a
function that promises to use its argument linearly may be soundly treated as a
function that may or may not use its argument linearly).
In general one has to be careful to be explicit about the kinds of type
variables when working with first-class polymorphism.
Similar issues arise from the interaction between first-class polymorphism and
\sysname's effect type system~\citeanon{LindleyC12}.

Existing infrastructure for subkinding in the implementation of \sysname was
helpful for adding support for FreezeML as we exploit it for tracking the
monomorphism / polymorphism distinction.
However, there is a further subtlety: in FreezeML type variables of monomorphic
kind may be instantiated with (though not unified with) polymorphic types; this
behaviour differs from that of other kinds in \sysname.
%

The \sysname source language allows the programmer to explicitly distinguish
between rigid and flexible type variables. Flexible type variables can be
convenient to use as wild-cards during type inference.
As a result, type annotations in \sysname are slightly richer than those admitted by
the well-scopedness judgement of Figure~\ref{fig:terms-well-scopedness}.
It remains to verify the formal properties of the richer system.







\section{Related Work}
\label{sec:related-work}

There are many previous attempts to bridge the gap between ML and System F.
%
Some systems employ more expressive types than those of System F;
others implement heuristics in the type system to achieve a balance between
increased complexity of the system and reducing the number of necessary type
annotations;
finally, there are systems like ours that eschew such heuristics for the sake of
simplifying the type system further.
Users then have to state their intentions explicitly, potentially resulting in
more verbose programs.



%
%

\paragraph{Expressive Types}

MLF~\cite{BotlanR03} (sometimes stylised as \mlf) is considered to be the most
expressive of the conservative ML extensions so far.
%
  MLF achieves
its expressiveness by going beyond regular System F types and introducing
polymorphically bounded types, though translation from MLF to System F and vice
versa remains possible~\cite{BotlanR03, Leijen07}.
MLF also extends ML with
type annotations on lambda binders.
Annotations on binders that are
\emph{used} polymorphically are mandatory, since type inference will not guess
second-order types.
This is required to maintain principal types.


HML~\cite{Leijen09} is a simplification of MLF.  In HML all polymorphic function
arguments require annotations.  It 
significantly simplifies the type inference algorithm compared to MLF, though
polymorphically bounded types are still used.


\paragraph{Heuristics}

HMF~\cite{Leijen08} contrasts with the above systems in that it only uses
regular System F types (disregarding order of quantifiers).
Like FreezeML, it only allows principal types for let-bound variables, and type
annotations are needed on all polymorphic function parameters.
HMF allows both instantiation and generalisation in argument positions, taking
n-ary applications into account.
The system uses weights to select between less and more polymorphic types.
Whole lambda abstractions require an annotation to have a polymorphic return
type.
Such term annotations are \emph{rigid}, meaning they suppress instantiation and
generalisation.
As instantiation is implicit in HMF, rigid annotations can be seen as a means to
freeze arbitrary expressions.
%


Several systems for first-class polymorphism were proposed in the context of
the Haskell programming language.
These
systems include boxy types~\cite{VytiniotisWJ06}, FPH~\cite{VytiniotisWJ08}, and
GI~\cite{SerranoHVJ18}.
The Boxy Types system, used to implement GHC's
\texttt{ImpredicativeTypes} extension, was very fragile and thus difficult to
use in practice.
Similarly, the FPH system  -- based on MLF -- was simpler but still difficult to implement in practice.
GI is the latest development in this line of research.  Its key ingredient is a
heuristic that restricts polymorphic instantiation, based on whether a variable
occurs under a type constructor and argument types in an application.
Like HMF, it uses System F types, considers n-ary applications for typing, and
requires annotations both for polymorphic parameter and return types.
However, only top-level type variables may be re-ordered.
The authors show how to combine their system
with the OutsideIn(X)~\cite{VytiniotisJSS11} constraint-solving type inference
algorithm used by the Glasgow Haskell Compiler.
They also report a prototype
implementation of GI as an extension to GHC with encouraging experience porting
existing Hackage packages that use rank-$n$ polymorphism.

\paragraph{Explicitness}
Some early work on first-class polymorphism was based on the observation
that polymorphism can be encapsulated inside nominal
types~\cite{Remy94,LauferO94,OderskyL96,Jones97}.

The QML~\cite{RussoV09} system explicitly distinguishes between polymorphic
schemes and quantified types and hence does not use plain System F types.
%
%
Type schemes are used for ML let\-/polymorphism and introduced and eliminated
implicitly.
Quantified types are used for first-class polymorphism, in particular for polymorphic function arguments.
Such types must always be introduced and eliminated explicitly,
which requires stating the full type and not just instantiating the type variables.
%
%
All
polymorphic instantiations must therefore be made explicitly by annotating terms
at call sites.
%
%
%
Neither $\Let$- nor $\lambda$-bound variables can be annotated with a type.
%
%

%

Poly-ML~\cite{GarrigueR99} is similar to QML in that it distinguishes
two incompatible sorts of polymorphic types.
Type
schemes arise from standard ML generalisation; (boxed) polymorphic types are
introduced using a dedicated syntactic form which requires a type annotation.
Boxed polymorphic types are considered to be simple types, meaning that a
type variable can be instantiated with a boxed polymorphic type, but not with a
type scheme.
Terms of a boxed type are not instantiated implicitly, but must be opened
explicitly, resulting in instantiation.
Unlike QML, the instantiated type is deduced from the context, rather than requiring
an annotation.

%

Unlike FreezeML, Poly-ML supports inferring polymorphic parameter types for unannotated
lambdas, but this is limited to situations where the type is unambiguously
determined by the context.
This is achieved by using \emph{labels}, which track whether polymorphism was guessed
or confirmed by a type annotation.
Whereas FreezeML has type annotations on binders, Poly-ML has type annotations
on terms and propagates them using the label system.

In Poly-ML, the example $\lambda x. \dec{auto}\ x$ typechecks, guessing a
polymorphic type for $x$; FreezeML requires a type annotation on $x$.
In FreezeML the
program $\Let\; id = \lambda x.x \;\In\; \Let\;c = id\ 3 \;\In\; \dec{auto}\ \lceil id \rceil$
typechecks, whereas in Poly-ML a type annotation is required (in order to
convert between $\forall a. a \to a$ and $[\forall a. a \to a]$).
However, Poly-ML could be extended with a new construct for introducing boxed
polymorphism without a type annotation, using the principal type instead.
With such a change it is possible to translate from FreezeML into
this modified version of Poly-ML without inserting any new type
annotations (see~\auxref{sec:freezeml-to-poly-ml}).

\Auxref{sec:other-systems} contains an example-based comparison of \freezeml,
GI, MLF, HMF, FPH, and HML.

%

\paragraph{Instantiation as subsumption}

In FreezeML instantiation induces a natural subtyping relation such that $A \leq
B$ iff $B$ is an instance of $A$. In other systems (e.g. \cite{OderskyL96}) such
a subtyping relation applies implicitly to all terms via a subsumption
rule. This form of subsumption is fundamentally incompatible with frozen
variables, which explicitly suppress instantiation, enabling fine-grained
control over exactly where instantiation occurs. Nonetheless, subsumption comes
for free on unfrozen variables and potentially elsewhere if one adopts more
sophisticated instantiation strategies.

\section{Conclusions}
\label{sec:conclusion}

In this paper, we have introduced \freezeml as an exercise in language design
for reconciling ML type inference with System F-style first-class polymorphism.
We have also implemented \freezeml as part of \sysdesc programming
language~\citeanon{CooperLWY06}, which uses a variant of ML type inference
extended with row types, and has a kind system readily adapted to check that
inferred function arguments are monotypes.

Directions for future work include extending \freezeml to accommodate features such
as higher-kinds, GADTs, and dependent types,
as well as exploring different implicit instantiation strategies.
It would also be instructive to rework our formal account using the methodology
of \citet{GundryMM10} and use that as the basis for mechanised soundness and
completeness proofs.



\begin{acks}
This work was supported by EPSRC grant EP/K034413/1 `From Data Types to Session
Types---A Basis for Concurrency and Distribution', ERC Consolidator Grant Skye
(grant number 682315), by an LFCS internship, and by an ISCF Metrology
Fellowship grant provided by the UK government's Department for Business, Energy
and Industrial Strategy (BEIS). We are grateful to James McKinna, Didier
R\'{e}my, Andreas Rossberg, and Leo White for feedback and to anonymous
reviewers for constructive comments.
\end{acks}

\balance
\bibliography{bibliography}


\ifextended
\pagebreak
\onecolumn
\appendix

\section{\freezeml vs. Other Systems}
\label{sec:other-systems}

In this appendix we present an example-based comparison of \freezeml with other
systems for first-class polymorphism: GI~\cite{SerranoHVJ18},
MLF~\cite{BotlanR03}, HMF~\cite{Leijen08}, FPH~\cite{VytiniotisWJ08}, and
HML~\cite{Leijen09}.  Sections A-E of Figure~\ref{fig:freezeml-examples} have
been presented in~\cite{SerranoHVJ18}, together with analysis of how the five
systems behave for these examples.  We now use these examples to
compare \freezeml with other systems.

Firstly, we focus on which examples can be typechecked without explicit type
annotations. (We do not count \freezeml freezes, generalisations, and
instantiations as annotations, since these are mandatory in our system by design
and they do not require spelling out a type explicitly, allowing the programmer
to rely on type inference.) Out of 32 examples presented in Sections A-E of the
Figure~\ref{fig:freezeml-examples}, MLF typechecks all but B1 and E1, placing it
first in terms of expressiveness. HML ranks second, being unable to typecheck
B1, B2 and E1\footnote{Table presented in~\cite{SerranoHVJ18} claims that HML
cannot typecheck E3 but \expertname pointed out to us in private correspondence
that this is not the case and HML can indeed typecheck E3.}.  \freezeml handles
all examples except for A8, B1, B2, and E1, ranking third.  FPH, GI, and HMF
fail to typecheck 6 examples, 8 examples, and 11 examples respectively.  If we
permit annotations on binders only, the number of failures for most systems
decreases by 2, because the systems can now typecheck Examples B1 and B2.  MLF
was already able to typecheck B2 without an annotation, so now it handles all
but E1. If we permit type annotations on arbitrary terms the number of examples
that cannot be typechecked becomes: MLF -- 1 (E1), \freezeml -- 2 (A8, E1) -- GI
and HML -- 2 (E1, E3), FPH -- 4, and HMF -- 6.  These observations are
summarised in Table~\ref{tab:system-comparison} below.

\begin{table}[h]
\caption{Summary of the number of examples not handled by each system}
\label{tab:summary}
\begin{tabular}{|l|c|c|c|c|c|c|}
\hline
Annotate?    & MLF& HML& \freezeml & FPH & GI & HMF\\
\hline
Nothing & 2 & 3 & 4 & 6 & 8 & 11\\
Binders & 1 & 2 & 2 & 4 & 6 & 6\\
Terms   & 1 & 2 & 2 & 4 & 2 & 6\\
\hline
\end{tabular}
\label{tab:system-comparison}
\end{table}



Due to \freezeml's approach of explicitly annotating polymorphic instantiations, we
might require $\freeze{-}$, $\gen$, and $\inst$ annotations where other systems
need no annotations whatsoever. This is especially the case for Examples
A10-12, which all other five systems can handle without annotations.  We are
being more verbose here, but the additional ink required is minimal and we see
this as a fair price for the benefits our system provides. Also, being explicit
about generalisations allows us to be precise about the location of quantifiers
in a type. This allows us to typecheck Example E3, which no other system except
MLF can do.

\freezeml is incapable of typechecking A8, under the assumption that the only allowed
modifications are insertions of freeze, generalisation, and instantiation. We
can however $\eta$-expand and rewrite A8 to F10.

When dealing with $n$-ary function applications, \freezeml is insensitive to the
order of arguments. Therefore, if an application $M\ N$ is well-formed then so
are $\dec{app}\ M\ N$ and $\dec{revapp}\ N\ M$, as shown in section D of the
table. Many systems in the literature also enjoy this property, but there are
exceptions such as Boxy Types~\cite{VytiniotisWJ06}.

\section{Specifications of Core Calculi}
\label{sec:core-calculi}

In this appendix we provide full specification of two core calculi on which we
base \freezeml --- call-by-value System F and ML --- as well as translation from
ML to System F.

\subsection{Call-by-value System F}

\input{figures/system_f_syntax.tex}

We begin with a standard call-by-value variant of System F.
The syntax of System F types, environments, and terms is given in
Figure~\ref{fig:system-f-syntax}.

We let $a, b, c$ range over type variables.
We assume a collection of type constructors $\tc$ each of which has a fixed arity
$\arity(\tc)$.
Types formed by type constructor application include base types ($\Int$ and
$\Bool$), lists of elements of type $A$ ($\List\,A$), and functions from $A$ to
$B$ ($A \to B$).
Data types may be Church-encoded using polymorphic functions~\cite{Wadler90},
but for the purposes of our examples we treat them specially.
Types comprise type variables ($a$), fully-applied type constructors
($\tc\,\many{A}$), and polymorphic types ($\forall a.A$).
Type environments track the types of term variables in a term. Kind environments
track the type variables in a term. For the calculi we present in this section,
we only have a single kind, $\poly$, the kind of all types, which we
omit.
Nevertheless, kind environments are still useful for explicitly tracking which
type variables are in scope, and when we consider type inference
(Section~\ref{sec:inference}) we will need a refined kind system in order to
distinguish between monomorphic and polymorphic types.

We let $x, y, z$ range over term variables.
Terms comprise variables ($x$), term abstractions ($\lambda x^A.M$), term
applications ($M\,N$), type abstractions ($\Lambda a.V$), and type applications
($M\,A$).
We write $\Let\; x^A = M \;\In\; N$ as syntactic sugar for $(\lambda x^A.N)\,M$, we
write $M\,\many{A}$ as syntactic sugar for repeated type application
$M\,A_1\,\cdots\,A_n$, and $\Lambda \many{a}.V$ as syntactic sugar for repeated
type abstraction $\Lambda a_1.\cdots\Lambda a_n.V$.  We also may write $\Lambda
\Delta.A$ when $\Delta = \many{a}$.
We restrict the body of type abstractions to be syntactic values in accordance
with the ML value restriction~\cite{Wright95}.

\input{figures/system_f_kinding_typing.tex}

Well-formedness of types and the typing rules for System F are given in
Figure~\ref{fig:system-f-typing}.
Standard equational rules ($\beta$) and ($\eta$) for System F are given in
Figure~\ref{fig:system-f-equations}.
%

\input{figures/system_f_equational_rules.tex}





\subsection{ML}

\input{figures/ml_syntax.tex}

\input{figures/ml_kinding_typing.tex}

We now outline a core fragment of ML. The syntax is given in
Figure~\ref{fig:ml-syntax}, well-formedness of types and the typing rules in
Figure~\ref{fig:ml-typing}.
Unlike in System F we here separate monomorphic types ($S, T$) from type schemes
($P, Q$) and there is no explicit provision for type abstraction or type
application.
Instead, only variables may be polymorphic and polymorphism is introduced by
generalising the body of a let-binding ($\mlLab{Let}$), and eliminated
implicitly when using a variable ($\mlLab{Var}$).

Instantiation applies a type instantiation to the monomorphic body of a
polymorphic type. The rules for type instantiations are given in
Figure~\ref{fig:ml-typing}.
We may apply type instantiations to types and type schemes in the standard way:

\[
\ba{@{}r@{~}c@{~}l@{~}c@{~}r@{~}c@{~}l@{}}
\emptyset(S)            &=& S                     &\hspace{20pt}& \rsubst[a \mapsto S](a) &=& S\\
\rsubst(D~\many{S})     &=& D~(\many{\rsubst(S)}) &\hspace{20pt}& \rsubst[a \mapsto S](b) &=& \rsubst(b)
\ea
\]

Generalisation is defined at the bottom of Figure~\ref{fig:ml-typing}.  If $M$ is a value, the
generalisation operation $\mgend{S}{M}$ returns the list of type variables in $S$
that do not occur in the kind environment $\Delta$, in the order in which they
occur, with no duplicates. To satisfy the value restriction, $\mgend{S}{M}$ is
empty if $M$ is not a value.

A crucial difference between System F and ML is that in System F the order in
which quantifiers appear is important ($\forall a\,b.A$ and $\forall b\,a.A$ are
different types), whereas in ML, because instantiation is implicit, the order
does not matter.
As we are concerned with bridging the gap between the two we have developed an
extension of ML in which the order of quantifiers is important.  However, this
change does not affect the behaviour of type inference for ML terms since the
order of quantifiers is lost when polymorphic variable types are instantiated, as in
rule \textsc{ML-Var}.

\subsection{ML as System F}

ML is remarkable in providing statically typed polymorphism without the
programmer having to write any type annotations.
In order to achieve this coincidence of features the type system is carefully
constructed, and crucial operations (instantiation and generalisation) are left
implicit (i.e., not written as explicit constructs in the program).
This is convenient for programmers, but less so for metatheoretical study.

\input{figures/ml_to_system_f.tex}
In order to explicate ML's polymorphic type system, let us consider a translation of ML into
System F. Such a translation is given in Figure~\ref{fig:ml-to-f}.
As the translation depends on type information not available in terms, formally
it is defined as a translation from derivations to terms (rather than terms to
terms). But we abuse notation in the standard way to avoid explicitly writing
derivation trees everywhere.
Each recursive invocation on a subterm is syntactic sugar for invoking the
translation on the corresponding part of the derivation.

The translation of variables introduces repeated type applications.
Recall that we use $\Let \: x^A
\: = \: M \: \In \: N$ as syntactic sugar for $(\lambda x^A. N) \: M$ in System
F.
Translating the let binding of a value then yields repeated type abstractions.
For non-values $M$, $\Delta'$ is empty.

\begin{restatable}{thm}{translationtyping}
If $\Delta; \Gamma \vdash M : S$ then $\Delta; \Gamma \vdash \transmltof{M} :
S$.
\end{restatable}

The fragment of System F in the image of the translation is quite restricted in
that type abstractions are always immediately bound to variables and type
applications are only performed on variables.  Furthermore, all quantification must be
top-level.
Next we will extend ML in such a way that the translation can also be extended
to cover the whole of System F.

\section{Well-foundedness of FreezeML typing}
\label{sec:well-founded-freezeml}

In this appendix we give the full details of how FreezeML's typing relation can
be defined, despite the apparent failure of well-foundedness in the rule-based
presentation in  Figures~\ref{fig:freezeml-typing}
and~\ref{fig:freezeml-aux-judgments}.

We will define a function $\J{M}$ as follows by recursion on terms $M$.  The
result of $\J{M}$ is a set of triples $(\Delta,\Gamma,A)$.  Note that there is no
requirement that $M$, or $N$, are closed.  We define an auxiliary function
$P(-)$ that takes a set of triples $(\Delta,\Gamma,A)$ and produces a set of
quadruples $(\Delta,\Gamma,\Delta',A)$.  Intuitively, $\J{M}$ corresponds to
those triples $(\Delta,\Gamma,A)$ with respect to which $M$ is well-formed, and
$P(\J{M})$ corresponds analogously to those $(\Delta,\Gamma,\Delta',A)$
characterising a principal typing derivation for $M$.  We will make this
relationship precise shortly.

\begin{eqnarray*}
  \J{\freeze{x}} &=& \{(\Delta,\Gamma,A) \mid x : A \in \Gamma \text{ and }
                     \Delta \vdash \Gamma \text{ and } \Delta \vdash A : \poly\}\\
  \J{x} &=& \{(\Delta,\Gamma,\delta(H)) \mid x : \forall \Delta'.H \in
                    \Gamma \text{ and }
      \Delta \vdash \rsubst : \Delta' \Rightarrow_\poly \cdot \text{ and }
                     \Delta \vdash \Gamma \text{ and } \Delta \vdash \rsubst(H) : \poly\}\\
\J{M\,N} &=& \{(\Delta,\Gamma,B) \mid (\Delta,\Gamma,A
            \to B) \in \J{M} \text{ and } (\Delta, \Gamma,A) \in \J{N}\}\\
\J{\lambda x.M} &=& \{(\Delta,\Gamma,S \to B) \mid (\Delta,(\Gamma,x:S),B) \in
                   \J{M} \}\\
\J{\lambda (x:A).M} &=& \{(\Delta,\Gamma,A \to B) \mid (\Delta,(\Gamma,x:A),B) \in
                   \J{M} \}\\
\J{\Let \; x = M\; \In \; N} &=& \{(\Delta,\Gamma,B) \mid
(\Delta', \Delta'') = \mgend{A'}{M} \text{ and }
     (\Delta, \Delta'', M, A') \Updownarrow A \\
&&\qquad\qquad\text{ and }
     (\Delta, \Gamma,\Delta'',A')  \in P(\J{M}) \text{ and }
     (\Delta, (\Gamma, x : A),B) \in \J{N} \}
\\
\J{\Let \; (x : A) = M\; \In \; N} &=& \{(\Delta,\Gamma,B) \mid (\Delta', A')
                                          = \msplit(A, M) \text{ and }
     ((\Delta, \Delta'),\Gamma,A') \in \J{M} \text{ and }
     (\Delta, (\Gamma, x : A),B) \in \J{N}\}
\smallskip\\
P(X) &=& \{(\Delta,\Gamma,\Delta',A') \mid ((\Delta,\Delta'),\Gamma,A') \in
            X \text{ and for all }\Delta'', A''\\
&&\qquad\qquad\qquad\begin{array}{l}
       \text{if }
       \Delta'' = \ftv(A'') - \Delta \text{ and }
       ((\Delta, \Delta''), \Gamma, A'') \in X \\
       \text{then there exists }
       \rsubst
       \text{ such that }
       \;\Delta  \vdash \rsubst : \Delta' \Rightarrow_\poly \Delta''
       \text{ and }
       \rsubst(A') = A''\}
\end{array}
\end{eqnarray*}
As usual, we adopt the implicit convention that variables $x$ are
$\alpha$-renamed so as not to conflict with other names already in scope; that
is, in the cases for lambda-abstraction we implicitly assume $x \notin
FV(\Gamma)$ and for let, assume  $x \notin FV(\Gamma,M)$.

\begin{lemma}
For each $M$, the set $\J{M}$ is well-defined.
\end{lemma}
\begin{proof}
  Straightforward by induction on $M$, since in each case $\J{M}$ is defined
  in terms of $\J{-}$ applied to immediate subterms of $M$.
\end{proof}

\begin{definition}
We define the typing relation $\Delta;\Gamma \vdash M : A$ in terms
of $\J{-}$, and $\meta{principal}$ in terms of $P(-)$:
\[
\begin{array}{rcl}
\Delta;\Gamma \vdash M : A &\iff&
(\Delta,\Gamma,A) \in \J{M}\\
\meta{principal}(\Delta,\Gamma,M,\Delta',A) &\iff& (\Delta,\Gamma,\Delta',A) \in
  P(\J{M})\;.
\end{array}
\]
\end{definition}
  We will show that this relation satisfies all
of the rules listed in Figures~\ref{fig:freezeml-typing}
and~\ref{fig:freezeml-aux-judgments} and that the rules are invertible.
We first show that $\meta{principal}$ satisfies the definition given in Figure~\ref{fig:freezeml-aux-judgments}.

\begin{lemma}
$\meta{principal}(\Delta,\Gamma,M,\Delta',A')$ holds if and only if $\typ{\Delta,\Delta';\Gamma}{M:A'}$  and
 for all $\Delta'', A''$ if $\Delta'' = \ftv(A'') - \Delta$ and
      $\typ{\Delta, \Delta''; \Gamma}{M:A''}$ then there exists
       $\rsubst$
       such that
      $\Delta  \vdash \rsubst : \Delta' \Rightarrow_\poly \Delta''$ and
      $\rsubst(A') = A''$.
\end{lemma}
\begin{proof}
By unfolding definitions:
\begin{eqnarray*}
\meta{principal}(\Delta,\Gamma,M,\Delta',A')
&\iff&
(\Delta,\Gamma,\Delta',A') \in
  P(\J{M})\\
&\iff&
 ((\Delta,\Delta'),\Gamma,A') \in
             \J{M} \text{ and for all }\Delta'', A''\\
&&\qquad\begin{array}{l}
       \text{if }
       \Delta'' = \ftv(A'') - \Delta \text{ and }
       ((\Delta, \Delta''), \Gamma, A'') \in \J{M} \\
       \text{then there exists }
       \rsubst
       \text{ such that }
       \;\Delta  \vdash \rsubst : \Delta' \Rightarrow_\poly \Delta''
       \text{ and }
       \rsubst(A') = A''
\end{array}\\
&\iff&
\typ{\Delta,\Delta';\Gamma}{M:A'} \text{ and for all }\Delta'', A''\\
&&\qquad\begin{array}{l}
       \text{if }
       \Delta'' = \ftv(A'') - \Delta \text{ and }
       \typ{\Delta, \Delta'';\Gamma}{M:A''}\\
       \text{then there exists }
       \rsubst
       \text{ such that }
       \;\Delta  \vdash \rsubst : \Delta' \Rightarrow_\poly \Delta''
       \text{ and }
       \rsubst(A') = A''
\end{array}
\end{eqnarray*}
\end{proof}

Likewise, the next lemma shows that all of the inference rules listed in
Figure~\ref{fig:freezeml-typing} hold of the typing relation as defined using $\J{}$.  This is largely straightforward, but the
details are given explicitly for the cases involving $\meta{principal}$ and \freezemlLab{Let},
in order to make it clear that there is no circularity.

\begin{lemma}
~\\
  \begin{enumerate}
  \item If $x : A \in \Gamma$ and $\Delta \vdash \Gamma$ and $\Delta \vdash A
  : \poly$ then $\typ{\Delta; \Gamma}{\freeze{x} : A}$.
\item If $x : \forall \Delta'.H \in \Gamma$ and $\Delta \vdash \Gamma$ and $\Delta \vdash \rsubst(H)
  : \poly$ and $\Delta \vdash \rsubst : \Delta'
  \Rightarrow_\poly \cdot$ then $\typ{\Delta; \Gamma}{x : \rsubst(H)}$.
\item If $\typ{\Delta; \Gamma}{M : A \to B}$ and $\typ{\Delta; \Gamma}{N : A}$ then
    $\typ{\Delta; \Gamma}{M\,N : B}$.
\item If $\typ{\Delta; \Gamma, x : S}{M : B}$ then
   $\typ{\Delta; \Gamma}{\lambda x.M : S \to B}$.
\item If $\typ{\Delta; \Gamma, x : A}{M : B}$ then
   $\typ{\Delta; \Gamma}{\lambda (x:A).M : A \to B}$.
\item If the following hold
\[
\bl
(\Delta', \Delta'') = \mgend{A'}{M} \\
(\Delta, \Delta'', M, A') \Updownarrow A \\
\typ{\Delta,\Delta''; \Gamma}{M:A'} \\
\typ{\Delta; \Gamma, x : A}{N:B} \\
\meta{principal}(\Delta,\Gamma,M,\Delta'',A') \\
\el
\]
then $\typ{\Delta;\Gamma}{\Let \; x = M\; \In \; N: B}$.
\item If the following hold
\[
\bl
(\Delta', A') = \msplit(A, M) \\
\typ{\Delta, \Delta';\Gamma}{M:A'} \\
\typ{\Delta; \Gamma, x : A}{N:B} \\
\el
\]
then $\typ{\Delta;\Gamma}{\Let \; (x : A) = M\; \In \; N:B}$.

  \end{enumerate}
\end{lemma}
\begin{proof}
  In each case, the reasoning is straightforward by unfolding definitions.  We
  give the details of the case for \freezemlLab{Let}.

  Assume the following:
\[
\bl
(\Delta', \Delta'') = \mgend{A'}{M} \\
(\Delta, \Delta'', M, A') \Updownarrow A \\
\typ{\Delta,\Delta''; \Gamma}{M:A'} \\
\typ{\Delta; \Gamma, x : A}{N:B} \\
\meta{principal}(\Delta,\Gamma,M,\Delta'',A')
\el
\]

By definition, we also know that $(\Delta,(\Gamma,x:A),B) \in \J{N}$ and
$(\Delta,\Gamma,\Delta'',A') \in P(\J{M})$.  These are the required facts (along
with the first two) to conclude that $(\Delta,\Gamma,B) \in \J{\Let \; x = M\;
  \In \; N: B}$, which is equivalent by definition to $\typ{\Delta;\Gamma}{\Let \; x = M\; \In \; N: B}$.
\end{proof}

Furthermore, the rules are all \emph{invertible}; that is, if a conclusion of a
rule is derivable, then some instantiations of the hypotheses are also
derivable:
\begin{lemma}
~\\
  \begin{enumerate}
  \item If $\typ{\Delta; \Gamma}{\freeze{x} : A}$ then $x : A \in \Gamma$.
\item If  $\typ{\Delta; \Gamma}{x : A}$ then there exists $\Delta', H$ such that
  $x : \forall \Delta'.H \in \Gamma$ and $\Delta \vdash \rsubst : \Delta'
  \Rightarrow_\poly \cdot$.
\item If $\typ{\Delta; \Gamma}{M\,N : B}$ then there exists $A$ such that $\typ{\Delta; \Gamma}{M : A \to B}$ and $\typ{\Delta; \Gamma}{N : A}$.
\item If $\typ{\Delta; \Gamma}{\lambda x.M : S \to B}$ then $\typ{\Delta; \Gamma, x : S}{M : B}$.
\item If $\typ{\Delta; \Gamma}{\lambda (x:A).M : A \to B}$ then $\typ{\Delta; \Gamma, x : A}{M : B}$.
\item If  $\typ{\Delta;\Gamma}{\Let \; x = M\; \In \; N: B}$ then there exist
  $\Delta',\Delta'',A'$ such that:
\[
\bl
(\Delta', \Delta'') = \mgend{A'}{M} \\
(\Delta, \Delta'', M, A') \Updownarrow A \\
\typ{\Delta,\Delta''; \Gamma}{M:A'} \\
\typ{\Delta; \Gamma, x : A}{N:B} \\
\meta{principal}(\Delta,\Gamma,M,\Delta'',A') \\
\el
\]
\item If $\typ{\Delta;\Gamma}{\Let \; (x : A) = M\; \In \; N:B}$ then there exist
  $\Delta',A'$ such that:
\[
\bl
(\Delta', A') = \msplit(A, M) \\
\typ{\Delta, \Delta';\Gamma}{M:A'} \\
\typ{\Delta; \Gamma, x : A}{N:B} \\
\el
\]
\end{enumerate}
\end{lemma}
\begin{proof}
  In each case, the reasoning is straightforward by unfolding definitions.  We
  again give the details of the case for \freezemlLab{Let}.

Suppose that $\typ{\Delta;\Gamma}{\Let \; x = M\; \In \; N: B}$ holds; that is,
$(\Delta,\Gamma,B) \in \J{\Let \; x = M\; \In \; N}$.  By definition,
\begin{eqnarray*}
(\Delta,\Gamma,B) \in \J{\Let \; x = M\; \In \; N} &\iff&
(\Delta', \Delta'') = \mgend{A'}{M} \text{ and }
     (\Delta, \Delta'', M, A') \Updownarrow A \\
&&\text{ and }
     (\Delta, \Gamma,\Delta'',A')  \in P(\J{M}) \text{ and }
     (\Delta, (\Gamma, x : A),B) \in \J{N} \\
&\iff&
(\Delta', \Delta'') = \mgend{A'}{M} \text{ and }
     (\Delta, \Delta'', M, A') \Updownarrow A \\
&&\text{ and }
\typ{\Delta,\Delta';\Gamma}{M:A'} \text{ and }
     \typ{\Delta; \Gamma, x : A}{N:B} \text{ and }
     \meta{principal}(\Delta, \Gamma,M,\Delta'',A')
\end{eqnarray*}
where in the last step we use the fact that $\meta{principal}(\Delta,
\Gamma,M,\Delta'',A')$ implies $\typ{\Delta,\Delta';\Gamma}{M:A'}$.
\end{proof}

Therefore, we may reason about the typing relation by induction on the structure
of $M$, and immediately applying inversion in each case.  It is also possible
to define functions structurally over derivations, provided that the principality information
is ignored.  For example, the translations in Figure~\ref{fig:trans-freezeml-to-f}
and~\auxref{app:freezeml-to-polyml} have this form.  To indicate that the
principality information is not used in recursion over derivations, this
assumption is greyed out.

In the previous inversion lemma, we did not mention the well-formedness
preconditions in the variable case.  Instead, we show that typing relations
always involve well-formed $\Gamma$ and $A$ with respect to $\Delta$:
\begin{lemma}
  If $\typ{\Delta;\Gamma}{M:A}$ then $\Delta \vdash \Gamma$ and $\Delta \vdash A
  : \poly$.
\end{lemma}
\begin{proof}
  By induction on $M$, then analysis of the corresponding case of $\J{-}$.  The
  cases for variables and frozen variables are immediate since the required
  relations are preconditions.  For most other cases, the induction hypothesis and then
  inversion on some subderivation suffices.  We give the details for let, to
  illustrate the required reasoning.

Suppose $(\Delta,\Gamma,B) \in \J{\Let \; x = M\; \In \; N}$.  By definition,
this means that $(\Delta,(\Gamma,x:A),B) \in \J{N}$ must also hold (among other
preconditions).  Therefore, by the induction hypothesis for $N$, we have $\Delta
\vdash \Gamma,x:A$ and $\Delta \vdash B : \poly$.  By inversion on derivations
of context well-formedness we have $\Delta \vdash \Gamma$, which concludes the
proof for this case.
\end{proof}
Notice in particular that in cases such as ascribed lambda and let, we need not
explicitly check that $A$ is well-formed with respect to $\Delta$, since it is
necessary by construction (though it also would not hurt to perform such a check).

\section{Example Translation from \freezeml to System F}
\label{sec:freezeml-to-poly-ml}

Below is an example translation from \freezeml to System F, where $\dec{app}$,
$\dec{auto}$, and $\dec{id}$ have the types given in
Figure~\ref{fig:function-signatures}.
\[
\bl
    \quad\transfreezemlfof{\Let \: \dec{app} = \lambda f. \lambda z. f \: z \: \In \:
           \dec{app} \: \freeze{\dec{auto}} \: \freeze{\dec{id}}} \\[1ex]
=   \Let \: \dec{app}^{\forall a\,b. (a \to b) \to a \to b} =\\
\hspace{1cm}\Lambda a \,b. \transfreezemlfof{\lambda f. \lambda z. f \: z} \ \In \: \transfreezemlfof{\dec{app} \: \freeze{\dec{auto}} \: \freeze{\dec{id}}}\\[1ex]
=   ( \lambda \dec{app}^{\forall a\,b. (a \to b) \to a \to b}. \\
\hspace{1cm}   \transfreezemlfof{\dec{app} \: \freeze{\dec{auto}} \: \freeze{\dec{id}}} ) \ ( \Lambda a \,b. \transfreezemlfof{\lambda f. \lambda z. f \: z} ) \\[1ex]
=   ( \lambda \dec{app}^{\forall a\,b. (a \to b) \to a \to b}. \\
\hspace{1cm}   \transfreezemlfof{\dec{app} \: \freeze{\dec{auto}} \: \freeze{\dec{id}}} ) \ ( \Lambda a \,b. \lambda f^{a \to b}. \lambda z^a. f \: z ) \\
\el
\]
where subterm
$\transfreezemlfof{\dec{app} \: \freeze{\dec{auto}} \: \freeze{\dec{id}}}$
further translates as:
\[
   \transfreezemlfof{\dec{app} \: \freeze{\dec{auto}} \: \freeze{\dec{id}}}
    =  \dec{app}\,
          \bl
          ((\forall a. a \to a ) \to (\forall a. a \to a )) \\
          (\forall a. a \to a) \\
           \dec{auto} \\
           \dec{id} \\
          \el \\
\]

The type of the whole translated term is $\forall a. a \to a$.  The
translation enjoys a type preservation property.

\section{Translation from \freezeml to Poly-ML}
\label{app:freezeml-to-polyml}

\newcommand{\transfreezemlttopolyml}[2][]{\llbracket #2 \rrbracket_{\tau}}
\newcommand{\transfreezemlttSigmatopolyml}[2][]{\llbracket #2 \rrbracket_{\sigma}}
\newcommand{\transfreezemlttVarSigmatopolyml}[2][]{\llbracket #2 \rrbracket_{\varsigma}}
\newcommand{\transfreezemlfopolyml}[2][]{\left\llbracket #2 \right\rrbracket}
\newcommand{\polyopen}[1]{\langle #1 \rangle}

\paragraph{Types}
Let $\epsilon$ be a fixed label.
Then $\transfreezemlttopolyml{}$ is defined as follows:
\begin{align*}
\transfreezemlttopolyml{a} &= a \\
\transfreezemlttopolyml{A_1 \to A_2} &= \transfreezemlttopolyml{A_1} \to \transfreezemlttopolyml{A_2} \\
\transfreezemlttopolyml{\forall \Delta.H} &= [\forall \Delta. \transfreezemlttopolyml{H} ]^\epsilon \qquad\qquad \text{if $\Delta \neq \cdot$ }
\end{align*}

Further, $\transfreezemlttSigmatopolyml{}$ is defined as follows,
meaning that $\transfreezemlttSigmatopolyml{}$ behaves like $\transfreezemlttopolyml{}$ but leaves quantifiers at the toplevel unboxed.
\begin{align*}
\transfreezemlttSigmatopolyml{a} &= \transfreezemlttopolyml{a} \\
\transfreezemlttSigmatopolyml{A_1 \to A_2} &= \transfreezemlttopolyml{A_1 \to A_2} \\
\transfreezemlttSigmatopolyml{\forall \Delta.H} &= \forall \Delta. \transfreezemlttopolyml{H}  \qquad\qquad \text{if $\Delta \neq \cdot$ }
\end{align*}

Finally, $\transfreezemlttVarSigmatopolyml{A}$ is defined as $ \forall \epsilon. \transfreezemlttopolyml{A}$ and is applied to typing environments by applying $\transfreezemlttVarSigmatopolyml{}$ to the types therein.

\paragraph{Terms (Core)}
\[
\begin{array}{rcll}
  \transfreezemlfopolyml[\Bigg]
    {\inferrule
      {x : A \in \Gamma}
      {\typ{\Delta; \Gamma}{\freeze{x} : A}}}
  &= &x \\[4ex]

  \hspace{-10pt}\transfreezemlfopolyml[\ginormous]{\raisebox{-7pt}
    {$\inferrule
      {{x : \forall \Delta'.H \in \Gamma} \\\\
        \Delta \vdash \rsubst : \Delta' \Rightarrow_\poly \cdot
      }
      {\typ{\Delta; \Gamma}{x : \rsubst(H)}}
    $}} &=
          &\begin{cases}
            x & \text{if } \Delta' = \cdot \\
            \polyopen{x} & \text{otherwise}
            \end{cases} \\[6ex]

\hspace{-10pt}\transfreezemlfopolyml[\ginormous]{\raisebox{-7pt}
    {$\inferrule
      {\typ{\Delta; \Gamma}{M : A \to B} \\\\
        \typ{\Delta; \Gamma}{N : A}
      }
      {\typ{\Delta; \Gamma}{M\,N : B}}
    $}}
  &= &\transfreezemlfopolyml{M} \; \transfreezemlfopolyml{N} \\[6ex]

  \transfreezemlfopolyml[\Bigg]
    {\inferrule
      {\typ{\Delta; \Gamma, x : S}{M : B}}
      {\typ{\Delta; \Gamma}{\lambda x.M : S \to B}}
    }
  &= &\lambda x.\transfreezemlfopolyml{M} \\[4ex]

  \transfreezemlfopolyml[\Bigg]
    {\inferrule
      {\typ{\Delta; \Gamma, x : A}{M : B}}
      {\typ{\Delta; \Gamma}{\lambda (x : A).M : A \to B}}
    }
  &= &\lambda (x : \transfreezemlttopolyml{A}).\transfreezemlfopolyml{M} \\[1ex]
  &=& \lambda x. \Let \: x =  (x : \transfreezemlttopolyml{A}) \: \In \: \transfreezemlfopolyml{M} \\[4ex]

\end{array}
\]

\paragraph{Terms (Let, value-restricted)}

\[
\begin{array}{rcll}

\transfreezemlfopolyml[\Ginormouser]
    {\raisebox{-33pt}{$\inferrule*
      { (\Delta', \Delta'') = \mgend{A'}{M} \\\\
        \Delta, \Delta''; \Gamma \vdash M : A' \\\\
        (\Delta, \Delta'', M, A') \Updownarrow A \\\\
        \typ{\Delta; \Gamma, x : A}{N : B} \\\\
        {\color{gray}\meta{principal}(\Delta, \Gamma, M, \Delta'', A')}
     }
      {\typ{\Delta; \Gamma}{\Let \; x = M\; \In \; N : B}}
    $}}
  &=
  &\begin{cases}
    \Let \: x = [ \transfreezemlfopolyml{M} : \transfreezemlttSigmatopolyml{A}] \; \In\; \transfreezemlfopolyml{N}
    & \text{if $\Delta' \neq \cdot$ \begin{minipage}[t]{4cm}\fe{could omit the type annotation if using principal-type version of [ ]}\end{minipage}} \\
    \Let \: x = \transfreezemlfopolyml{M} \; \In\; \transfreezemlfopolyml{N} & \text{otherwise } \\
  \end{cases} \\[16ex]

\end{array}
\]
Note that in the first case above, the type annotation
$\transfreezemlttSigmatopolyml{A}$ would not be needed if Poly-ML was extended
with a boxing operator that does not require a type annotation but uses the
principal type instead.

\[
\begin{array}{rcll}

\transfreezemlfopolyml[\Ginormous]
    {\raisebox{-27pt}{$\inferrule*
      { (\Delta', A') = \msplit(A, M) \\\\
        \typ{\Delta, \Delta'; \Gamma}{M : A'} \\\\
        A = \forall \Delta'.A' \\\\
        \typ{\Delta; \Gamma, x : A}{N : B}
      }
      {\typ{\Delta; \Gamma}{\Let \; (x : A) = M\; \In \; N : B}}
    $}} &=

  &\begin{cases}
    \Let \: x = [ \transfreezemlfopolyml{M} : \transfreezemlttSigmatopolyml{A}] \; \In\; \transfreezemlfopolyml{N}
    & \text{if $\Delta' \neq \cdot$ } \\
    \Let \: x = \transfreezemlfopolyml{M} \; \In\; \transfreezemlfopolyml{N} & \text{otherwise } \\
  \end{cases} \\[16ex]

\end{array}
\]

\begin{lemma}
If $\typ{\Delta;\Gamma}{M : A}$ in FreezeML, then $\typ{\transfreezemlttVarSigmatopolyml{\Gamma}}{ \transfreezemlfopolyml{M} : \transfreezemlttopolyml{A} }$ in Poly-ML.
\end{lemma}

\section{Proofs from Section~\ref{sec:translations}}

{



For convenience,
we use the following (derivable) System F typing rules, allowing $n$-ary type applications
and abstractions:
\begin{mathpar}
\inferrule*[Right=\fLab{PolyApp*}]
    {\typ{\Delta;\Gamma}{M : \forall \Delta' . B} \\
      \Delta' = a_1, \dotsc, a_n \\
      \many{A} = A_1, \dotsc, A_n
    }
    {\typ{\Delta;\Gamma}{M\,\many{A} : B[A_1/a_1] \cdots [A_n/a_n]} \\\\
    \text{where $\many{A}$ may be empty} } \\
\inferrule*[Right=\fLab{PolyLam*}]
    {\typ{\Delta, \Delta';\Gamma}{V : A}}
    {\typ{\Delta;\Gamma}{\Lambda \, \Delta' . V : \forall \Delta' . A}}
\end{mathpar}

Recall that we have defined $\Let \: x^A = M \: \In \: N$ as syntactic
sugar for $(\lambda x^A. N) \: M$ in System F .

For readability, we preserve the syntactic sugar in the proofs and use the
following typing rule:

\begin{gather*}
{\inferrule*[right=\fLab{Let}]
      {
        \typ{\Delta; \Gamma}{M : A} \\
        \typ{\Delta; \Gamma, x : A}{N : B}
      }
      {\typ
        {\Delta; \Gamma}
        {\Let \: x^A = M \: \In \: N : B }
      }
    }
\end{gather*}

\newcommand{\hypo}[1]{#1}

\begin{lemma}
\label{lem:system-f-value-translates-to-fml-value}
For each System F value $V$, $\transftofreezeml{V}$ is a \freezeml value.

\end{lemma}
\begin{proof}
By induction on structure of $V$.
\end{proof}

\ftofreezemltypepreservation*
\begin{proof}
The proof is by well-founded induction on derivations of $\typ{\Delta;\Gamma}{M
  : A}$.  This means that we may apply the induction hypothesis to any judgement
appearing in a subderivation, not just to those appearing in the immediate
ancestors of the conclusion.
We slightly strengthen the induction hypothesis so that the $A$ is the \emph{unique}
type of $\transftofreezeml{M}$.
Formally, we show that if $\typ{\Delta;\Gamma}{M : A}$ holds in System F,
then $\typ{\Delta;\Gamma}{\transftofreezeml{M} : A}$ holds in \freezeml and
for all $B$ with $\typ{\Delta;\Gamma}{\transftofreezeml{M} : B}$ we have $A = B$.
We show how to extend $\transftofreezeml{-}$ to a function that translates
System F type derivations to \freezeml type derivations.

\newcommand{\transto}{\leadsto_F}

\begin{itemize}
\item Case \fLab{Var}, $\mathcal{J} =  \typ{\Delta; \Gamma}{x : A}$:
\begin{equations}
 \transftofreezeml[\bigg]
  {\inferrule
    {x : A \in \Gamma}
    {\typ{\Delta; \Gamma}{x : A}}
  }

  &\Longrightarrow
  &\inferrule*[vcenter]
    {x : A \in \Gamma}
    {\typ{\Delta; \Gamma}{\freeze{x} : A}} \\[4ex]
\end{equations}

\item Case \fLab{Lam},  $\mathcal{J} =  \typ{\Delta;\Gamma}{\lambda x^A.M : A \to B}$:
\begin{equations}
  \transftofreezeml[\bigg]
  {\inferrule
    {\typ{\Delta;\Gamma, x : A}{M : B}}
    {\typ{\Delta;\Gamma}{\lambda x^A.M : A \to B}}
  }

  &\Longrightarrow
  &\inferrule*[vcenter]
    {\hypo{\typ{\Delta; \Gamma, x : A}{\transftofreezeml{M} : B}}}
    {\typ{\Delta; \Gamma}{\lambda (x : A).\transftofreezeml{M} : A \to B}} \\[4ex]
\end{equations}

\item Case \fLab{App}, $\mathcal{J} = \typ{\Delta;\Gamma}{M \; N : B}$:
\begin{equations}
  \transftofreezeml[\ginormous]
  {\raisebox{-7.5pt}{$
    {\inferrule
      {\typ{\Delta;\Gamma}{M : A \to B} \\\\
        \typ{\Delta;\Gamma}{N : A}
      }
      {\typ{\Delta;\Gamma}{M \; N : B}}
    }
   $}}

  &\Longrightarrow
  &\inferrule
    {\hypo{\typ{\Delta; \Gamma}{\transftofreezeml{M} : A \to B}} \\\\
     \hypo{\typ{\Delta; \Gamma}{\transftofreezeml{N} : A}}
    }
    {\typ{\Delta; \Gamma}{\transftofreezeml{M} \; \transftofreezeml{N} : B}} \\[8ex]
\end{equations}

\item Case \fLab{TAbs}, , $\mathcal{J} = \typ{\Delta;\Gamma}{\Delta a. V  : \forall a.B}$

Let $B = \forall \Delta_B. H_B$.
By Lemma~\ref{lem:system-f-value-translates-to-fml-value},  $\transftofreezeml{V}$ is a value, and $\Let\; y = \transftofreezeml{V} \;\In\; y$ is a guarded value which we refer to as $U_@$.

\begin{gather*}
\transftofreezeml[\bigg]
{
  \inferrule
    {\typ
          {\Delta, a;\Gamma}
          {V : B}
    }
    {\typ{\Delta;\Gamma}{\Lambda a . V  : \forall a . B}}
}
   \quad\Longrightarrow
 \\[2ex]
    \inferrule*
    {
     \mathcal{D} \\
     \inferrule*
       {x : \forall a.B \in \Gamma, (x : \forall a.B)}
       {\typ{\Delta; \Gamma, (x:\forall a.B)}{\freeze{x} :
           \forall a.B}} \\
       ((a, \Delta_B),H_B) = \meta{split}(\forall a.B, U_@) \\
    }
    {\typ
      {\Delta; \Gamma}
      { \Let \; (x : \forall a. B) = U_@ \; \In \;
        \freeze{x} :  \forall a.B}
    }
  \end{gather*}

The sub-derivation $\mathcal{D}$ for $\typ{\Delta, a, \Delta_B; \Gamma}{U_@ : H_B}$ differs based on whether $\transftofreezeml{V}$ is a guarded value or not: \\

If $\transftofreezeml{V} \in \dec{GVal}$:
$\transftofreezeml{V}$ must have a guarded type and hence we have $B = H_B$ and $\Delta_B = \cdot$.
By induction we have $\typ{\Delta, a;\Gamma}{\transftofreezeml{V} : B}$ and hence $\ftv(B) \subseteq \Delta, a$. This further implies $\mgend[(\Delta, a, \Delta_B)]{H_B}{\transftofreezeml{V}} = (\cdot, \cdot)$.
Let $\rsubst$ be the empty substitution.
\begin{mathpar}
\inferrule*
    {
     \Delta, a, \Delta_B; \Gamma \vdash \transftofreezeml{V} : H_B \\
     \inferrule*
       {y : H_B \in \Gamma \\\\
        \Delta, a, \Delta_B \vdash \rsubst : \cdot \Rightarrow_\poly \cdot}
      {\typ{\Delta; \Gamma, y : H_B}{y : \rsubst(H_B)}}
      \\\\
     ((\Delta, a, \Delta_B), \cdot, \transftofreezeml{V}, H_B) \Updownarrow H_B \\
      (\cdot, \cdot) = \mgend[(\Delta, a, \Delta_B)]{H_B}{\transftofreezeml{V}} \\
     \meta{principal}((\Delta, a, \Delta_B), \Gamma, \transftofreezeml{V}, \cdot, B')
    }
    {\typ{\Delta, a, \Delta_B; \Gamma}{\Let\; y = \transftofreezeml{V} \;\In\; y : H_B}}

\end{mathpar}

If $\transftofreezeml{V} \not\in \dec{GVal}$:
Let $B' = \forall \Delta'.H'$ be alpha-equivalent to $B$ such that all $\Delta'$ are fresh.
We then have $\Delta, a; \Gamma \vdash \transftofreezeml{V} : B'$ by induction.
This implies $\mgend[(\Delta, a, \Delta_B)]{B'}{\transftofreezeml{V}} = (\cdot, \cdot)$.
Let $\rsubst$ be defined such that $\rsubst(\Delta') = \Delta_B$, which implies $\rsubst(H') = H_B$.
\begin{mathpar}
\inferrule*
    {
     \Delta, a, \Delta_B; \Gamma \vdash \transftofreezeml{V} : B' \\
     \inferrule*
       {y : \Delta'.H' \in \Gamma \\\\
        \Delta, a, \Delta_B \vdash \rsubst : \Delta' \Rightarrow_\poly \cdot}
      {\typ{\Delta, a, \Delta_B; \Gamma, y : B'}{y : \rsubst(H')}}
      \\\\
     ((\Delta, a, \Delta_B), \cdot, \transftofreezeml{V}, B') \Updownarrow B' \\
      (\cdot, \cdot) = \mgend[(\Delta, a, \Delta_B)]{B'}{\transftofreezeml{V}} \\
     \meta{principal}((\Delta, a, \Delta_B), \Gamma, \transftofreezeml{V}, \cdot, B')
    }
    {\typ{\Delta, a, \Delta_B; \Gamma}{\Let\; y = \transftofreezeml{V} \;\In\; y : H_B}}

\end{mathpar}

In both cases, satisfaction of
$\meta{principal}((\Delta, a, \Delta_B), \Gamma, \transftofreezeml{V}, \cdot, B')$ follows from the
fact that by induction, $B'$ is the unique type of $\transftofreezeml{V}$.

\item Case \fLab{TApp}, , $\mathcal{J} = \typ{\Delta;\Gamma}{M \; A : B[A/a]}$

Let $B = \forall \Delta_B. H_B$ and w.l.o.g.\ $a \fresh \Delta_B$ and $\ftv(A) \fresh a, \Delta_B$ .
Let $U_@$ be defined as in the previous case. We then have

\begin{gather*}
  \transftofreezeml[\bigg]
  {
  {\inferrule
    {\typ{\Delta;\Gamma}{M : \forall a . B}}
    {\typ{\Delta;\Gamma}{M\,A : B[A/a]}}
  }
  }
  \Longrightarrow \\[2ex]
  \inferrule
    {
     \mathcal{D} \\
     \inferrule*
       {x : B[A/a] \in \Gamma, (x : B[A/a])}
       {\typ{\Delta; \Gamma, (x:B[A/a])}{\freeze{x} :
           B[A/a]}} \\
     (\Delta_B,H_B[A/a]) = \meta{split}(B[A/a], U_@) \\
    }
    {\typ
      {\Delta; \Gamma}
      {\Let \, (x  : B[A/a]) = U_@  \: \In \: \freeze{x} : B[A/a]}
    }
\end{gather*}

We consider the sub-derivation $\mathcal{D}$ for $\typ{\Delta, \Delta_B; \Gamma}{U_@ : H_B[A/a]}$ \\

By induction, we have $\typ{\Delta;\Gamma}{\transftofreezeml{V} : \forall a. B}$, which implies that $\transftofreezeml{V}$ is not a guarded value.

Let $B' = \forall \Delta'.H'$ be alpha-equivalent to $B$ such that all $\Delta'$ are fresh.
We then have $\Delta; \Gamma \vdash \transftofreezeml{V} : \forall a. B'$ by induction.
This implies $(\cdot, \cdot) = \mgend[(\Delta,  \Delta_B)]{\forall a. B'}{\transftofreezeml{V}}$.
Let $\rsubst$ be defined such that $\rsubst(\Delta') = \Delta_B$ and $\rsubst(a) = A$, which implies $\rsubst(H') = H_B[A/a]$.
\begin{mathpar}
\inferrule*
    {
     \Delta; \Gamma \vdash \transftofreezeml{V} : \forall a. B' \\
     \inferrule*
       {y : a, \Delta'.H' \in \Gamma \\\\
        \Delta, \Delta_B \vdash \rsubst : a, \Delta' \Rightarrow_\poly \cdot}
      {\typ{\Delta, \Delta_B; \Gamma, y : \forall a. B'}{y : \rsubst(H')}}
      \\\\
     ((\Delta, \Delta_B), \cdot, \transftofreezeml{V}, \forall a. B') \Updownarrow \forall a. B' \\
      (\cdot, \cdot) = \mgend[(\Delta,  \Delta_B)]{\forall a. B'}{\transftofreezeml{V}} \\
     \meta{principal}((\Delta, \Delta_B), \Gamma, \transftofreezeml{V}, \cdot, \forall a. B')
    }
    {\typ{\Delta,  \Delta_B; \Gamma}{\Let\; y = \transftofreezeml{V} \;\In\; y : H_B[A/a]}}

\end{mathpar}

As in the previous case, satisfaction of
$\meta{principal}((\Delta, a, \Delta_B), \Gamma, \transftofreezeml{V}, \cdot, \forall a.B')$ follows from the
fact that by induction, $\forall a.B'$ is the unique type of $\transftofreezeml{V}$.

\end{itemize}

Finally, we observe that the translated terms indeed have unique types: For
variables, the type is uniquely determined from the context.  Functions are
translated to annotated lambdas, without any choice for the parameter type.  For
term applications, uniqueness follows by induction.  For term applications an
abstractions, the result type of the expression is the type of freezing $x$.  In
both cases, this variable is annotated with a type.

This completes the proof, since any derivation is in one of the forms
used in the above cases.
\end{proof}

\freezemltoftypepreservation*
\begin{proof}
We perform induction on $M$, in each case using inversion on the derivation of $\typ{\Delta, \Gamma}{M : A}$.  In
each case we show how the definition of $\transfreezemlfof{-}$ can be extended to a
function returning the desired derivation.

\begin{itemize}
\item Case \freezemlLab{Freeze}:
     \begin{equations}

  \transfreezemlfof[\Bigg]
    {\inferrule
      {x : A \in \Gamma}
      {\typ{\Delta; \Gamma}{\freeze{x} : A}}
    }
  &\Longrightarrow\;
  & \inferrule*[vcenter,right=\fLab{Var}]
      {x : A \in \Gamma}
      {\typ{\Delta; \Gamma}{x : A}} \\[6ex]

     \end{equations}
\item Case \freezemlLab{Var}:
  Let $\Delta' = (a_1, \dotsc, a_n)$.

     \begin{equations}

  \transfreezemlfof[\Bigg]
    {\inferrule
      {{x : \forall \Delta'.H \in \Gamma} \\
       \Delta \vdash \rsubst : \Delta' \Rightarrow \cdot
      }
      {\typ{\Delta; \Gamma}{x : \rsubst(H)}}
    }
  &\Longrightarrow
  &\inferrule*[vcenter,right=\fLab{PolyApp*}]
    {\inferrule*[vcenter,right=\fLab{Var}]
      {x : \forall \Delta'.H \in \Gamma}
      {\typ{\Delta; \Gamma}{x : \forall a_1, \dotsc, a_n.H}} \\
    }
    {\typ{\Delta;\Gamma}{ x  \,\rsubst{a_1} \, \cdots \, \rsubst{a_n} : H[\rsubst{a_1}/a_1]\cdots[\rsubst{a_n}/a_n]}} \\

  && \\[2ex]

     \end{equations}
\item Case \freezemlLab{Lam}:
     \begin{equations}

  \transfreezemlfof[\Bigg]
    {\inferrule
      {\typ{\Delta; \Gamma, x : S}{M : B}}
      {\typ{\Delta; \Gamma}{\lambda x.M : S \to B}}
    }
  &\Longrightarrow
  &\inferrule*[vcenter,right=\fLab{Lam}]
     {\hypo{\typ{\Delta;\Gamma, x : S}{\transfreezemlfof{M} : B}}}
     {\typ{\Delta;\Gamma}{\lambda x^S.\transfreezemlfof{M} : S \to B}} \\[6ex]

     \end{equations}
\item Case \freezemlLab{Lam-Ascribe}
     \begin{equations}

  \transfreezemlfof[\Bigg]
    {\inferrule
      {\typ{\Delta; \Gamma, x : A}{M : B}}
      {\typ{\Delta; \Gamma}{\lambda (x : A).M : A \to B}}
    }
  &\Longrightarrow
  &\inferrule*[vcenter,right=\fLab{Lam}]
     {\hypo{\typ{\Delta;\Gamma, x : A}{\transfreezemlfof{M} : B}}}
     {\typ{\Delta;\Gamma}{\lambda x^A.\transfreezemlfof{M} : A \to B}} \\[6ex]

     \end{equations}
\item Case \freezemlLab{App}:
     \begin{equations}
  \transfreezemlfof[\Bigg]
    {\inferrule
      {\typ{\Delta; \Gamma}{M : A \to B} \\
        \typ{\Delta; \Gamma}{N : A}
      }
      {\typ{\Delta; \Gamma}{M\,N : B}}
    }
  &\Longrightarrow
  &\inferrule*[vcenter,right=\fLab{App}]
    {\hypo{\typ{\Delta;\Gamma}{\transfreezemlfof{M} : A \to B}} \\
     \hypo{\typ{\Delta;\Gamma}{\transfreezemlfof{N} : A}}
    }
    {\typ{\Delta;\Gamma}{\transfreezemlfof{M}\,\transfreezemlfof{N} : B}} \\[6ex]

    \end{equations}
\item Case \freezemlLab{Let}: In this case there are two subcases, depending on
  whether $M$ is a guarded value or not.
  \begin{itemize}
  \item $M = V \in \dec{GVal}$: In this case, we have $\mgen(\Delta,A',M) = (\Delta',\Delta')$
    for some possibly nonempty $\Delta'$, and $(\Delta,\Delta',M,A')
    \Updownarrow \forall \Delta'. A'$.  We
    proceed as follows:
\begin{gather*}
  \transfreezemlfof[\mediumginormous]
    {\raisebox{-15pt}{$\inferrule
      { \typ{\Delta, \Delta'; \Gamma}{V : A'} \\\\
        \typ{\Delta; \Gamma, x : \forall \Delta'.A'}{N : B} \\\\
        \dots
      }
      {\typ{\Delta; \Gamma}{\Let \; x = V\; \In \; N : B}}
    $}}
  \Longrightarrow  \\[2ex]
  \begin{array}{l}
    \inferrule*[right=\fLab{Let}]{
    \inferrule*[right=\fLab{PolyLam*}]{\typ{\Delta,\Delta';\Gamma}{\transfreezemlfof{V}:A'}}
    {\typ{\Delta;\Gamma}{\Lambda\, \Delta'.\transfreezemlfof{V}:\forall \Delta'.A' }}
    \\
    \typ{\Delta;\Gamma,x:\forall \Delta'.A'}{\transfreezemlfof{N}:B}}
{\typ{\Delta;\Gamma}{\Let \: x^A = \Lambda \, \Delta'. \transfreezemlfof{V}\
    \In\; \transfreezemlfof{N} : B}}
  \end{array}
\end{gather*}
where we rely on the fact that $\transfreezemlfof{V}$ is a value in System F as well,
and appeal to the derivable rule \fLab{PolyLam*}.
\item $M \notin \dec{GVal}$.  In this case, we know that $\mgen(\Delta,A,M) =
  (\cdot,\Delta')$ and $(\Delta,\Delta',M,A') = \rsubst(A') = A$ for
  some $\rsubst$ satisfying $\typ{\Delta}{\rsubst : \Delta' \Rightarrow_\mono \cdot}$.   We proceed as follows:
     \begin{gather*}
  \transfreezemlfof[\mediumginormous]
    {\raisebox{-15pt}{$\inferrule
      { \typ{\Delta,\Delta'; \Gamma}{M : A'} \\\\
        \typ{\Delta; \Gamma, x : A}{N : B} \\\\
        \dots
      }
      {\typ{\Delta; \Gamma}{\Let \; x = M\; \In \; N : B}}
    $}}
  \Longrightarrow \\[2ex]
  \begin{array}{l}
    \inferrule{
    \typ{\Delta;\Gamma}{\rsubst(\transfreezemlfof{M}):\rsubst(A')}
    \\
    \typ{\Delta;\Gamma,x:A}{\transfreezemlfof{N}:B}}
{\typ{\Delta;\Gamma}{\Let \: x^A = \transfreezemlfof{M}\
    \In\; \transfreezemlfof{N} : B}}
  \end{array}
     \end{gather*}
where we make use of a standard substitution lemma for System F to instantiate type variables
from $\Delta'$ in $\transfreezemlfof{M}$ and $A$ to obtain a derivation    of
$\typ{\Delta;\Gamma}{\rsubst(\transfreezemlfof{M}):\rsubst(A')}$, which suffices
since $A = \rsubst(A')$.  Note that $\transfreezemlfof{M}$ could
contain free type variables from $\Delta'$ since all inferred types are
translated to explicit annotations.
  \end{itemize}
\item Case \freezemlLab{Let-Ascribe}: This case is analogous to the case for
  \freezemlLab{Let}.

\end{itemize}
\end{proof}

}
\section{Type Substitutions, Environments and Well-Scoped Terms}
\label{sec:app-subst}

{

This section collects, and sketches (mostly straightforward) proofs of
properties about type substitutions, kind and type environments, and the
well-scoped term judgement.
We may then use the properties from this section without explicitly referencing them
in subsequent sections.

Note that when types appear on their own or in
contexts $\Gamma$, we identify $\alpha$-equivalent types.

We use the following notations in this and subsequent sections,
where $\Theta = ( a_1 : K_1, \dotsc, a_n : K_n )$.
Recall that this implies  all $a_i$ being pairwise different.

\begin{itemize}
\item
Let $(b : K ) \in \Theta$ hold iff $b = a_i$ and $K = K_i$ for some $1 \le i \le n$
and let $b \in \Theta$ hold iff $(b : K) \in \Theta$ holds for some $K$.
\item
For all $1 \le i \le n$, we define $\Theta(a_i) = K_i$.
\item
We define $\ftv(\Theta)$ as $(a_1, \dotsc, a_n)$.
\item
Given $\theta$ such that $\Delta \vdash \theta : \Theta \Rightarrow \Theta'$,
then $\ftv(\theta)$ is defined as $\ftv( \theta(a_1) \to \dotsc \to \theta(a_n) )$.
\item
Given $\Theta' = (b_1 : K'_1, \dotsc, b_m : K'_m)$, $\Theta' \subseteq \Theta$ holds iff
there exists a function $f$ from $\{1, \dotsc, m \}$ to $\{1, \dotsc, n \}$ such that
for all $1 \le i \le m$, we have $b_i = a_{f(i)}$ and $K'_i = K_{f(i)}$.
\item
We have $\Theta \permutation \Theta'$ iff
$\Theta \subseteq \Theta$ and $\Theta' \subseteq \Theta$.
\item
Given $\Delta = (a_1, \dotsc, a_n)$, all of the above notations are defined on $\Delta$ by
applying them to $\Theta = (a_1 : \mono, \dotsc, a_n : \mono)$.
\item Given kinds $K, K'$, we write $K \le K'$ iff $K \sqcup K' = K'$.
\end{itemize}

\begin{lemma}
\label{lem:substequality}
If $A = B$ then $\theta(A) = \theta(B)$ for any $\theta$.
\end{lemma}
\begin{proof}
  The point of this property is that alpha-equivalence is preserved by
  substitution application, because substitution application is
  capture-avoiding.  Concretely, the proof is by induction on the (equal)
  structure of $A$ and $B$.  In the case of a binder
  $A = \forall a. A' = \forall b. B' = B$, where one or both of $a,b$ are
  affected by $\theta$, alpha-equivalence implies that we may rename $a$ and $b$
  respectively to a sufficiently fresh $c$, such that $A'[c/a] = B'[c/a]$ and
  $\theta(c) = c$.  Therefore, by induction
  $\theta(A) = \theta(\forall a. A') = \theta(\forall c. A'[c/a]) = \forall
  c. \theta(A'[c/a]) = \forall c. \theta(B'[c/b]) = \theta(\forall c. B'[c/b]) =
  \theta(\forall b. B') = \theta(B)$.
\end{proof}

\begin{lemma}
\label{lem:substforall}
$\theta(\forall a.A) = \theta(\forall c. A[c/a])$, where $c \notin \ftv(\theta)
\cup \ftv(A)$
is fresh.
\end{lemma}
\begin{proof}
  This is a special case of the previous property, observing that $\forall a. A
  = \forall c. A[c/a]$ if $c$ is sufficiently fresh.
\end{proof}
\begin{lemma}
\label{lem:restriction-wellformed}
If $\Delta \vdash \theta[a \rightarrow A] : \Theta,(a : K) \Rightarrow \Theta'$,
then $\Delta \vdash \theta: \Theta \Rightarrow \Theta'$ and $\Delta,\Theta'
\vdash A : K$.
\end{lemma}
\begin{proof}
  This follows by inversion on the substitution well-formedness judgement.
\end{proof}

\begin{lemma}
\label{lem:substitution-lookup}
If $\Delta \vdash \theta : \Theta \Rightarrow \Theta'$ and $\Delta,\Theta \vdash a : K$
then $\Delta,\Theta' \vdash \theta(a) : K$.
\end{lemma}
\begin{proof}
  By induction on the structure of the derivation of $\Delta \vdash \theta :
  \Theta \Rightarrow \Theta'$.  The base case is straightforward: if $\theta$ is empty
  then $\Theta$ is also empty so $a \in \Delta$.  Moreover, $\theta(a) = a$ so
  we can conclude $\Delta,\Theta' \vdash \theta(a) : K$.  For the inductive
  case, we have a derivation of the form:
\[\inferrule
  {\Delta \vdash \theta : \Theta \Rightarrow \Theta' \\
   \Delta, \Theta \vdash A' : K'}
  {\Delta \vdash \theta[a' \mapsto A'] : (\Theta, a' : K') \Rightarrow \Theta'}
\]
There are two cases.  If $a = a'$ then the subderivation of $\Delta, \Theta
\vdash A' : K'$ proves the desired conclusion since $\theta[a'\mapsto A'](a) =
A'$ and $K = K'$.
Otherwise, $a \neq a'$ so from $\Delta,\Theta,a':K' \vdash a : K$ we can infer
that $\Delta,\Theta \vdash a : K$ as well.  So, by induction we have that
$\Delta,\Theta' \vdash \theta(a) : K$.  Since $a \neq a'$ we can also conclude that
$\Delta,\Theta' \vdash \theta[a' \mapsto A'](a) : K$, as desired.
\end{proof}

\begin{lemma}
\label{lem:stability-wf-typ}
If $\Delta, \Theta \vdash A : K$ and $\Delta \vdash \theta : \Theta \Rightarrow \Theta'$,
then $\Delta, \Theta' \vdash \theta A : K$.
\end{lemma}
\begin{proof}
  By induction on the structure of the derivation of $\Delta, \Theta \vdash A :
  K$.  The case for \textsc{TyVar} is \ref{lem:substitution-lookup}.  The cases for
  \textsc{Cons} and \textsc{Upcast} are immediate by induction.  For the
  \textsc{ForAll} case, assume the derivation is of the form:
\[  \inferrule*
    {\Delta,\Theta, a:\poly \vdash A : \poly}
    {\Delta,\Theta \vdash \forall a.A : \poly}
\]
Without loss of generality, assume $a$ is fresh and in particular not mentioned
in $\Theta,\Theta',\Delta$.  Then we can derive
$\Delta \vdash \theta[a \mapsto a] : \Theta,a:\star \Rightarrow
\Theta',a:\star$, and we may apply the induction hypothesis to conclude that
$\Delta,\Theta',a:\star \vdash \theta[a \mapsto a](A) : \poly$.  Moreover, since $a$ was
sufficiently fresh, and is unchanged by $\theta[a \mapsto a]$, we can conclude
$\Delta,\Theta' \vdash \forall a. A : \star$.
\end{proof}
\begin{lemma}
\label{lem:stability-wf-env}
If $\wfctx{\Delta, \Theta}{\Gamma}$ and $\Delta \vdash \theta : \Theta \Rightarrow
\Theta'$.
then $\wfctx{\Delta, \Theta'}{\theta\Gamma}$.
\end{lemma}
\begin{proof}
  By induction on the derivation of $\wfctx{\Delta,\Theta}{ \Gamma}$.  The base
  case is:
\[\inferrule*{\strut}{\wfctx{\Delta,\Theta}{\cdot}}\]
Moreover, it follows from $\Delta \vdash \theta : \Theta \Rightarrow \Theta'$ that $\Delta \fresh
\Theta'$, so the conclusion is immediate, since $\theta(\cdot) = \cdot$.
In the inductive case, the derivation of $\wfctx{\Delta,\Theta}{\Gamma,x:A}$ is of
the form:
\[\inferrule{\wfctx{\Delta,\Theta}{ \Gamma} \\
\Delta,\Theta \vdash A : \star\\
\forall a \in \ftv(A). (\Delta,\Theta)(a) = \mono}
{\wfctx{\Delta,\Theta}{ \Gamma,x:A}}
\]
In this case, by induction we have $\wfctx{\Delta,\Theta'}{ \theta\Gamma}$ and
using Lemma~\ref{lem:stability-wf-typ} we have $\Delta,\Theta' \vdash \theta
A : K$.  We also need to show that $\forall a \in \ftv(\theta(A))$, we have
$(\Delta,\Theta')(a) = \mono$.  There are two cases: if $a\in \Delta$ this is
immediate.  If $a \in \Theta'$, then since $a \in \ftv(\theta(A))$ we know
that there must exist $b \in \Theta$ such that $a \in \ftv(\theta(b))$ and
$b \in \ftv(A)$.  By virtue of the assumption $\forall a \in
\ftv(A). (\Delta,\Theta)(a) = \mono$, we know that $(\Delta,\Theta)(b) =
\mono$, hence $\Theta(b) = \mono$.  This implies that
$\typ{\Delta,\Theta'}{\theta(b) : \mono}$, which further implies that all the
  free type variables of $\theta(b)$, including $a$, must also have kind
  $\mono$.
Now the desired conclusion $\wfctx{\Delta,\Theta'}{ \theta(\Gamma,x:A)}$ follows.
\end{proof}

\begin{lemma}\label{lem:subst-substitution}
  \begin{enumerate}
  \item If $\typ{\Delta}{\rsubst_1 : \Delta_1 \Rightarrow_K \Delta_2}$ and
    $\typ{\Delta}{\rsubst_2 : \Delta_2 \Rightarrow_K \Delta_3}$ then
    $\typ{\Delta}{\rsubst_2 \circ \rsubst_1 : \Delta_1 \Rightarrow_K \Delta_3}$.
\item If $\typ{\Delta}{\theta : \Theta \Rightarrow \Theta'}$ and
  $\Delta' \fresh \Theta' \fresh \Delta''$ and
  $\typ{\Delta,\Theta}{\rsubst_1 : \Delta' \Rightarrow_K \Delta''}$ then
  $\typ{\Delta,\Theta'}{\theta \circ \rsubst_1 : \Delta' \Rightarrow_K \Delta''}$.
  \end{enumerate}
  \end{lemma}
  \begin{proof}
    In both cases, by straightforward induction on structure of $\rsubst_1$.
  \end{proof}

\begin{lemma}
\label{lem:ty-extend-delta}
  If $\Theta \vdash A: K$ and $\Theta' \fresh \Theta$ then $\Theta,\Theta'
  \vdash A : K$.
\end{lemma}
\begin{proof}
  Straightforward by induction on the structure of derivations of $\Theta \vdash
  A: K$.  The only subtlety is in the case for $\forall$-types, where we assume
  without loss of generality that the bound type variable $a$ is renamed away
  from $\Theta$ and $\Theta'$, so that the induction hypothesis applies.
\end{proof}
\begin{lemma}
\label{lem:substitution-extend-delta}
If $\Delta \vdash \theta : \Theta \Rightarrow \Theta'$
  and $\Delta' \fresh \Delta, \Theta'$ as well as $\Delta' \fresh \Theta$
then $\Delta, \Delta' \vdash \theta : \Theta  \Rightarrow  \Theta'$.
\end{lemma}
\begin{proof}
  By induction on the derivation of $\Delta \vdash \theta : \Theta \Rightarrow
  \Theta'$.  The base case is immediate given that $\Delta'$ is fresh for
  $\Delta$ and $\Theta'$.  For the inductive case, we have a
  derivation of the form:
\[
\inferrule*{
\Delta \vdash \theta : \Theta \Rightarrow \Theta' \\
\Delta,\Theta' \vdash A: K
}{
\Delta \vdash \theta[a \mapsto A] : (\Theta,a:K) \Rightarrow \Theta'
}
\]
By induction (since $\Delta'$ is clearly fresh for $\Delta,\Theta,$ and $\Theta'$) we
have $\Delta,\Delta' \vdash \theta : \Theta \Rightarrow \Theta'$.  Moreover, by
weakening (Lemma~\ref{lem:ty-extend-delta}) we also have
$\Delta,\Delta',\Theta' \vdash A : K$.  We can conclude, as required, that $\Delta,\Delta'
\vdash \theta[a\mapsto A] : (\Theta,a:K) \Rightarrow \Theta'$.
\end{proof}

\begin{lemma}
\label{lem:stability-subst-under-promotion}
If $\Theta_D = \meta{demote}(K, \Theta, \Delta')$
and
$\Delta \vdash \theta : \Theta_D \Rightarrow \Theta'$
then
$\Delta \vdash \theta : \Theta \Rightarrow \Theta'$.
\end{lemma}
\begin{proof}
If $K = \poly$, $\meta{demote}$ yields $\Theta = \Theta_D$ and the statement
holds immediately.

Otherwise, if $K = \mono$, we perform induction on $\Theta_D$.
By definition of $\meta{demote}$, we have $\ftv(\Theta) = \ftv(\Theta_D)$.

If $\Theta_D = \cdot$ we have $\Theta = \cdot$ and can derive the following:
\begin{mathpar}
\inferrule
  { }
  {\Delta \vdash \emptyset : \cdot \Rightarrow \Theta'}

\end{mathpar}

Let $\Theta_D = (\Theta''_D, a : K')$.
By inversion we then have

\begin{mathpar}
\inferrule
  {\Delta \vdash \fsubst : \Theta''_D \Rightarrow \Theta' \\
   \Delta, \Theta' \vdash A : K'}
  {\Delta \vdash \fsubst[a \mapsto A] :
     (\Theta''_D, a : K') \Rightarrow \Theta'}
\end{mathpar}

By $\ftv(\Theta) = \ftv(\Theta_D)$ we have $\Theta = (\Theta'', a : K'')$.
By induction this implies $\Delta \vdash \theta : \Theta'' \Rightarrow \Theta'$.

If $K' = \poly$, then by definition of $\meta{demote}$ we have $a \not\in
\Delta'$ and $K'' = \poly$.
We can then derive the following:

\begin{mathpar}
\inferrule
  {\Delta \vdash \fsubst : \Theta'' \Rightarrow \Theta' \\
   \Delta, \Theta' \vdash A : \poly}
  {\Delta \vdash \fsubst[a \mapsto A] :
       (\Theta'', a : \poly) \Rightarrow \Theta'}
\end{mathpar}

Otherwise, we have $K' = \mono$ and show that $\Delta, \Theta' \vdash A : K''$
holds.
If $K'' = \mono$, this follows immediately from $\Delta, \Theta' \vdash A : K'$.
If $K'' = \poly$, we upcast $\Delta, \Theta' \vdash A : \mono$ to $\Delta,
\Theta' \vdash A : \poly$.

In both cases for $K''$, we can then derive the following:

\begin{mathpar}
\inferrule
  {\Delta \vdash \fsubst : \Theta \Rightarrow \Theta' \\
   \Delta, \Theta' \vdash A : K''}
  {\Delta \vdash \fsubst[a \mapsto A]
      : (\Theta'', a : \mono) \Rightarrow \Theta'}
\end{mathpar}
\end{proof}

\begin{lemma}\label{lem:demote-sound}
If $\Theta' = \dec{demote}(K,\Theta,\Delta)$
then $\ftv(\Theta) = \ftv(\Theta')$ and $\typ{\Delta}{\idsubst :
  \Theta\Rightarrow \Theta'}$.
\end{lemma}
\begin{proof}
  Proof by case analysis on $K$ and induction on $\Theta$.
There are three cases.  If $K = \poly$ then the result is immediate since
  $\Theta = \Theta'$.  If $K = \mono$ and $\Theta = \cdot$ then the result is
  also immediate.
  Otherwise, if $K = \mono$ and $\Theta = \Theta_1,a:K$ then
  $\dec{demote}(K,\Theta,\Delta) = \dec{demote}(K,\Theta_1,\Delta),a:K'$, where
  $\Theta_1' = \dec{demote}(K,\Theta_1,\Delta)$ and $K'$ is $\mono$ if $a \in
  \Delta$, otherwise $K = K'$.  Then by induction we have
  $\ftv(\Theta_1) = \ftv(\Theta_1')$ and $\typ{\Delta}{\idsubst: \Theta_1
    \Rightarrow \Theta_1'}$.
Clearly, $\ftv(\Theta_1,a:K) = \ftv(\Theta_1',a:K')$.  To see that $\typ{\Delta}{\idsubst :
  \Theta\Rightarrow \Theta'}$, consider two cases: if $a \in \Delta$ then $K' =
\mono$ and we can conclude $\typ{\Delta}{\idsubst :
  \Theta,a:K\Rightarrow \Theta_1',a:\mono}$ since if $K = \poly$ then we can use
\textsc{Upcast}.  Otherwise, $K = K'$ so the result is immediate.
\end{proof}

\begin{lemma}\label{lem:demote-complete}
Let $\Delta : \Theta \Rightarrow \Theta'$ and $\Delta, \Theta \vdash A : K$
such that $\Delta, \Theta' \vdash \theta(A) : K'$ for some $K'$ with $K' \le K$.
Furthermore, let $\dec{demote}(K',\Theta,\ftv(A) - \Delta) = \Theta_D$.
Then $\typ{\Delta,\Theta_D}{A : K'}$.
\end{lemma}
\begin{proof}
  For $K' = K$, the statement follows immediately.
  Therefore, we consider only the case $K = \poly, K' = \mono$.

  We perform induction on the derivation of  $\Delta, \Theta' \vdash \theta(A) : \mono$.

  \begin{description}
  \item[Case $\theta(A) = a$:]
    \begin{mathpar}
    \inferrule*
      {a:K' \in \Theta'}
      {\Delta, \Theta' \vdash a : \mono}
    \end{mathpar}

    We have $A = b$ for some $b \in \Delta, \Theta$.
    If $b \in \Delta$, then $\Delta \vdash b : \mono$ follows immediately.
    Otherwise, we have $(b : K'') \in \Theta$ for some $K''$.
    By $b \in \ftv(A) - \Delta$, we then have $(b : \mono)$ in $\Theta_D$.

  \item[Case $\theta(A) = D \: \theta(A_1) \dots \theta(A_n)$:]
  \begin{mathpar}
  \inferrule*
    {\arity(D) = n \\\\
     \Delta, \Theta' \vdash \theta(A_1) : \mono \quad \cdots \quad \Delta,\Theta' \vdash A_n : \mono}
    {\Delta, \Theta' \vdash \tc\,\many{\theta(A)} : \mono}
  \end{mathpar}

  By induction we have $\Delta, \Theta_D \vdash \theta(A_i) : \mono$ for all $1 \le i \le n$.
  We can therefore derive $\Delta, \Theta_D \vdash \tc\,\many{A} : \mono$.

\end{description}

Note that we can disregard upcasts and $\theta(A) = \forall b.B$
as they would both yield $K' = \poly$:

\begin{mathpar}
  \inferrule*
    {\Delta, \Theta', b:\mono \vdash B : \poly}
    {\Delta, \Theta' \vdash \forall b.B : \poly}

  \inferrule*
    {\Delta, \Theta' \vdash A : \mono}
    {\Delta, \Theta' \vdash A : \poly}
\end{mathpar}

\end{proof}

The following property states the well-formedness conditions
needed in order for composition of substitutions to imply
composition of the functions induced by them.

{
\resetNum
\proofContext{lem:decompose-composition}
\begin{lemma}
Let the following conditions hold:
\begin{align}
&\Delta \vdash \theta' : \Theta \Rightarrow \Theta' \\
&\Delta \vdash \theta'' : \Theta' \Rightarrow \Theta'' \\
&\theta = \theta'' \circ \theta' \\
& \Delta, \Theta \vdash  A
\end{align}
Then $\theta(A) = \theta''\theta'(A)$ holds.

\end{lemma}
}

\begin{lemma}
\label{lem:subst-gen-inst}
If $\termwf{\Delta}{M}$, and $\Delta \vdash \theta : \Theta \Rightarrow
\Theta'$, then:
\begin{enumerate}
\item If $\ftv(A) - (\Delta, \Theta) \fresh \Theta'$ then
$\mgend[(\Delta,\Theta)]{A}{M} = \mgend[(\Delta,\Theta')]{\theta(A)}{M}$;
\item if $\Delta'' \fresh \Delta,\Theta$ and $\Delta'' \fresh \Theta'$ and
$((\Delta,\Theta), \Delta'', M, A') \Updownarrow A $ then
$((\Delta,\Theta'), \Delta'', M, \theta(A')) \Updownarrow \theta(A)$;
\end{enumerate}
\end{lemma}
\begin{proof}
  \begin{enumerate}
    \item For part 1: Observe that
    \begin{eqnarray*}
      \mgend[(\Delta,\Theta)]{A}{M} &=&\left\{\ba{ll}
                                        (\Delta',\Delta') & M \in \dec{GVal}\\
(.,\Delta') & M \notin \dec{GVal}
\ea\right.
\\
      \mgend[(\Delta,\Theta')]{\theta(A)}{M} &=& \left\{\ba{ll}
                                        (\Delta'',\Delta'') & M \in \dec{GVal}\\
(.,\Delta'') & M \notin \dec{GVal}
\ea\right.
    \end{eqnarray*}
where $\Delta' = \ftv(A) - (\Delta,\Theta)$ and $\Delta'' = \ftv(\theta(A)) -
(\Delta,\Theta')$.  So, the equation $\mgend[(\Delta,\Theta)]{A}{M} =
\mgend[(\Delta,\Theta')]{\theta(A)}{M}$ holds if and only if $\Delta' =
\Delta''$.
Suppose $a\in \Delta'$, that is, it is a free type variable of $A$ and not among
$\Delta,\Theta$.  Since $\theta$ only affects type variables in $\Theta$,
we have $\theta(a) = a$ and it
follows that $a \in \ftv(\theta(A))$.  Moreover, by assumption $\Delta' \fresh
\Theta'$ so $a \in \ftv(\theta(A)) - (\Delta,\Theta') = \Delta''$.  Conversely,
suppose $a \in \Delta''$, that is, $a$ is a free type variable of $\theta(A)$
and not among $\Delta,\Theta'$. Since $a \not\in \Delta,\Theta'$, we must have $\theta(a) =
a$ since $\theta$ was a well-formed substitution mentioning only type variables
in $\Delta,\Theta'$.  This implies that $a \in \ftv(A)$ since $a$ cannot have
been introduced by $\theta$.

We has thus shown $\Delta' \permutation \Delta''$.
To show $\Delta' = \Delta''$,
assume $a, b \in \Delta'$ such that $a$ occurs before $b$ in $\Delta'$.
This means that the first occurrence of $a$ in $A$ is before the first occurrence of $b$ in $A$.
For all $c \in \Theta$ we have $c \neq b$ and $ \ftv( \theta(c) ) \fresh b$.
Thus, the first occurrence of $a$ in $\theta(A)$ remains before the first occurrence of $b$ in $\theta(A)$.

\item For part 2: We consider two cases.
  \begin{itemize}
  \item If the derivation is of the form
    \[
    \inferrule{M \in \dec{GVal}}
    {((\Delta,\Theta),\Delta'', M, A') \Updownarrow \forall \Delta''. A'}
    \]
    then we may derive
    \[
    \inferrule{M \in \dec{GVal}}
    {((\Delta,\Theta'),\Delta'', M, \theta(A')) \Updownarrow \forall \Delta''. \theta(A')}
    \]
    by observing that since $\Delta'' \fresh \Theta$ and $\Delta'' \fresh \Theta$, we know that $\theta(\forall
    \Delta''. A') = \forall \Delta''. \theta(A')$.

  \item If the derivation is of the form
    \[
    \inferrule{\typ{\Delta,\Theta}{\rsubst : \Delta'' \Rightarrow_\mono \cdot}\\
      M \notin \dec{GVal}}
    {((\Delta,\Theta),\Delta'', M, A') \Updownarrow \rsubst(A')}
    \]
Then first we observe (by property~\ref{lem:subst-substitution}) that
$\typ{\Delta,\Theta'}{\theta \circ \rsubst : \Delta'' \Rightarrow_\mono \cdot}$, so we can derive
    \[
    \inferrule{\typ{\Delta,\Theta'}{\theta \circ \rsubst : \Delta'' \Rightarrow_\mono \cdot}\\
      M \notin \dec{GVal}}
    {((\Delta,\Theta'),\Delta'', M, \theta(A')) \Updownarrow \theta\circ \rsubst(\theta(A'))}
    \]
observing that $\theta(\rsubst(A')) = \theta \circ \rsubst(\theta(A'))$ since
$\ftv(\theta) \fresh \Delta''$.
  \end{itemize}
  \end{enumerate}
\end{proof}

\begin{lemma}
\label{lem:bijective-renaming}
Let $\Delta \vdash \theta : \Theta \Rightarrow \Theta'$ be a bijection between
the type variables in $\Theta$ and $\Theta'$.
Furthermore, let $\Theta \fresh \Delta' \fresh \Theta'$ and $\termwf{\Delta}{M}$
hold.

Then the following holds:
\begin{enumerate}
\item
\label{lem-part:bijective-renaming:typ}
If
  $\typ{\Delta,\Theta;\Gamma}{M : A}$
then
  $\typ{\Delta,\Theta';\theta(\Gamma)}{M : \theta(A)}$.
\item
\label{lem-part:bijective-renaming:principal}
If
  $\meta{principal}( (\Delta, \Theta), \Gamma, M, \Delta' ,A)$
then
  $\meta{principal}( (\Delta, \Theta'), \theta(\Gamma), M, \Delta' ,\theta(A))$.
\end{enumerate}
\end{lemma}
\begin{proof}
\begin{enumerate}
\item
For the first part of the lemma, we perform induction on $M$ and focus on the case $\Let \; x = M \;\In\ N$  .
By inversion, we have the following:

\begin{equations}
(\Delta', \Delta'') = \mgend[(\Delta,\Theta)]{A'}{M} \\
((\Delta, \Theta), \Delta'', M, A') \Updownarrow A \\
\Delta, \Theta, \Delta''; \Gamma \vdash M : A' \\
\typ{\Delta, \Theta; \Gamma, x : A}{N : B} \\
\meta{principal}((\Delta, \Theta) \Gamma, M, \Delta'', A')
\end{equations}

We assume w.l.o.g. that $\Delta'' \fresh \Theta'$. (This is justified, as per
the induction hypothesis, we may otherwise just apply an appropriate renaming
substitution.)
By induction, we then have
$\typ{\Delta, \Theta', \Delta''; \theta(\Gamma)}{M : \theta(A')}$.

By $\Theta' \fresh \Delta'' \fresh \Delta, \Theta$ and $\Delta \vdash \theta :
\Theta \Rightarrow \Theta'$, we also have $\mgend[\Delta]{M}{A'} =
\mgend[\Delta]{M}{\theta(A')} = (\Delta', \Delta'')$.
Similarly,
$((\Delta, \Theta') \Delta'', M, A') \Updownarrow \theta(A) $ holds:
If $M \in \dec{GVal}$, then $A = \forall \Delta''. A'$ and $\theta(A) = \forall
\Delta ''. \theta(A')$.
Otherwise, $A = \rsubst(A')$ for some $\rsubst$ with $\Delta \vdash \Delta''
\Rightarrow_\mono \cdot$.
Hence, the domains of $\rsubst$ and $\theta$ are disjoint, and we have
$((\Delta, \Theta') \Delta'', M, A') \Updownarrow \theta(\rsubst(A'))$.
$\meta{principal}( (\Delta, \Theta'), \theta(\Gamma), M, \Delta' ,\theta(A))$
follows directly from induction and the second part of the lemma.
Likewise, $\typ{\Delta, \Theta'; \theta(\Gamma), x : \theta(A)}{N : \theta(B)}$
follows by induction.

We have thus shown all properties needed to derive $\typ{\Delta, \Theta';
\theta(\Gamma)}{\Let\; x = M \;\In\; N : \theta(B)}$:
\[
\inferrule*
    {(\Delta', \Delta'') = \mgend[(\Delta, \Theta')]{A'}{M} \\
     ((\Delta, \Theta'), \Delta'', M, A') \Updownarrow \theta(A) \\
     \Delta, \Theta', \Delta''; \theta(\Gamma) \vdash M : \theta(A') \\
     \typ{\Delta, \Theta'; \theta(\Gamma, x : A)}{N : \theta(B)} \\\\
     \meta{principal}((\Delta, \Theta'), \Gamma, M, \Delta'', \theta(A'))
    }
    {\typ{\Delta, \Theta'; \theta(\Gamma)}{\Let \; x = M\; \In \; N : \theta(B)}}
\]

\item
To show the second part of the lemma, observe that
$\typ{\Delta, \Theta', \Delta'}{M : \theta(A)}$ and
$\Delta' = \ftv(\theta A) - \Delta, \Theta'$ follows directly from
$\meta{principal}( (\Delta, \Theta), \Gamma, M, \Delta' ,A)$ by applying the
first part of the lemma.

Further, assume
$\typ{\Delta, \Theta', \Delta_p; \theta(\Gamma)}{M : A_p}$, where $\Delta_p = \ftv(A_p) - (\Delta, \Theta')$.
Let $\rsubst_F$ be a bijective instantiation that maps variables in $\Delta_p$
to fresh ones, yielding $\Delta, \Theta' \vdash \rsubst_F : \Delta'
\Rightarrow_\mono \Delta_F$ for some appropriate $\Delta_F$.
Due to $\theta$ being a bijection, we can use its inverse $\theta^{-1}$ with
$\Delta, \Delta_F \vdash \theta^{-1} : \Theta' \Rightarrow \Theta$.

We apply the first part of the lemma to $\theta^{-1} \circ \rsubst_F$, yielding
$\typ
  {\Delta; \Theta, \Delta_F; \theta^{-1}(\rsubst_F(\theta(\Gamma)))}
  {M : \theta^{-1}(\rsubst_F(A_p)}$ and $\ftv(\theta^{-1}(\rsubst_F(A_p)) - \Delta, \Theta = \Delta_F$.
We have $\ftv(\theta\Gamma) \subseteq \Delta, \Theta' \fresh \Delta'$ and
therefore $\theta^{-1}(\rsubst_F(\theta(\Gamma))) = \Gamma$.

By $\meta{principal}( (\Delta, \Theta), \Gamma, M, \Delta' ,A)$, we then have
that there exists $\rsubst$ such that
$\Delta, \Theta \vdash \rsubst : \Delta' \Rightarrow \Delta_F$ and
$\rsubst(A) = \theta^{-1}(\rsubst_F(A_p))$.
Hence, $\rsubst_F^{-1}(\theta(\rsubst(A))) = A_p$, meaning that $\rsubst \circ
\theta \circ \rsubst_F^{-1}$ is the instantiation showing that
$\meta{principal}( (\Delta, \Theta'), \Gamma, M, \Delta' ,A)$ holds

\end{enumerate}
\end{proof}

\begin{lemma}
\label{lem:composition-strengthening}
If
  $\Delta \vdash \theta : \Theta \Rightarrow \Theta''$ and
  $\Delta \vdash \theta' : \Theta \Rightarrow \Theta'$ and
  $\Delta \vdash \theta'' : \Theta' \Rightarrow \Theta'', \Theta_E$ as well as
  $\theta = \theta'' \circ \theta'$ then
for all
  $a \in \ftv(\theta') - \Delta$ we have
  $\Delta, \Theta'' \vdash \theta''(a)$.
\end{lemma}
\begin{proof}
Via induction on $\theta'$, observing that
$\Delta \vdash \theta : \Theta \Rightarrow \Theta''$ dictates the behaviour of $\theta''$
on all  variables in the intersection of $\Theta'$ and the codomain of $\theta'$.
\end{proof}

\section{Correctness of unification proofs}


{


\subsection{Soundness of unification}

\resetNum

\thmunifysound*

\begin{proof}

Via induction on the maximum of the sizes of $A$ and $B$.
We only consider the cases where unification succeeds.

\begin{enumerate}
\item $\unify(\Delta, \Theta, a, a)$: we have $\theta = \iota_{\Delta, \Theta}$
  (identity substitution) and the result is immediate.

\item $\unify(\Delta, (\Theta, a : K'), a, A)$ or $\unify(\Delta, (\Theta, a :
  K'), A, a)$: We consider the first case; the second is symmetric. We have
\[\ba{rcl}
\dec{unify}(\Delta,(\Theta,a:K'),a,A) &=& (\Theta_1,\idsubst[a\mapsto A])\\
\dec{demote}(K',\Theta,\ftv(A) - \Delta) &=& \Theta_1\\
 \typ{\Delta,\Theta_1&}{& A : K'}
\ea\]
First, observe that $a \notin \ftv(A)$ since $a \not\in \Delta,\Theta$ and
$\ftv(\Theta_1) = \ftv(\Theta)$.  Therefore
\[\idsubst[a \mapsto A](a) = A = \idsubst[a \mapsto A](A)\]
Next, by Lemma~\ref{lem:demote-sound} we know that $\typ{\Delta}{\idsubst :
  \Theta \Rightarrow \Theta_1}$.  Moreover, by $\typ{\Delta,\Theta_1
  }{ A : K'}$ we can derive $\typ{\Delta}{\idsubst[a \mapsto A]
    :\Theta,a:K' \Rightarrow \Theta_1}$.

\item $\unify(\Delta, \Theta, D\,A_1\,\dots\,A_n, D\,B_1\,\dots\,B_n)$: we need
  to show that types under the constructor $D$ are pairwise identical after a
  substitution: $\theta(A_1) = \theta(B_1),\dots,\theta(A_n) = \theta(B_n)$,
  where $n = \arity(D)$.
  We perform a nested induction, showing that for all $0 \le j \le n+1$ the
  following holds:
  $\Delta \vdash \theta_j \vdash \Theta \Rightarrow \Theta_j$
  and
  for all $1 \le i < j$ we have
  $\theta_j(A_i) = \theta_j(B_i)$.

  For $j = 0$, this holds immediately.


  In the inductive step,
  by definition of $\meta{unify}$ we have $\theta_{j+1} = \theta' \circ \theta_j$,
  and by the outer induction
  $\theta'(A_j) = \theta'(B_j)$
   and $\Delta \vdash \theta' :
  \Theta_{j} \Rightarrow \Theta_{j+1}$.
  Together, we then have
  $\Delta \vdash \theta_{j+1} : \Theta \Rightarrow \Theta_{j+1}$.
  From Lemma~\ref{lem:substequality} we
  know that $\theta_{j+1}$ maintains equalities established by
  $\theta_j$, and so we have
  $\theta_{j+1}(A_i) = \theta_{j+1}(B_i)$ for all $1 \le i < j+1$.

  From the definition of substitution we then have
  \[
  \theta(D\,A_1\,\dots\,A_n) =
    D\,\theta(A_1)\,\dots\,\theta(A_n) = D\,\theta(B_1)\,\dots\,\theta(B_n)
    = \theta(D\,B_1\,\dots\,B_n)
  \]
  with $\Delta \vdash \theta : \Theta \Rightarrow \Theta_{n+1}$.

\item \label{un:snd:caseforall} $\unify(\Delta, \Theta, \forall a. A, \forall
  b. B)$: In this case we must have
\begin{gather}
\dec{unify}((\Delta,c),\Theta,A[c/a],B[c/b]) \; =  \; (\Theta_1,\theta) \notag\\
c \; \fresh \; \Delta,\Theta \mathlabelNum{un:snd:freshDeltaTheta} \\
\ftv(B) \fresh c \; \fresh \; \ftv(A) \mathlabelNum{un:snd:freshAB} \\
c  \; \fresh \; \ftv(\theta) \mathlabelNum{un:snd:cftvtheta}
\end{gather}
so from the inductive hypothesis we have $\theta(A[c/a]) =
  \theta(B[c/b])$ \labelNum{un:snd:thetaeq}, where $c$ is fresh and $\Delta, c
  \vdash \theta : \Theta \Rightarrow \Theta_1$.    We now derive:
\[
  \begin{array}{cll}
    &\theta(\forall a. A) \\
=
    &\theta(\forall c. A[c/a]) &\reason{by \refNum{un:snd:cftvtheta,un:snd:freshAB}, \cref{lem:substforall}} \\
=
    &\forall c. \theta(A[c/a]) & \reason{by \refNum{un:snd:freshDeltaTheta,un:snd:cftvtheta}} \\
  \end{array}
\]
  \noindent
and by exactly the same reasoning, $\theta(\forall b. B) = \forall
c. \theta(B[c/b])$.  Then by \refNum{un:snd:thetaeq} we can conclude
$\theta(\forall a. A) = \forall c. \theta(A[c/a]) = \forall c. \theta(B[c/b])
= \theta(\forall b. B)$,
  which is the desired equality, and $\Delta \vdash \theta : \Theta_1
  \Rightarrow \Theta$ because $c \notin \ftv(\theta)$ implies that we can remove
  it from $\Delta$ without damaging the well-formedness of $\theta$.

\end{enumerate}
\end{proof}

\subsection{Completeness of unification}

\fe{This is a standalone property and is similar in spirit to
\cref{lem:theta-is-id-on-non-Gamma-tyvars}.
Should we move this lemma to the previous appendix or move the
 other lemma out of it?
}

\begin{lemma}[Unifiers are surjective]
\label{lem:unify-surjective}
Let $\unify(\Delta, \Theta, A, B) = (\Theta', \theta)$.
Then $\ftv(\Theta') \subseteq \ftv(\Theta)$ and for all
$b \in \Theta'$ there exists $a \in \Theta$ such that $b \in \ftv(\theta(a))$.
\end{lemma}

\begin{proof}
The first part follows immediately from the fact that in each case $\Theta'$,
is always constructed from $\Theta$ by removing variables or demoting them.

For the second part, observe that $\theta'$ is constructed by manipulating
appropriate identity functions.
Mappings are only changed in the cases $(a, A)$ and $(A,a)$,
such that $\theta(a) = A$.
However, at the same time, $a$ is removed from the output.

\end{proof}

\resetNum

\thmunifycomplete*

\proofContext{unify:compl}
\begin{proof}
Via induction on the maximum of the sizes of $A$ and $B$.

\begin{enumerate}

\item Case $A = a = B$:
In this case $\unify(\Delta, \Theta, a, a)$ succeeds and returns
$(\Theta,\idsubst_{\Delta,\Theta})$.  Moreover, we may choose $\theta'' = \theta$ and
conclude that $\Delta \vdash \theta : \Theta \Rightarrow \Theta'$ and $\theta =
\theta \circ \idsubst_{\Delta, \Theta}$, as desired.

\item Case $A = a \neq B$ or $B = b \neq A$.  The two cases where one side is a
  variable are symmetric; we consider $A = a \neq B$.

  Since
  $\theta(a) = \theta(B)$ for $B \neq a$, we must have that $a \in \Theta$.
  Thus, $\Theta = \Theta_1',a:K'$ for some kind $K'$ such that $K' \le K$
  (due to assumption $\Delta, \Theta \vdash A : K$).
  Also, since types are
  finite syntax trees we must have $a \neq \ftv(B)$~\labelNum{a-ftv-B}.
  By assumption $\Delta \vdash \theta : \Theta \Rightarrow \Theta'$,
  we have $\theta(a) : K'$ and by $\theta(a) = \theta(B)$ therefore also
  $\Delta, \Theta' \vdash \theta(B) : K'$~\labelNum{theta-on-B-K'}.

 We now define
  $\Theta_1 = \dec{demote}(K',\Theta_1',\ftv(B) - \Delta)$ and
   choose $\theta''$ to agree with $\theta$ on $\Theta_1$, and undefined on
  $a$, yielding $\Delta \vdash \theta'' : \Theta'_1 \Rightarrow \Theta'$%
  ~\labelNum{theta''-wf-weak}.
  By $\refNum{a-ftv-B}$ we then have $\theta''(B) = \theta(B)$,
  making \refNum{theta-on-B-K'} equivalent to
  $\Delta, \Theta' \vdash \theta''(B) : K'$.
  We apply \cref{lem:demote-complete}, yielding $\Delta, \Theta_1 \vdash B : K'$

  Hence unification succeeds in this case with
  $\dec{unify}(\Delta,\Theta,a,B) = (\Theta_1,\idsubst[a \mapsto B])$.

  We strengthen \refNum{theta''-wf-weak} to
  $\typ{\Delta}{\theta'' : \Theta_1 \Rightarrow \Theta'}$ by
  observing that for each $b \in \ftv(B) - \Delta$ (i.e., those variables
  potentially demoted to $K'$ in $\Theta_1$),
  we have $\Delta, \Theta' \vdash \theta(b) : K'$.
  If $K' = \poly$ we have $K' = K$ by $K' \le K$ and
  $\Delta, \Theta' \vdash \theta(b) : K'$ follows immediately. Otherwise,
  if $K ' = \mono$, then due to $b \in \ftv(B)$, $\theta(b)$ occurs
  in $\theta(B)$, and $\theta(b) : \poly \ge K'$ would violate $\refNum{theta-on-B-K'}$.

  Clearly, $\theta'' \circ (\idsubst[a \mapsto B]) =
  (\theta'' \circ \idsubst)[a \mapsto \theta''(B)] = \theta$ since $\theta''$
  agrees with $\theta$ on all variables other than $a$, and $a \not\in \ftv(B)$
  as well as $\theta(a) = \theta(B)$.

\item $\theta(D\,A_1\,\dots\,A_n) = \theta(D\,B_1\,\dots\,B_n)$: by definition
  of substitution we have $\theta(A_i) = \theta(B_i)$, where $i \in 1,\dots,n$
  and $n \ge 0$.
  We perform a nested induction, showing that for all $0 \le j \le n+1$ the
  following holds:
  We have $\Delta \vdash \theta_j : \Theta \Rightarrow \Theta_j$~\labelNum{D:theta-j-wf} and
  there exists $\theta''_j$ such that
    $\Delta \vdash \theta''_j : \Theta_n \Rightarrow \Theta'$ and
    $\theta''_j \circ \theta_j = \theta$~\labelNum{D:theta''-n}
  as well as
  for all $1 \le i < j$
  unification of $\theta_i(A_i)$ and  $\theta_i(B_i)$ succeeds.

  \begin{enumerate}
  \item $j = 0$: unification succeeds with $\theta' = \theta_1 = \iota$ and the
    theorem holds for $\theta'' = \theta$ and $\Theta'' = \Theta$.
  \item $j \ge 1$:
    We use \refNum{D:theta''-n} to obtain
    $\theta''_j(\theta_j (A_j)) = \theta(A)$~\labelNum{D:theta-A-split} and
    $\theta''_j(\theta_j (B_j)) = \theta(B)$~\labelNum{D:theta-B-split}.

    We then have
    $$(\Theta_{j+1}, \theta'_{j+1}) = \unify(\Delta, \Theta_j, \theta_j(A_j),
    \theta_j(B_j))$$ and $\theta_{j+1} = \theta'_{j+1} \circ \theta_j$
    (by definition of $\unify$).

    By \refNum{D:theta-A-split,D:theta-B-split,D:theta-j-wf} the outer induction shows that
    unification of $\theta_j(A_j)$ and $\theta_j(B_j)$ succeeds and
    there exists $\Delta \vdash \theta'' : \Theta_{j+1} \Rightarrow \Theta'$
    such that $\theta'' \circ \theta'_{j+1} = \theta''_j$~\labelNum{D:theta''-j}.
    By \cref{thm:unification-sound}, we have
    $\Delta \vdash \theta'_{j+1} : \Theta_j \Rightarrow \Theta_{j+1}$
    and hence by composition also $\Delta \vdash \theta_{j+1} : \Theta \Rightarrow \Theta_{j+1}$.
    Further, by \refNum{D:theta''-n,D:theta''-j},
    we have
    \[
      (\theta'' \circ \theta'_{j+1}) \circ \theta_j =
      \theta''_j \circ \theta_j =
      \theta
    \]
    Choosing
    $\theta''_{j+1} = \theta''$
    then satisfies
    $\theta''_{j+1} \circ \theta_{j+1} = \theta$
    and $\Delta  \vdash \theta''_{j+1} : \Theta_{j+1} \Rightarrow \Theta'$.
  \end{enumerate}

\item $\theta(\forall a. A) = \theta(\forall b. B)$: we take fresh $c \notin
  \ftv(\theta, A, B)$.  By Lemma~\ref{lem:substforall} and definition of
  substitution we have $\theta(A[c/a]) = \theta(B[c/b])$.  By induction
  $\unify((\Delta, c), \Theta, A[c/a], B[c/b])$ succeeds with $(\Theta_1,
  \theta')$ and there exist $\theta''$ such that
  $\theta = \theta'' \circ \theta'$~\labelNum{forall:comp-theta''F} and
  ${\Delta,c \vdash \theta'' : \Theta_1 \Rightarrow \Theta'}$%
  ~\labelNum{forall:theta''-wf}.
  The latter implies $c \not\in \Delta,\Theta'$.
  By \refNum{forall:comp-theta''F} and $c \not\in \ftv(\theta)$ we
  have $c \not \in \ftv(\theta')$.

  This means that
  $\unify(\Delta, \Theta, \forall a. A, \forall b. B)$ succeeds with
  $(\Theta_1, \theta')$

  We strengthen \refNum{forall:theta''-wf} to
  ${\Delta \vdash \theta'' : \Theta_1 \Rightarrow \Theta'}$ by
  showing that $c \not\in \ftv(\theta'')$.
  Hence, assume $e \in \Theta$ such that $c \in \ftv(\theta''(e))$.
  By \cref{lem:unify-surjective}, there exists $f \in \Theta$ such that
  $e \in \ftv(\theta'(f))$.
  This would imply $c \in \ftv(\theta''(\theta'(e))$, which by
  \refNum{forall:comp-theta''F} contradicts $\Delta \vdash \theta : \Theta
  \Rightarrow \Theta'$ and $c \not\in \Delta,\Theta'$.

\end{enumerate}
\end{proof}

}

\section{Correctness of type inference algorithm}
\label{sec:correctness-type-inference-algo}


\setlength{\jot}{2pt}

This section contains proofs of correctness of the type inference algorithm.
All of the properties (lemmas and theorems) in this appendix are parameterised
by a single term $M$ and we prove them correct simultaneously by induction on
the structure of $M$.
As the proof is by induction on the structure of \emph{terms}, there is no need
to concern ourselves about what it would mean to perform induction on the
structure of \emph{derivations} in light of the negative occurrence of the
typing relation in the principal type restriction (as discussed in
Section~\ref{part:principal-type-restriction}).

The dependencies between different proofs are shown in
\cref{fig:proof-dep-graph}.
A straight arrow $P \xrightarrow{\phantom{........}} Q$ denotes a direct
dependency in which for any term $M$, the property $Q[M]$ depends on $P[M]$.
A dashed arrow $P \xdashrightarrow{\phantom{........}} Q$ denotes a decreasing
dependency in which for any term $M$, the property $Q[M]$ depends only on
$P[M']$ where $M'$ is a \emph{strict} subterm of $M$.
%
%
All cycles in \cref{fig:proof-dep-graph} include a dashed arrow, ensuring that
all properties depend only on one another in a well-founded way.

\begin{figure}[ht]
\begin{tikzpicture}[node distance=3.5cm]

\tikzstyle{normal}=[latex'-]
\tikzstyle{onsub}=[latex'-, dashed]


\node (I1-header) {\textbf{\Cref{subsec:correctness-inference:principality}}};
\node[right=of I1-header] (I2-header) {\textbf{\Cref{subsec:correctness-inference:soundness}}};
\node[right=of I2-header] (I3-header) {\textbf{\Cref{subsec:correctness-inference:completeness-mg}}};

\node[matrix, row sep=1.25cm, below=0.5cm of I1-header] (I-1)
{
\node
   (lem-inferred-types-are-principal)
   {\cref{lem:inferred-types-are-principal}}; \\
\node
  (lem-inferred-types-vs-principal-types)
  {\cref{lem:inferred-types-vs-principal-types}}; \\
\node
  (lem-stability-principality-substitution)
  {\cref{lem:stability-principality-substitution}}; \\
\node
  (lem-subst-gamma-and-A)
  {\cref{lem:subst-gamma-and-A}}; \\
\node
  (lem-shrink-context)
  {\cref{lem:shrink-context}}; \\
\node
  (lem-extend-tv-context-typing)
  {\cref{lem:extend-tv-context-typing}}; \\
\node
  (lem-principal-rename-gen-vars)
  {\cref{lem:principal:rename-gen-vars}}; \\
\node
  (lem-bijection-on-principal-types-vars)
  {\cref{lem:bijection-on-principal-types-vars}}; \\
};

\node[matrix, row sep=4cm,anchor=center] (I-2) at ( I-1.east -| I2-header)
{
\node
  (lem-theta-is-id-on-non-Gamma-tyvars)
  {\cref{lem:theta-is-id-on-non-Gamma-tyvars}}; \\
\node
  (thm-inference-sound)
  {\cref{thm:inference-sound}}; \\
};

\node[matrix,anchor=center] (I-3) at ( I-1.east -| I3-header)
{
\node
  (thm-inference-completeness-mg)
  {\cref{thm:inference-completeness-mg}};\\
};


\draw
  (lem-inferred-types-are-principal) edge[normal, bend left] (thm-inference-completeness-mg)
  (lem-inferred-types-are-principal) edge[normal] ($(thm-inference-sound.north)+(-2mm,0)$);

\draw
  (lem-inferred-types-vs-principal-types) edge[normal] ($(thm-inference-sound.north)+(-4mm,0)$);

\draw
  (lem-inferred-types-vs-principal-types.west) edge[bend right, -latex'] (lem-shrink-context.west)
  (lem-inferred-types-vs-principal-types.east) edge[normal, bend angle=30, bend left] ($(thm-inference-completeness-mg.north)+(-3mm,0)$);

\draw
  (lem-stability-principality-substitution) edge[normal, bend angle=30, bend left] ($(thm-inference-completeness-mg.north)+(-5mm,0)$)
  (lem-stability-principality-substitution) edge[normal, bend angle=10, bend left] ($(thm-inference-sound.north west)+(2mm,0)$)
  (lem-stability-principality-substitution) edge[normal] (lem-inferred-types-vs-principal-types);

\draw
  (lem-subst-gamma-and-A) edge [normal, onsub] (lem-stability-principality-substitution);

\draw
  (lem-shrink-context) edge[onsub] (lem-subst-gamma-and-A)
  (lem-shrink-context.south) ++(-2mm,0) edge[onsub] ($(lem-extend-tv-context-typing.north)+(-2mm,0)$);

\draw
  (lem-extend-tv-context-typing.north) ++(2mm,0) edge[onsub] ($(lem-shrink-context.south)+(2mm,0)$);


\draw
  (lem-bijection-on-principal-types-vars.west) edge[normal, bend left] (lem-subst-gamma-and-A.west);



\draw
  (thm-inference-sound) edge[onsub] (lem-theta-is-id-on-non-Gamma-tyvars)
  (thm-inference-sound.west) ++(0,-1mm) edge[onsub] (lem-subst-gamma-and-A)
  (thm-inference-sound.north) ++(-1mm,0) edge[onsub, bend angle=20, bend right] (lem-inferred-types-are-principal)
  (thm-inference-sound.north west) ++(0,0) edge[onsub, bend angle=10, bend left]  ($(lem-stability-principality-substitution.south)+(1mm,0)$)
  (thm-inference-sound.west) ++(0,-2mm) edge[onsub] (lem-principal-rename-gen-vars);

\draw
  (thm-inference-completeness-mg) edge[onsub] (thm-inference-sound)
  (thm-inference-completeness-mg) edge[onsub] ($(lem-bijection-on-principal-types-vars.east)+(0,-1mm)$);

\node[matrix, draw, below right=-1.55cm and 6.7cm of lem-bijection-on-principal-types-vars.center, outer sep=0mm, inner sep=1.5mm, overlay]
{
  \node (P1) {$P$}; &[5mm] \node (Q1) {$Q \;:$}; & \node{for all $M$, $Q[M]$ depends on $P[M]$}; \\
  \node (P2) {$P$}; & \node (Q2) {$Q \;:$}; & \node{\begin{minipage}{5cm}for all $M$, $Q[M]$ depends on $P[M']$, \\where $M'$ is a strict subterm of $M$\end{minipage}}; \\
};

\draw (Q1) edge[normal] (P1);
\draw (Q2) edge[onsub]  (P2);

\end{tikzpicture}
\caption{Dependencies between properties in
\Cref{sec:correctness-type-inference-algo}}
\label{fig:proof-dep-graph}
\end{figure}
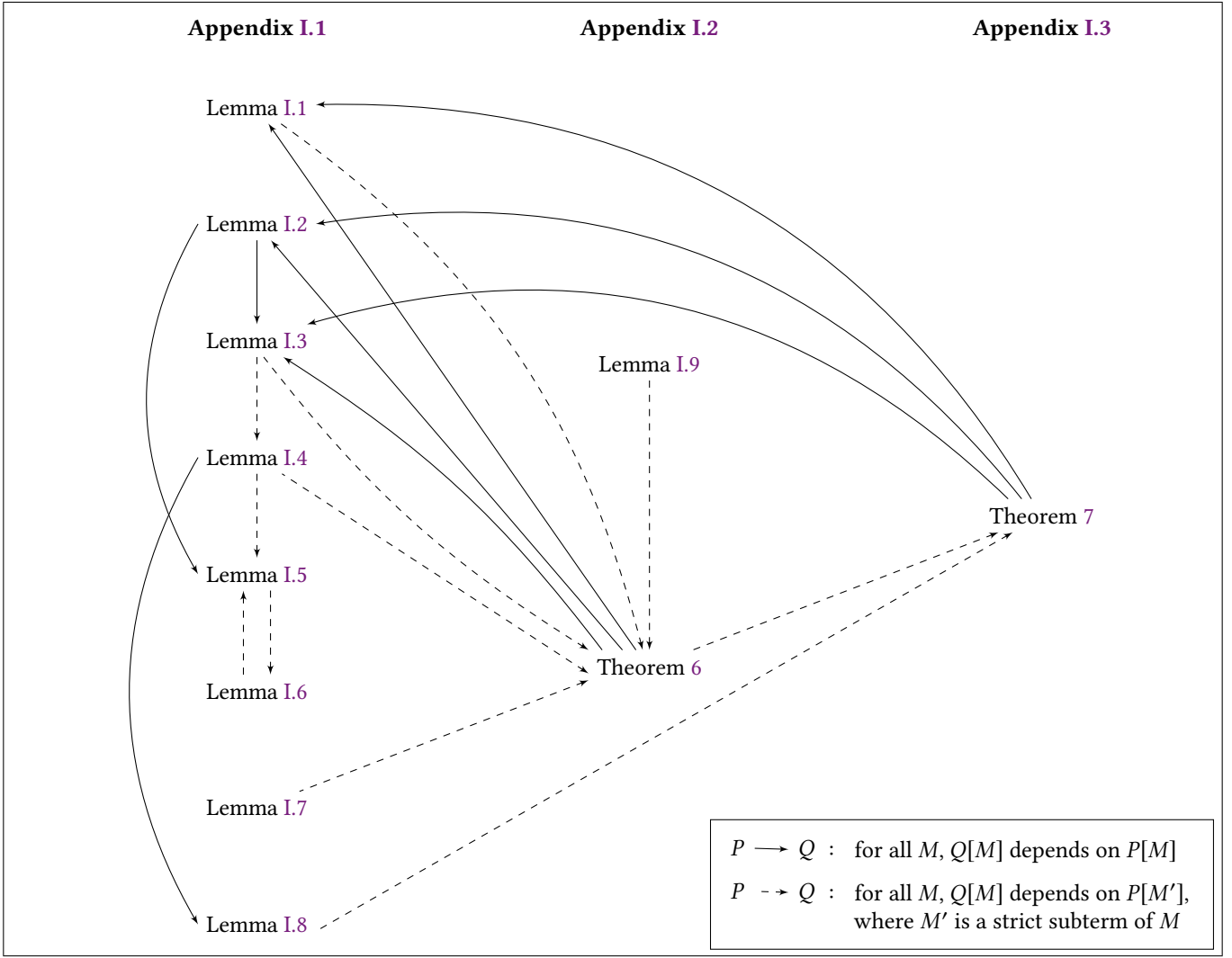

\subsection{Principality}
\label{subsec:correctness-inference:principality}

In this subsection we collect together proofs of properties related to
principality.

{

\resetNum
\proofContext{inferred-types-are-principal}

\begin{lemma}[Inferred types are principal]
\label{lem:inferred-types-are-principal}
If
  $\Infer(\Delta, \Theta,\Gamma,M) = (\Theta', \theta, A)$ and
  $\termwf{\Delta}{M}$ and
  $\wfctx{\Delta, \Theta}{\Gamma}$
then
  $\meta{principal}((\Delta,\Theta' - \Delta'), \theta \Gamma,\Delta',A)$ holds, where
  $\Delta' = \ftv(A) - \Delta - \ftv(\theta)$.
\end{lemma}
\resetNum
\proofContext{inferred-types-are-principal}
\begin{proof}
By \cref{thm:inference-sound} we have
$\Delta \vdash \theta : \Theta \Rightarrow \Theta'$~\labelNum{theta-wf}
and $\typ{\Delta, \Theta'; \theta(\Gamma)}{M : A}$~\labelNum{typ-M}.
The latter implies $\Delta, \Theta' \vdash A$ and
hence $\Delta' \subseteq \Theta'$.
We can therefore rewrite \refNum{typ-M} as
$\typ{\Delta, (\Theta' - \Delta'), \Delta'; \theta(\Gamma)}{M : A}$,
satisfying the first condition of
$\meta{principal}((\Delta,\Theta' - \Delta'), \theta \Gamma,\Delta',A)$.

By definition of $\Delta'$, we have $\Delta' \fresh \ftv(\theta)$.
We can therefore strengthen \refNum{theta-wf} to
  $\Delta \vdash \theta : \Theta \Rightarrow \Theta' - \Delta'$%
  ~\labelNum{theta-wf-str}

Let $\Delta_p, A_p$ such that
  $\Delta_p = \ftv(A_p) - (\Delta, \Theta' - \Delta')$ and
  $\typ{\Delta, (\Theta' - \Delta'), \Delta_p}{M : A_p}$%
  ~\labelNum{typ-M-Ap}.
The latter implies $\Delta_p \fresh \Delta, \Theta' - \Delta'$ and
we can weaken \refNum{theta-wf-str} to
  $\Delta \vdash \theta :
    \Theta \Rightarrow (\Theta' - \Delta'), \Delta_p$~\labelNum{theta-wf-str-wk}.

Hence, we can apply \cref{thm:inference-completeness-mg}, to
\refNum{theta-wf-str-wk,typ-M-Ap}, stating that there exists $\theta''$ s.t.\
$\Delta \vdash \theta'' : \Theta' \Rightarrow (\Theta' - \Delta'), \Delta_p$
and $\theta''(A) = A_p$ and $\theta = \theta'' \circ \theta$.

The latter implies that for all $a \in \ftv(\theta)$, $\theta''(a) = a$ must
hold. Hence, by defining $\rsubst$ as a restriction of $\theta''$ such that
$\rsubst(a) = \theta''(a)$ for all
$a \in \ftv(A) - \Delta - \ftv(\theta)$ (i.e., $\Delta'$),
we get $\Delta \vdash \rsubst : \Delta' \Rightarrow_\poly (\Theta' - \Delta'), \Delta_p$
and maintain $\rsubst(A) = A_p$. We rewrite the former to
$\Delta, (\Theta' - \Delta') \vdash \rsubst : \Delta' \Rightarrow_\poly  \Delta_p$,
obtaining an instantiation as required by the definition of $\meta{principal}$.
\end{proof}

\resetNum
\proofContext{inferred-types-vs-principal-types}
\begin{lemma}[Inferred types and principal types are isomorphic]
\label{lem:inferred-types-vs-principal-types}
Let the following conditions hold:
\begin{align}
  &\wfctx{\Delta, \Theta}{\Gamma} \mathlabelNum{wftcx}\\
  &\termwf{\Delta}{M} \mathlabelNum{termwf} \\
  &\Delta' \fresh \Theta' \\
  &\meta{principal}((\Delta, \Theta),\Gamma,M, \Delta',A)
    \mathlabelNum{principal}\\
  &\Infer(\Delta, \Theta, \Gamma, M) = (\Theta', \theta,A')
    \mathlabelNum{infer}\\
  &\Delta'' = \ftv(A') - \Delta - \ftv(\theta) \mathlabelNum{Delta''}
\end{align}
Then there exists $\rsubst$ such that
  $\Delta, (\ftv(\theta) - \Delta) \vdash
    \rsubst : \Delta'' \Rightarrow_\mono \Delta'$
and
  $\rsubst(\Delta'') = \Delta'$ and $\rsubst(A') = \theta(A)$.
\end{lemma}
\begin{proof}
By definition of $\meta{principal}$ we have
  $\typ{\Delta, \Theta, \Delta' ; \Gamma}{M : A}$~\labelNum{typ-M} and
  $\Delta' = \ftv(A) - \Delta, \Theta$.

Applying \cref{thm:inference-sound} to \refNum{infer}, we get
$\Delta \vdash \Theta \Rightarrow \Theta'$ and
$\typ{\Delta, \Theta';\theta(\Gamma)}{M : A'}$~\labelNum{typ-M-A'}.

We have
  $\Delta \vdash \idsubst_{\Delta, \Theta} : \Theta \Rightarrow \Theta$
and therefore by $\Delta' \mathop{\#} \Theta$ and weakening also
  $\Delta \vdash \idsubst_{\Delta, \Theta} :
    \Theta \Rightarrow \Theta, \Delta'$~\labelNum{idsubst-weakened}.
Trivially, we can rewrite \refNum{typ-M} and \refNum{principal} as
  $\typ
    {\Delta, \Theta, \Delta' ; \idsubst_{\Delta, \Theta} (\Gamma)}
    {M : A}$~\labelNum{typ-M-substed}
and
  $\meta{principal}(
    (\Delta, \Theta), \allowbreak
    \idsubst_{\Delta, \Theta}(\Gamma),
    M,  \allowbreak
    \Delta', \allowbreak
    A)$~\labelNum{principal-idsubsted},
respectively.
We can apply \cref{thm:inference-completeness-mg},
using
\refNum{%
   termwf,%
   idsubst-weakened,%
   typ-M-substed%
},
which yields existence of $\theta''$ such that
  $\Delta \vdash \theta'' : \Theta' \Rightarrow \Theta, \Delta'$ and
  $\idsubst_{\Delta, \Theta} = \theta'' \circ \theta$~\labelNum{composition} and
  $\theta''(A') = A $~\labelNum{A'-vs-A}.
The latter implies that $\theta''$ maps the type variables from $\Delta''$
surjectively into $\Delta'$~\labelNum{theta''-surjective}.

Let $\Theta_{\theta} = \ftv(\theta)$.  By \refNum{composition}, we then have
$\Delta' \fresh \Theta_{\theta}$ and $\theta$ is a bijection from $\Theta$ to
$\Theta_{\theta}$. Conversely, the restriction of $\theta''$ to
$\Theta_{\theta}$ is a bijection from $\Theta_{\theta}$ to $\Theta$.

We can therefore apply
\cref{lem:bijective-renaming}(\ref{lem-part:bijective-renaming:principal}) and
obtain
  $\meta{principal}(
    (\Delta, \Theta_{\theta}),
    \theta(\Gamma),
    M,
    \Delta',
    \theta(A))$~\labelNum{principal-substed}.

By \cref{lem:shrink-context} and $\wfctx{\Delta, \Theta_{\theta}}{\theta(\Gamma)}$
as well as $\Delta, \Theta_{\theta}, \Delta'' \subseteq \ftv(A')$,
we can strengthen \refNum{typ-M-A'} to
  $\typ{\Delta, \Theta_{\theta}, \Delta'';\theta(\Gamma)}{M : A'}$%
    ~\labelNum{typ-M-A'-strengthened}.

We have
  $\ftv(\theta(A)) -(\Delta, \Theta_{\theta}) =
   \Delta' =
   \ftv(A) - (\Delta, \Theta)$.
By definition of $\meta{principal}$, \refNum{principal-substed,typ-M-A'-strengthened} imposes
that there exists $\rsubst_I$ such that
$\Delta, \Theta_{\theta} \vdash \rsubst_I : \Delta' \Rightarrow_\mono \Delta''$ and
$\rsubst_I(\theta(A)) = A'$.

Using \refNum{A'-vs-A}, we rewrite the latter to
\begin{equation}
\rsubst_I(\theta(\theta''(A')) = A' \mathlabelNum{roundabout}
\end{equation}
This implies that $\theta''$ maps $\Delta''$ not only surjectively (cf.\
\refNum{theta''-surjective}), but bijectively into $\Delta'$.
By \refNum{A'-vs-A} we further have $\theta''(\Delta'') = \Delta'$ (i.e., the
order of variables is preserved).

Since $\theta$ is the identity on $\Delta'$, $\rsubst_I$ must be the inverse of
$\theta''$ on $\Delta'$.  Hence, we define $\rsubst$ such that $\rsubst(a) =
\theta''(a)$ for all $a \in \Delta''$, yielding $\Delta, \Theta_{\theta} \vdash
\rsubst : \Delta'' \Rightarrow_\mono \Delta'$.  As the inverse of $\rsubst_I$,
applying $\rsubst$ to both sides of \refNum{roundabout} yields
$\theta(\theta''(A')) = \theta(A) = \rsubst(A')$, which is the desired property.

\end{proof}
}

{
\proofContext{stability-principality-substitution}
\resetNum
\begin{lemma}[Stability of principality under substitution]
\label{lem:stability-principality-substitution}
Let the following conditions hold:
\begin{align}
&\wfctx{\Delta, \Theta}{\Gamma} \mathlabelNum{ctx} \\
&\Delta' \fresh \Theta' \mathlabelNum{Theta'-vs-Delta'} \\
&\termwf{\Delta}{M} \mathlabelNum{M-wf}\\
&\Delta \vdash \theta : \Theta \Rightarrow \Theta' \mathlabelNum{theta}\\
&\meta{principal}((\Delta, \Theta), \Gamma, M, \Delta',A) \mathlabelNum{principal}
\end{align}

Then $\meta{principal}((\Delta, \Theta'), \theta\Gamma, M, \Delta',\theta A)$ holds.
\end{lemma}

\begin{proof}
By definition of $\meta{principal}$, we have
  $\typ{\Delta,\Delta',\Theta; \Gamma}{M : A}$ and
  $\Delta' = \ftv(A) - \Delta, \Theta~\labelNum{Delta'}$.

By \refNum{Theta'-vs-Delta'}, we can weaken \refNum{theta} to
  $\Delta, \Delta' \vdash \Theta \Rightarrow \Theta'$.
Together with the latter, we can then apply
\cref{lem:stability-wf-typ}  and obtain
  $\Delta, \Delta', \Theta' \vdash \theta A$.

Let $\Delta'' = \ftv(\theta A) - \Delta, \Theta'$.
By \refNum{Theta'-vs-Delta',theta,Delta'}, \cref{lem:subst-gen-inst} yields
$\Delta' = \Delta''$~\labelNum{Delta'-vs'-Delta''}.

Let $A_p$ and $\Delta_p$ such that $\Delta_p = \ftv(A_p) - \Delta, \Theta'$ and $\typ{\Delta,\Theta',\Delta_p; \theta \Gamma}{M : A_p}$~\labelNum{M-A-p}.
Our goal is to  show that there exists $\rsubst$ such that $\Delta,\Theta' \vdash \rsubst : \Delta'' \Rightarrow \Delta_p$ and $\rsubst(\theta A) = A_p$.


We  weaken \refNum{theta} to $\Delta \vdash \theta : \Theta \Rightarrow \Theta',\Delta_p$.
We can then apply  \cref{thm:inference-completeness-mg} to \refNum{M-A-p},
which states that $\Infer(\Delta, \Theta, \Gamma, M)$ returns $(\Theta'',\theta',A')~\labelNum{infer}$
and there exists $\theta''$ such that
\begin{align}
&\Delta \vdash \theta'' : \Theta'' \Rightarrow \Theta',\Delta_p \mathlabelNum{theta''} \\
&\theta = \theta'' \circ \theta' \mathlabelNum{theta-composition} \\
&\theta''(A') = A_p \mathlabelNum{A'-vs-A-p}
\end{align}
By definition of $\Infer$, all type variables in $\Theta''$ but not in $\Theta$ are fresh, which implies $\Delta' \fresh \Theta''$~\labelNum{Delta'-fresh-Theta''}.

By \cref{thm:inference-sound}, we have $\Delta \vdash  \theta' : \Theta \Rightarrow \Theta''$~\labelNum{theta'} and $\Delta,\Theta'' \vdash A'$~\labelNum{A'}.

We split $\Delta_p$ into a (possibly empty) part that is contained in $\Delta'$ and a remaining part that is not.
Concretely, let $\Delta'_p, \Delta''_p$ such that $\Delta_p \permutation (\Delta'_p, \Delta''_p)$ and $\Delta'_p \subseteq \Delta'$ and $\Delta''_p \fresh \Delta' $.
We weaken \refNum{theta'}, \refNum{theta''}, and \refNum{theta}, respectively:
\begin{align}
&\Delta, \Delta' \vdash \theta' : \Theta \Rightarrow \Theta'' \mathlabelNum{theta'-weakened} \\
&\Delta, \Delta' \vdash \theta'' : \Theta'' \Rightarrow \Theta',\Delta''_p \mathlabelNum{theta''-weakened} \\
&\Delta, \Delta' \vdash \theta : \Theta \Rightarrow \Theta',\Delta''_p \mathlabelNum{theta'-weakened-for-comp}
\end{align}

Let $\Delta''' \coloneqq \ftv(A') - \Delta - \ftv(\theta')$~\labelNum{Delta'''},
which implies $\Delta''' \subseteq \Theta''$~\labelNum{Delta'''-subseteq} (using \refNum{theta',A'}).
Further, let $\Theta_{\theta'} = \ftv(\theta') - \Delta$.

By
\refNum{%
  M-wf,%
  ctx,%
  Delta'-fresh-Theta'',%
  principal,%
  infer,%
  Delta'''%
},
\cref{lem:inferred-types-vs-principal-types} yields
existence of $\rsubst_b$ such that
$\Delta, \Theta_{\theta'} \vdash \rsubst_b : \Delta''' \Rightarrow_\mono \Delta'$~\labelNum{rsubst-b} and
$\rsubst_b(\Delta''') = \Delta'$ \labelNum{rsubst-b-Delta'-vs-Delta'''} and
$\rsubst_b(A') = \theta' A$ \labelNum{rsubst-b-A-vs-A'}.

Let $\Delta' = (a_1, \dotsc, a_n)$ and $\Delta''' = (b_1, \dotsc, b_n)$.
Let $\rsubst$ be defined such that
for all $1\le i \le n$, $\rsubst(a_i) = \theta''(b_i)$.
By \refNum{Delta'''-subseteq,theta'',Delta'-vs'-Delta''} this yields $\Delta, \Theta' \vdash \rsubst : \Delta'' \Rightarrow_\poly \Delta_p$.

Next, we show $\rsubst \theta''\rsubst_b(A') = \theta''(A')$~\labelNum{theta''-renaming-trick}:
To this end, we show that for each $a \in \ftv(A')$ we have $\rsubst \theta''\rsubst_b(a) = \theta''(a)$.
By the definition of $\Delta'''$ (cf. \refNum{Delta'''}) and \refNum{A'}, we have $\ftv(A') \subseteq \Delta, \Delta''', \Theta_{\theta'}$.

We consider three cases:
\begin{description}
\item[Case 1 $a = b_i \in \Delta'''$:]
We have $\rsubst_b(b_i) = a_i $ by \refNum{rsubst-b-Delta'-vs-Delta'''}.
By $a_i \in \Delta'$ and \refNum{theta''-weakened} we have $\theta''(a_i) = \theta''(\rsubst_b(b_i)) = a_i$.
By definition of $\rsubst$, we have $\rsubst(a_i) = \rsubst(\theta''(\rsubst_b(b_i))) = \theta''(b_i)$.
\item[Case 2 $a \in \Theta_{\theta'}$:]
We have $a \not\in \Delta'''$ and therefore $\rsubst_b(a) = a$ by \refNum{rsubst-b}.
By \refNum{theta'}, we have $\Theta_{\theta'} \subseteq \Theta''$.
Applying \cref{lem:composition-strengthening} to \refNum{theta,theta-composition} yields
 $\theta''(a) = \theta''(\rsubst_b(a)) = A $ for some $A$ with $\Delta, \Theta' \vdash A$.
By $\Delta'' = \Delta' \fresh \Delta,\Theta'$ we then have $\rsubst(A) = A$.
In total, this yields $\rsubst \theta''\rsubst_b(a) = \theta''(\rsubst_b(a)) = \theta''(a)$.
\item[Case 3 $a \in \Delta$:]
We have $\rsubst(a) = a$, $\rsubst_b(a) = a$ and $\theta''(a) = a$.
This immediately yields $\rsubst \theta''\rsubst_b(a) = \theta''(a) = a$.
\end{description}

Finally, we show $\rsubst(\theta A) = A_p$:
\begin{equation*}
\begin{array}{cll}
   &\rsubst \theta(A) \\
  = &\rsubst \theta''\theta'(A)  &\reason{by \refNum{theta-composition,theta'-weakened,theta''-weakened,theta'-weakened-for-comp}}  \\
  = &\rsubst \theta''\rsubst_b(A')  &\reason{by \refNum{rsubst-b-A-vs-A'}}  \\
  = &\theta''(A')  &\reason{by \refNum{theta''-renaming-trick}}\\
  = &A_p &\reason{by \refNum{A'-vs-A-p}}\\
\end{array}
\mathlabelNum{A-p-vs-A}
\end{equation*}
\end{proof}
}

\begin{lemma}
\label{lem:subst-gamma-and-A}
If $\termwf{\Delta}{M}$ and $\Delta \vdash \theta : \Theta \Rightarrow
\Theta'$ and $\typ{\Delta,\Theta;\Gamma}{M :A}$, then
  $\typ{\Delta,\Theta';\theta\Gamma}{M : \theta A}$.
\end{lemma}
\begin{proof}
  By induction on structure of $M$.
  In each case we apply inversion on
    derivations of $\typ{\Delta,\Theta;\Gamma}{M :A}$ and $\termwf{\Delta}{M}$
    and start by showing the final steps in each derivation, then describe how to
    construct the needed conclusion.
    \begin{description}
    \item[Case $\freeze{x}$:]  In this case we have derivations of the form:
      \[
      \inferrule {x : A \in \Gamma} {\typ{\Delta,\Theta; \Gamma}{\freeze{x} :
          A}}
      \qquad
      \inferrule {\strut} {\termwf{\Delta}{\freeze{x}}}
      \]
      Then we have $x : \theta(A) \in \theta(\Gamma)$, and may conclude
      \[
      \inferrule {x : \theta(A) \in \theta(\Gamma)} {\typ{\Delta,\Theta';
          \theta(\Gamma)}{\freeze{x} : \theta(A)}}
      \]

    \item[Case $x$:]  In this case, we have derivations of the form:
      \[
      \inferrule
      {{x : \forall \Delta'.H \in \Gamma} \\
        \Delta,\Theta \vdash \rsubst : \Delta' \Rightarrow_\poly \cdot }
      {\typ{\Delta,\Theta; \Gamma}{x : \rsubst(H)}}
      \qquad
      \inferrule {\strut} {\termwf{\Delta}{x}}
      \]
      As before, we have $x : \theta(\forall \Delta'.H) \in \theta(\Gamma)$.
      Moreover, we can assume without loss of generality that the type variables
      in $\Delta'$ are fresh, so
      $\theta(\forall \Delta'. H) = \forall \Delta'.\theta(H)$. Since
      $\wfctx{\Delta,\Theta}{\Gamma}$, we know that
      $\forall a \in \ftv(A). (\Delta,\Theta)(a) = \mono$.  Hence, for each such
      $a$, the substituted type $\theta(a)$ is a monotype, which implies that
      $\theta(H)$ is also a guarded type.  Next, by
      Lemma~\ref{lem:subst-substitution} we have
      $\typ{\Delta,\Theta'}{\theta \circ \rsubst : \Delta' \Rightarrow_\poly
        \cdot}$. We may conclude:
      \[
      \inferrule
      {{x : \forall \Delta'.\theta(H) \in \theta(\Gamma)} \\
        \Delta,\Theta' \vdash \theta \circ \rsubst : \Delta' \Rightarrow_\poly \cdot }
      {\typ{\Delta,\Theta'; \theta(\Gamma)}{x : \theta(\rsubst(H))}}
      \]
      Note that in this case it is critical that we maintain the invariant
      (built into the context well-formedness judgement) that type variables in
      $\Gamma$ are always of kind $\mono$.  This precludes substituting a type
      variable $a = H$ with a $\forall$-type, thereby changing the outer
      quantifier structure of $\forall \Delta'.H$.
    \item[Case $\lambda x. M$:]  In this case we have derivations of the
      form:
      \[
      \inferrule {\typ{\Delta,\Theta; \Gamma, x : S}{M : B}}
      {\typ{\Delta,\Theta; \Gamma}{\lambda x.M : S \to B}}
      \qquad
      \inferrule {\termwf{\Delta}{M}} {\termwf{\Delta}{\lambda x.M}}
      \]
      By induction, we have that
      $\typ{\Delta,\Theta';\theta(\Gamma,a:S) }{ M : \theta{B}}$.  Moreover,
      clearly $\theta(\Gamma,a:S) = \theta(\Gamma),a:\theta(S)$.  Since $S$ is a
      monotype, and $\theta$ is a well-kinded substitution,
      and $\termwf{\Delta, \Theta}{\Gamma, x : S}$,
      all of the free type
      variables in $S$ are of kind $\mono$ and are replaced with monotypes.
      Hence $\theta(S)$ is also a monotype, so we may derive:
      \[
      \inferrule {\typ{\Delta,\Theta'; \theta(\Gamma), x : \theta(S)}{M :
          \theta(B)}} {\typ{\Delta,\Theta'; \theta(\Gamma)}{\lambda x.M :
          \theta(S) \to \theta(B)}}
      \]
      since $\theta(S \to B) = \theta(S) \to \theta(B)$.
    \item[Case $\lambda (x:A_0).M$:]
      \[
      \inferrule {\typ{\Delta,\Theta; \Gamma, x : A_0}{M : B_0}}
      {\typ{\Delta,\Theta; \Gamma}{\lambda (x : A_0).M : A_0 \to B_0}}
      \qquad
      \inferrule
      {\Delta \vdash A_0 : \star\\
        \termwf{\Delta}{M} } {\termwf{\Delta}{\lambda (x : A_0).M }}
      \]
      By induction, we have that
      $\typ{\Delta,\Theta';\theta(\Gamma,x:A_0)}{M: \theta(B_0)}$, and again
      $\theta(\Gamma,x:A_0) = \theta(\Gamma),x:\theta(A_0)$.  Moreover, since
      $\Delta \vdash A_0:\star$, we know that $\ftv(A_0) \subseteq \Delta$.
      Since the only variables substituted by $\theta$ are those in $\Theta$,
      which is disjoint from $\Delta$, we know that $\theta(A_0) = A_0$.  Thus,
      we can proceed as follows:
      \[
      \inferrule {\typ{\Delta,\Theta'; \theta(\Gamma), x : A_0}{M :
          \theta(B_0)}} {\typ{\Delta,\Theta'; \theta(\Gamma)}{\lambda (x :
          A_0).M : A_0 \to \theta(B_0)}}
      \]
      observing that
      $\theta(A_0\to B_0) = \theta(A_0) \to \theta(B_0) = A_0 \to \theta(B_0)$,
      as required.  This case illustrates part of the need for the
      $\termwf{\Delta}{M}$ judgement: to ensure that the free type variables in
      terms are always treated rigidly and never ``captured'' by substitutions
      during unification or type inference.
    \item[Case $M~N$:] In this case we proceed (refreshingly
      straightforwardly) by induction as follows.
      \[
      \inferrule
      {\typ{\Delta,\Theta; \Gamma}{M : A_0 \to B_0} \\
        \typ{\Delta,\Theta; \Gamma}{N : A_0} } {\typ{\Delta,\Theta;
          \Gamma}{M\,N : B_0}}
      \qquad
      \inferrule
      {\termwf{\Delta}{M} \\
        \termwf{\Delta}{N} } {\termwf{\Delta}{M\,N}}
      \]
      By induction, we obtain the necessary hypotheses for the desired
      derivation:
      \[
      \inferrule {\typ{\Delta,\Theta'; \theta(\Gamma)}{M : \theta(A_0) \to
          \theta(B_0)} \\
        \typ{\Delta,\Theta'; \theta(\Gamma)}{N : \theta(A_0)} }
      {\typ{\Delta,\Theta; \theta(\Gamma)}{M\,N :\theta( B_0)}}
      \]
      again observing that $\theta(A_0 \to B_0) = \theta(A_0) \to \theta(B_0)$.
    \item[Case $\Let\; x = M \;\In\; N$]:  In this case we have
      derivations of the form:
      \[
      \inferrule
      {(\Delta', \Delta'') = \mgend[(\Delta,\Theta)]{A'}{M} \\
        \Delta,\Theta, \Delta''; \Gamma \vdash M : A' \\
        ((\Delta,\Theta), \Delta'', M, A') \Updownarrow A_0 \\
        \typ{\Delta,\Theta; \Gamma, x : A_0}{N : B} \\
        \meta{principal}((\Delta,\Theta), \Gamma, M,\Delta'',A') }
      {\typ{\Delta,\Theta; \Gamma}{\Let \; x = M\; \In \; N : B}}
      \]
      \[
      \inferrule
      {\termwf{\Delta}{M } \\
        \termwf{\Delta}{N} } {\termwf{\Delta}{\Let \; x = M\; \In \; N }}
      \]
      We assume without loss of generality that $\Delta''$ is fresh with
      respect to $\Delta$, $\Theta$, and $\Theta'$.
      This is justified as we may otherwise apply a substitution $\theta_F$ to
      $\Delta,\Theta, \Delta''; \Gamma \vdash M : A'$ that replaces all
      variables in $\Delta''$ by pairwise fresh ones. By induction, this would
      yield a corresponding typing judgement for $M$ using those fresh
      variables.

      To apply the induction hypothesis to $M$, we need to extend $\theta$ to
      a substitution $\theta'$ satisfying
      $\typ{\Delta}{\theta': \Theta,\Delta'' \Rightarrow \Theta',\Delta''}$,
      which is the identity on all variables in $\Delta''$.
      Then by induction we have
      $\Delta,\Theta', \Delta''; \theta'(\Gamma) \vdash M : \theta'(A')$.
      Since $\theta'$ acts as the identity on $\Delta''$ its behaviour is the
      same as $\theta$ weakened to
      $\Delta, \Delta'' \vdash \theta : \Theta \Rightarrow \Theta'$, so we
      have
      $\Delta,\Theta', \Delta''; \theta(\Gamma) \vdash M : \theta(A')$.

      We also obtain by the induction hypothesis for $N$ that
      $\typ{\Delta,\Theta';\theta(\Gamma),x:\theta(A_0)}{ N : \theta(B)}$,
      since $\theta(\Gamma,x:A_0) = \theta(\Gamma),x:\theta(A_0)$.
      By Lemma~\ref{lem:subst-gen-inst}(1), we have that $(\Delta', \Delta'') =
      \mgend[(\Delta,\Theta')]{\theta(A')}{M}$
      and by Lemma~\ref{lem:subst-gen-inst}(2), we also know that
      $((\Delta,\Theta'), \Delta'', M, \theta(A')) \Updownarrow \theta(A_0)$.
      By applying \cref{lem:stability-principality-substitution} to
      $\meta{principal}((\Delta,\Theta), \Gamma, M,\Delta'',A')$
      we obtain
      $\meta{principal}((\Delta,\Theta'),\allowbreak
      \theta(\Gamma),\allowbreak M,\Delta'', \theta(A'))$.  We can conclude:
      \[
      \inferrule
      {(\Delta', \Delta'') = \mgend[(\Delta,\Theta')]{\theta(A')}{M} \\
        \Delta,\Theta', \Delta''; \theta(\Gamma) \vdash M : \theta(A') \\
        ((\Delta,\Theta'), \Delta'', M, \theta(A')) \Updownarrow \theta(A_0) \\
        \typ{\Delta,\Theta'; \theta(\Gamma), x : \theta(A_0)}{N : \theta(B)} \\
        \meta{principal}((\Delta,\Theta'), \theta(\Gamma), M,\Delta'', \theta(A')) }
      {\typ{\Delta,\Theta'; \theta(\Gamma)}{\Let \; x = M\; \In \; N : \theta(B)}}
      \]

    \item[Case $\Let\; (x : A_0) = M \;\In\; N$:]  In this case we have
      derivations of the form:
      \[
      \inferrule
      {(\Delta', A') = \msplit(A_0, M) \\
        \typ{\Delta,\Theta, \Delta'; \Gamma}{M : A'} \\
        A_0 = \forall \Delta'.A' \\
        \typ{\Delta,\Theta; \Gamma, x : A_0}{N : B} } {\typ{\Delta,\Theta;
          \Gamma}{\Let \; (x : A_0) = M\; \In \; N : B}}
      \]
      \[
      \inferrule
      {\Delta \vdash A_0 : \poly \\
        (\Delta', A') = \msplit(A_0, M) \\
        \termwf{\Delta, \Delta'}{M} \\
        \termwf{\Delta}{N} } {\termwf{\Delta}{\Let \; (x : A_0) = M\; \In \;
          N}}
      \]
      We have $A_0 = \forall \Delta'.A'$ and $\Delta' \fresh \Delta$.
      According to $\termwf{\Delta, \Delta'}{M}$, annotations in
      $M$ may use type variables from $\Delta, \Delta'$.
      By alpha-equivalence, we can assume $\Delta' \fresh \Theta$ and
      $\Delta' \fresh \Theta'$.
      Note that this may require freshening variables from $\Delta'$
      (but not $\Delta$) in $M$ as well.

      By induction (and rearranging contexts), we have that
      $\typ{\Delta,\Theta', \Delta'; \theta(\Gamma)}{M : \theta(A')}$ and
      $\typ{\Delta,\Theta'; \theta(\Gamma, x : A_0)}{N : \theta(B)}$.

      Moreover, since $\Delta \vdash A_0:\poly$, we know that
      $\theta(A_0) = A_0$ since $\theta$ only replaces variables in $\Theta$,
      which is disjoint from $\Delta$.  Furthermore,
      $(\Delta',A') = \msplit(A_0,M)$ implies that $A'$ is a subterm of $A_0$
      so $\theta(A') = A'$ also.  As a result, we can construct the following
      derivation:
      \[
      \inferrule
      {(\Delta', A') = \msplit(A_0, M) \\
        \typ{\Delta,\Theta', \Delta'; \theta(\Gamma)}{M : A'} \\
        A_0 = \forall \Delta'.A' \\
        \typ{\Delta,\Theta'; \theta(\Gamma), x : A_0}{N : \theta(B)} }
      {\typ{\Delta,\Theta'; \theta(\Gamma)}{\Let \; (x : A_0) = M\; \In \; N
          : \theta(B)}}
      \]
    \end{description}
\end{proof}

\begin{lemma}
\label{lem:shrink-context}
Let $\typ{\Delta; \Gamma}{M : A}$
and let $\Delta' \subseteq \Delta$ such that
and $\termwf{\Delta'}{M}$,
$\wfctx{\Delta'}{\Gamma}$, and $\Delta' \vdash \ftv(A)$.
Then
$\typ{\Delta'; \Gamma}{M : A}$ holds.
\end{lemma}
\begin{proof}

\fe{There exists a 4-line version of this proof, by simply applying
\cref{lem:subst-gamma-and-A} directly to $M$. This is okay in the current
dependency graph, but may cause troubles in the future. I've kept the shorter
proof as commented-out code here.}


By induction on $M$; we focus on the \freezemlLab{Let} case.
By inversion on the judgement $\typ{\Delta;\Gamma}{\Let\; x = M \;\In\; N : B}$,
we have:
\begin{equations}
(\Delta'_g, \Delta''_g) = \mgend{A'}{M} \\
(\Delta, \Delta''_g, M, A') \Updownarrow A \\
\Delta, \Delta''_g; \Gamma \vdash M : A' \\
\typ{\Delta; \Gamma, x : A}{N : B} \\
\meta{principal}(\Delta, \Gamma, M, \Delta''_g, A')
\end{equations}
By inversion on $\termwf{\Delta'}{\Let\; x = M \;\In\; N}$, we further have
$\termwf{\Delta'}{M}$ and $\termwf{\Delta'}{N}$.

We first show $(\Delta - \Delta') \fresh \ftv(A')$.
To this end, assume there exists $a$ s.t. $a \in \ftv(A')$ and $a \in (\Delta -
\Delta')$.
By $\Delta' \vdash \Gamma$, this implies $a \not\in \ftv(\Gamma)$.
Let $b$ be fresh and $\theta = [ a \mapsto b]$.
We apply \cref{lem:bijective-renaming}.\ref{lem-part:bijective-renaming:typ} to
$\Delta, \Delta''_g; \Gamma \vdash M : A'$, where $\Theta$ contains only $a$.
This yields $\typ{(\Delta \setminus a, \Delta''_g, b);\theta(\Gamma)}{M :
\theta{A'}}$.
By $a \not\in \ftv(\Gamma)$, this is equivalent to $\typ{(\Delta \setminus a,
\Delta''_g, b);\Gamma}{M : A'[b/a]}$.
According to \cref{lem:extend-tv-context-typing}, we can weaken this to
$\typ{(\Delta, \Delta''_g, b);\Gamma}{M : A'[b/a]}$.
However, by $a \in \Delta$ and $a \not\in \Delta''_g$ there exist no $\rsubst$
and $\Delta_?$ such that $\Delta \vdash \rsubst : \Delta''_g \Rightarrow
\Delta_? $ and $\rsubst(A') = A'[b/a]$.
This violates $\meta{principal}(\Delta, \Gamma, M, \Delta''_g, A')$.

Using $(\Delta - \Delta') \fresh \ftv(A')$, we obtain $\mgend{A'}{M} =
(\Delta'_g, \Delta''_g) = \mgend[\Delta']{A'}{M}$ and $\Delta', \Delta''_g
\vdash A'$.
This allows us to apply the induction hypothesis to $M$, yielding $\Delta,
\Delta''_g; \Gamma \vdash M : A'$.

We show $\meta{principal}(\Delta', \Gamma, M, \Delta''_g, A')$ as follows:
Let $\Delta''_p, A''_p$ such that $\Delta''_p = \ftv(A''_p) - \Delta'$ and
$\typ{\Delta', \Delta''_p;\Gamma}{M : A''_p}$.
Further, let $\Delta_1, \Delta_2, \Delta_3$ such that $\Delta_1, \Delta_2
\approx \Delta''_p$, and $\Delta_1 \subseteq \Delta$, and $\Delta_2 \fresh
\Delta$, and $\ftv(A''_p) - \Delta',\Delta_1 = \Delta_2$, and $\Delta', \Delta_1,
\Delta_3 \approx \Delta$.
We can therefore rewrite $\typ{\Delta', \Delta''_p;\Gamma}{M : A''_p}$ to
$\typ{\Delta', \Delta_1, \Delta_2;\Gamma}{M : A''_p}$.
We use \cref{lem:extend-tv-context-typing} to weaken the latter to
$\typ{\Delta', \Delta_1, \Delta_2, \Delta_3;\Gamma}{M : A''_p}$.
This in turn is equivalent to $\typ{\Delta, \Delta_2;\Gamma}{M : A''_p}$,
additionally recalling $\ftv(A''_p) - \Delta = \Delta_2$.
By $\meta{principal}(\Delta, \Gamma, M, \Delta''_g, A')$, we then have that
there exists $\rsubst_p$ s.t.\ $\Delta \vdash \rsubst_p : \Delta''_g \Rightarrow_\poly
\Delta_2$ and $\rsubst_p(A') = A''_p$.
We can re-arrange the former to $\Delta', \Delta_3 \vdash \rsubst_p : \Delta''_g
\Rightarrow_\poly \Delta_1, \Delta_2$.
We have $\Delta''_g \subseteq \ftv(A')$ but $\Delta_3 \fresh
\ftv(\rsubst_p(A'))$, which implies $\Delta_3 \fresh \ftv(\rsubst_p)$.
Hence, we have $\Delta' \vdash \Delta''_g \Rightarrow_\poly \Delta_1, \Delta_2$.
Thus, $\meta{principal}(\Delta', \Gamma, M, \Delta''_g, A')$ holds.

Next, we show that there exists $\tilde{A}$ such that
$(\Delta', \Delta''_g, M, A') \Updownarrow \tilde{A}$ and
$\typ{\Delta; \Gamma, x : \tilde{A}}{N : B}$.

We distinguish two cases:
\begin{itemize}
\item
If $M \in \dec{GVal}$, then $(\Delta, \Delta''_g, M,A') \Updownarrow A$, where
$A = \forall \Delta''_g. A'$.
We choose, $\tilde{A} = A$ and by $\Delta' \vdash A'$ immediately obtain
$\Delta' \vdash \tilde{A}$ and $(\Delta', \Delta''_g, M,A') \Updownarrow
\tilde{A}$.
By induction, we then have $\typ{\Delta; \Gamma, x : \tilde{A}}{N : B}$.
\item
If $M \not\in \dec{GVal}$, then $(\Delta, \Delta''_g, M,A') \Updownarrow A$,
where $A = \rsubst(A')$ for some $\rsubst$ with $\Delta \vdash \rsubst : \Delta''_g
\Rightarrow_\mono \cdot$.
Hence, $A$ may contain type variables from $\Delta - \Delta'$, which we define
as $\Delta_r$, and hence $\Delta' \vdash A$ may not hold.

To obtain a type well-defined under $\Delta'$ we define substitution
$\rsubst_G$, which maps all type variables in $\Delta_r$ to  some ground
type, e.g., $\Int$.
Formally, $\rsubst_G$ be defined such that
\[
  \rsubst_G(a) =\Int \qquad \text{for all } a \in \Delta_r,
\]
which implies $\Delta' \vdash \rsubst_G : \Delta_r \Rightarrow_\mono \cdot$.

We then define $\tilde{A}$ as $\rsubst_G(\rsubst(A')) = \rsubst_G(A)$ and have
$(\Delta, \Delta''_g, M, A') \Updownarrow \tilde{A}$, due to
$\Delta' \vdash (\rsubst_G \circ \rsubst) : \Delta''_g \Rightarrow_\mono \cdot$.

Further, we apply \cref{lem:subst-gamma-and-A} to
$\typ{\Delta', \Delta_r; \Gamma, x : A}{N : B}$ and $\rsubst_G$, yielding
$\typ{\Delta'; \rsubst_G(\Gamma), x : \rsubst_G(A)}{N : \rsubst_G(B)}$.

By $\Delta' \vdash \Gamma$ and $\Delta' \vdash B$ we have
$\rsubst_G(\Gamma) = \Gamma$ and $\rsubst_G(B) = B$.
Together, with $\rsubst_G(A) = \tilde{A}$ we have therefore shown
$\typ{\Delta'; \Gamma, x : \tilde{A}}{N : B}$.
\end{itemize}

We have now shown that we can derive the following:

\[
\inferrule*
    {(\Delta'_g, \Delta''_g) = \mgend{A'}{M} \\
     (\Delta', \Delta''_g, M, A') \Updownarrow \tilde{A} \\
     \Delta', \Delta''_g; \Gamma \vdash M : A' \\
     \typ{\Delta'; \Gamma, x : \tilde{A}}{N : B} \\\\
     \meta{principal}(\Delta', \Gamma, M, \Delta''_g, A')
    }
    {\typ{\Delta'; \Gamma}{\Let \; x = M\; \In \; N : B}}
\]
\end{proof}

\begin{lemma}
\label{lem:extend-tv-context-typing}
Let $\typ{\Delta;\Gamma}{M : A} $, and $\termwf{\Delta}{M}$, and $\Delta \fresh \Theta$.
Then $\typ{\Delta, \Theta;\Gamma}{M : A}$ holds.
\end{lemma}
\begin{proof}
\fe{It seems that as for \cref{lem:shrink-context}, there exists a trivial proof
here using $\Delta \vdash \theta : \cdot \Rightarrow \Theta$ and
\cref{lem:subst-gamma-and-A}. As for \cref{lem:shrink-context}, this slightly complicates the
dependencies, as we would use \cref{lem:subst-gamma-and-A} on the original term $M$, not
its subterms.}

We perform induction on $M$ and focus on the case $\Let \; x = M \;\In\ N$.  By
inversion, we have the following:
\begin{equations}
(\Delta', \Delta'') = \mgend[\Delta]{A'}{M} \\
(\Delta, \Delta'', M, A') \Updownarrow A \\
\Delta, \Delta''; \Gamma \vdash M : A' \\
\typ{\Delta; \Gamma, x : A}{N : B} \\
\meta{principal}(\Delta, \Gamma, M, \Delta'', A')
\end{equations}

We assume w.l.o.g.\  $\Delta'' \fresh \ftv(\Theta)$
(this is justified as we may otherwise use \cref{lem:bijective-renaming} to obtain
a type for $M$ satisfying this).

By induction, we immediately have
$\Delta, \Theta, \Delta''; \Gamma \vdash M : A'$ and
$\typ{\Delta, \Theta; \Gamma, x : A}{N : B}$.
Further, we have $\ftv(\Theta) \fresh \ftv(A')$ and thus $\mgend[(\Delta,
\Theta)]{A'}{M} = \mgend{A'}{M} = (\Delta', \Delta'')$.
By $\Delta'' \fresh \ftv(\Theta)$, we also immediately have $((\Delta, \Theta),
\Delta'', M, A') \Updownarrow A$:
$\Theta$ is only relevant if $M \not\in \dec{GVal}$, in which case we can weaken
the involved instantiation $\rsubst$ from
$\Delta \vdash \rsubst : \Delta'' \Rightarrow_\mono \cdot$ to
$\Delta, \Theta \vdash \rsubst : \Delta'' \Rightarrow_\mono \cdot$.

It remains to show that $\meta{principal}((\Delta, \Theta), \Gamma, M, \Delta'',
A')$ holds.
Let $\Delta_p$ and $A_p$ such that $\typ{\Delta, \Theta, \Delta_p; \Gamma}{M :
A_p }$ and $\ftv(A_p) - \Delta, \Theta = \Delta_p$.
Let $\Theta_s \subseteq \Theta$ such that $\Delta_p, \Theta_s = \ftv(A_p) -
\Delta$, which implies $\Delta, \Theta_s, \Delta_p \vdash A_p$.

By \cref{lem:shrink-context}, we then have $\typ{\Delta, \Delta_p, \Theta_s;
\Gamma}{M : A_p }$.
By $\meta{principal}(\Delta, \Gamma, M, \Delta'', A')$ there exists $\rsubst$
such that $\Delta \vdash \rsubst : \Delta'' \Rightarrow \Delta_p, \Theta_s$
$\rsubst(A') = A_p$.
By $\Theta_s \subseteq \Theta$ and $\Theta \fresh \Delta''$ we can weaken this
to $\Delta, \Theta \vdash \rsubst : \Delta'' \Rightarrow \Delta_p$.
This gives us $\meta{principal}((\Delta, \Theta), \Gamma, M, \Delta'', A')$.

\end{proof}

\begin{lemma}
\label{lem:principal:rename-gen-vars} Let $\Delta' = (a_1, \dotsc, a_n)$ and
$\Delta'' = (b_1, \dotsc, b_n)$ for some $n \ge 0$.  Let $\Delta, \Theta \vdash
\rsubst : \Delta' \Rightarrow \Delta''$ such that $\rsubst(a_i) = b_i$ for all
$1 \le i \le n$.  Furthermore, let $\termwf{\Delta}{M}$ and $\meta{principal}(
(\Delta, \Theta), \Gamma, M, \Delta', A )$ and $\wfctx{\Delta, \Theta}{\Gamma}$.

Then $\meta{principal}( (\Delta, \Theta), \Gamma, M, \Delta'', \rsubst A )$
holds.
\end{lemma}
\begin{proof}

We first show that $\typ{\Delta,\Theta,\Delta'';\Gamma}{M : \rsubst A }$ holds.
By $\meta{principal}( (\Delta, \Theta), \Gamma, M, \Delta', A )$ we have
$\typ{\Delta,\Theta,\Delta';\Gamma}{M : A }$.  We extend $\rsubst$ to a
substitution $\theta$ with $\Delta \vdash \theta : \Theta, \Delta' \Rightarrow
\Theta, \Delta''$ by defining $\theta(a) = \rsubst(a)$ for all $a \in \Delta'$
and by defining $\theta$ as the identity on all $a \in \Theta$.

We apply \cref{lem:bijective-renaming}(\ref{lem-part:bijective-renaming:typ}),
yielding $\typ{\Delta,\Theta,\Delta'';\theta(\Gamma)}{M : \theta A }$.  We have
$\theta(\Gamma) = \Gamma$ as well as $\theta(A) = \rsubst(A)$ and obtain the
desired judgement.

Now, let $\Delta_p$ and $A_p$ such that $\ftv(A_p) - \Delta, \Theta = \Delta_p$
and $\typ{\Delta, \Theta, \Delta_p;\Gamma}{M : A_p}$.  By $\meta{principal}(
(\Delta, \Theta), \Gamma, \allowbreak M, \allowbreak \Delta', \allowbreak A )$
we have that there exists an instantiation $\rsubst_p$ s.t.\ $\Delta, \Theta
\vdash \rsubst_p : \Delta' \Rightarrow_\mono \Delta_p$ and $\rsubst_p(A) = A_p$.

We need to show that then there also exists an instantiation $\rsubst'_p$ with
$\Delta, \Theta \vdash \rsubst'_p : \Delta'' \Rightarrow_\mono \Delta_p$ and
$\rsubst'_p(\rsubst A) = A_p$.  We observe that this holds for $\rsubst'_p =
\rsubst_p \circ \rsubst^{-1}$, where $\rsubst^{-1}$ is the inverse of $\rsubst$.
\end{proof}

}

\fe{
The following lemma is currently only used in the let case of the proof of
\cref{thm:inference-completeness-mg}. It's very ad hoc; I was hoping to
replace it with something re-using more of the earlier properties
of principality.
However, I haven't found a good replacement.
}

{
\proofContext{bijection-on-principal-types-vars}
\begin{lemma}
\label{lem:bijection-on-principal-types-vars}
Let the following conditions hold:
\resetNum

\begin{align}
&\termwf{\Delta}{M} \mathlabelNum{M-wf} \\
&\theta = \theta'' \circ \theta' \mathlabelNum{composition} \\
&\Delta \vdash \theta : \Theta \Rightarrow \Theta' \mathlabelNum{theta-wf} \\
&\Delta \vdash \theta'' : \Theta'' \Rightarrow \Theta', \Delta'' \mathlabelNum{theta''-wf} \\
&\Delta' = \ftv(A) - \Delta - \ftv(\theta') \mathlabelNum{Delta-prime} \\
&\typ{\Delta, \Theta''; \theta'\Gamma}{ M : A} \mathlabelNum{typ-M}\\
&\dec{principal}((\Delta, \Theta'), \theta \Gamma, \Delta'', A') \mathlabelNum{principal} \\
&\theta''(A) = A' \mathlabelNum{A-vs-A'}\\
\end{align}

Then $\theta''(\Delta') = \Delta''$ holds.
\end{lemma}
\begin{proof}
By \refNum{principal}, we have $\Delta'' = \ftv(A') - \Theta' - \Delta$.
Further, \refNum{typ-M} yields $\Delta, \Theta''\vdash A $~\labelNum{A-wf}.

Let $\Delta' = (a'_1, \dotsc, a'_n)$ for some $n \ge 0$
and let $\Delta_F = (f_1, \dotsc, f_n)$ for pairwise different, fresh type variables $f_i$.

By \refNum{A-wf}, we have $\ftv(A) \subseteq \Delta, \Theta''$.
Let $\Theta_{\theta'}$ be defined as $\ftv(\theta') - \Delta$.
We then have
$\Theta_{\theta'} \subseteq \Theta''$~\labelNum{image-of-theta-prime-minus-Delta-is-subset-of-Theta''}
and $\Delta' \mathbin{\#} \Theta_{\theta'}$~\labelNum{Delta'-Theta-img-disj} and
 $\Delta' \subseteq \ftv(\Theta'')$~\labelNum{Delta'-subs-Theta''}.

By \refNum{composition,theta-wf,theta''-wf}
we have $\Delta, \Theta' \vdash \theta''(a) :K $ for all $(a : K) \in \Theta_{\theta'}$~\labelNum{theta''-restriction}.

Let $\theta''_F$ be defined such that
\begin{equation}
\theta''_F(a) =
\begin{cases}
  \theta''(a) &\text{if $a \in \Theta_{\theta'}$} \\
  f_i           &\text{if $a = a'_i \in \Delta'$} \\
  A_D           &\text{if $a \in \Theta'' - \Theta_{\theta'} - \Delta'$}
\end{cases}
\mathlabelNum{def-of-theta''-F}
\end{equation}
where $A_D$ is some arbitrary type with $\Delta, \Theta' \vdash A_D : \mono$
(e.g., $\Int$, cf.~\cref{fig:freezeml-syntax}).

By
  \refNum{image-of-theta-prime-minus-Delta-is-subset-of-Theta'',%
          Delta'-Theta-img-disj,%
          Delta'-subs-Theta''},
this definition is well-formed.
Together with~\refNum{theta''-restriction}
we then have $\Delta \vdash \theta''_F : \Theta'' \Rightarrow \Theta', \Delta_F$~\labelNum{theta''-F-wf} and
$\theta = \theta''_F \circ \theta'$~\labelNum{num:compl:let:tilde-theta-second-composiiton}.

By \refNum{A-wf} and \cref{lem:stability-wf-typ}, we then have
$\Delta, \Theta', \Delta_F \vdash \theta''_F A$ which implies $\ftv(\theta''_F A)  \subseteq \Delta_F,\Delta,\Theta'$.
In general, for every $a \in \ftv(A)$, $\theta''_F(a)$ is part of $\theta''_F(A)$.
In particular, for each $a'_i \in \Delta' \subseteq \ftv(A)$, $\theta''_F(a'_i) = f_i$ occurs in $\theta''_F(A)$.
Thus, $\ftv(\theta''_F A) - \Delta,\Theta' = \Delta_F$ holds~\labelNum{ftv-of-substed-A}.

By \refNum{typ-M,theta''-F-wf,M-wf}, \cref{lem:subst-gamma-and-A} yields
$\typ{\Delta, \Theta', \Theta_F; \theta''_F\theta'\Gamma}{M : \theta''_F (A)}$,
which by~\refNum{num:compl:let:tilde-theta-second-composiiton}
is equivalent to $\typ{\Delta, \Theta', \Delta_F; \theta\Gamma}{M : \theta''_F (A)}$~\labelNum{type-M-substed}.
By definition of $\meta{principal}$ as well as
  \refNum{principal,%
          type-M-substed,%
          ftv-of-substed-A}
there exists $\rsubst$ such that $\Delta, \Theta' \vdash \rsubst : \Delta'' \Rightarrow \Delta_F$%
  ~\labelNum{num:compl:let:rsubst-sqsubseteq-wf} and
$\rsubst(A') = \theta''_F(A)$.

By \refNum{A-vs-A'}, the latter is equivalent to $\rsubst(\theta''(A)) = \theta''_F(A)$%
  ~\labelNum{num:compl:let:sqsubseteq-equality}

\fe{I feel like the rest of this proof is rather verbose, in particular since I reason
about similar bijectivity properties much more briefly in the proof of
\cref{lem:inferred-types-vs-principal-types}
}
Let $a \in \Delta' \subseteq \ftv(A)$, which implies $a = a'_i$ for some $ 1 \le i \le n$.
By \refNum{num:compl:let:sqsubseteq-equality}, we have
\[
\begin{array}{rrcll}
 &\rsubst\theta''(a_i) &= &\theta''_F(a_i) \\
\text{equiv.} \qquad
 &\rsubst\theta''(a_i) &= &f_i &\reason{by \refNum{def-of-theta''-F}}
\end{array}
\]

We therefore have that for each such $a'_i$, $\theta''(a_i)$ maps to pairwise different type variables $b_i$.
By \refNum{num:compl:let:rsubst-sqsubseteq-wf} and $\Delta, \Theta', \Delta'' \mathbin{\#} \Delta_F$, we have $\rsubst(b_i) \neq b_i$ and therefore
$b_i \in \Delta''$.
We have therefore shown that $\theta''$ maps $\Delta'$ injectively into
$\Delta''$~\labelNum{num:compl:let:theta-two-primes-one-injective-on-Delta-three-primes}.

We now show that $\theta''$ is also surjective from $\Delta'$ into $\Delta''$, which means that
$\theta''(\Delta')$ is a permutation of $\Delta''$. 
To this end, assume that there exists $b \in \Delta''$ such that there exists no $a \in \Delta'$ with $\theta''(a) = b$.
By $b \in \Delta'' \subseteq \ftv(A')$ and \refNum{A-vs-A'} we have that there must exist $a \in \ftv(A)$
such that $b \in \ftv(\theta''(a))$.
By \refNum{num:compl:let:theta-two-primes-one-injective-on-Delta-three-primes}, $a \in \Delta'$ would immediately yield a contradiction.
By $\ftv(A) \subseteq \Theta_{\theta'},\Delta', \Delta$, we therefore consider the cases $a \in \Theta_{\theta'}$ and $a \in \Delta$.
If $a \in \Theta_{\theta'}$, according to \refNum{theta''-restriction}, we then have
 $\ftv(\theta''(a)) \subseteq \Delta, \Theta'$, which is disjoint from $\Delta''$.
If $a \in \Delta$, we have $\theta''(a) = a \not\in \Delta''$.
As all choices for $a$ yield contradictions, we have shown that $\theta''(\Delta')$ is a permutation of $\Delta''$

We now show that $\theta''(\Delta') = \Delta''$ holds (i.e., $\theta''$ preserves the order of type variables).
To this end, let $a \in \ftv(A) - \Delta'$, which implies $a \in \Delta, \Theta_{\theta'}$.
If $a \in \Delta$, then $\theta''(a) = a \in \Delta \fresh \Delta''$.
If $a \in \Theta_{\theta'}$ then by \refNum{theta''-restriction}
  we have $\ftv(\theta''(a)) \subseteq \Delta, \Theta' \mathop{\#} \Delta''$.
Therefore, together with \refNum{A-vs-A'} for all $a'_i, a'_j \in \Delta'$ with $1 \le i < j \le n$, we have
that the first occurrence of $\theta''(a'_i)$ in $A'$ is located before the first occurrence of $\theta''(a'_j)$ in $A'$.

\end{proof}
}

\subsection{Soundness of type inference}
\label{subsec:correctness-inference:soundness}

{
\begin{lemma}
\label{lem:theta-is-id-on-non-Gamma-tyvars}
If
  $\termwf{\Delta}{M}$ and
  $(\Theta', \theta, A) = \Infer(\Delta, \Theta, \Gamma, M)$
then for all
  $a \in (\Theta - \ftv(\Gamma))$
we have
$\theta(a) = a$ and $a \not\in \ftv(A)$.
\end{lemma}
\begin{proof}
Straightforward by induction on the structure of $M$, in each case checking
that a successful evaluation of type inference only instantiates free
variables present in $\Gamma$.
Furthermore, each type variable in $A$ is either fresh or results from
using a type in $\theta(\Gamma)$.
\end{proof}
}

\newcounter{thmbackup}
\setcounter{thmbackup}{\thethm}
\setcounterref{thm}{thm:inference-sound}
\addtocounter{thm}{-1}

\begin{thm}
If $\wfctx{\Delta, \Theta}{ \Gamma}$ and $\termwf{\Delta}{M}$ and
$\Infer(\Delta, \Theta, \Gamma, M) = (\Theta', \theta, A_0)$ then
$\typ{\Delta, \Theta'; \theta(\Gamma)}{M : A_0}$ and $\Delta
\vdash \theta : \Theta \Rightarrow \Theta'$.
\end{thm}

\setcounter{thm}{\thethmbackup}

\begin{proof}
\resetNum
\proofContext{infer-soundness}
By induction on structure of $M$.
In each case, we have
  $\wfctx{\Delta, \Theta}{\Gamma}$~\labelNum{num:sound:gamma-well-formed},
  $\termwf{\Delta}{M}$~\labelNum{num:sound:term-well-formed},
  and
  $\Infer(\Delta, \Theta, \Gamma, M) =
    (\Theta', \theta, A_0)$.
For each case, we show:
\begin{enumerate}[label=\Roman*., ref=(\Roman*)]
\item
  $\typ
    {\Delta, \Theta';\theta{\Gamma}}
    {M : A_0}$ \label{proofobl:typed}
\item
  $\Delta \vdash \theta : \Theta \Rightarrow \Theta'$%
    ~\label{proofobl:theta-wf}
\end{enumerate}
We write $\ref{proofobl:typed}$ and $\ref{proofobl:theta-wf}$ to indicate
that we have shown the respective statement.

\begin{description}
\item[Case $\freeze{x}$:]
By definition of $\meta{infer}$, we have
  $A_0  = \Gamma(x)$, $\Theta' = \Theta$, and
  $\theta = \iota_{\Delta,\Theta}$,
which implies $\Delta \vdash \theta : \Theta \Rightarrow
\Theta'$~\ref{proofobl:theta-wf} and $\wfctx{\Delta, \Theta'}{\Gamma}$.
We can then derive:
\[
\inferrule*[right=\freezemlLab{Freeze}]
    {x : A_0 \in \Gamma}
    {\typ
      {\Delta, \Theta'; \Gamma}
      {\freeze{x} : A_0}~\ref{proofobl:typed}}
\]

\item[Case $x$:]
By definition of $\meta{infer}$, we have
  $(x : \forall \many{a}.H) \in \Gamma$ and
  $\many{b} \,\#\, \Delta, \Theta$ and
  $A_0 = H[\many{b}/\many{a}]$ as well as
  $\Theta' = \Theta, \many{b : \poly}$.
Due to $\alpha$-equivalence, we can assume $\many{a} \mathop{\#} \many{b}, \Delta, \Theta$.

Let $\rsubst = [\many{b}/\many{a}]$.
We have $\Delta, \Theta, \many{b} \vdash \delta(a) : \poly $ for all $a \in \many{a}$ and
therefore
  $\Delta, \Theta, \many{b} \vdash
    \rsubst : (\many{a} : \mono) \Rightarrow_\poly \cdot$.
By \refNum{num:sound:gamma-well-formed} and $\many{b} \fresh \Delta, \Theta$, we have
$\wfctx{\Delta,\Theta, \many{b : \poly}}{\Gamma}$ and derive the following:

\[
\inferrule*[right=\freezemlLab{Var}]
    {{x : \forall \many{a}.H \in \Gamma}\; \\
      \Delta, \Theta,\many{b} \vdash \rsubst
         : (\many{a:\mono}) \Rightarrow_\poly \cdot
    }
    {\typ{\Delta, \Theta, \many{b : \poly}}{x : \rsubst(H)}~\ref{proofobl:typed}}
\]

We weaken $\Delta \vdash \iota_{\Delta, \Theta} : \Theta \Rightarrow \Theta$
to
$\Delta \vdash
   \iota_{\Delta, \Theta} : \Theta \Rightarrow \Theta, \many{b : \poly}$%
  ~\ref{proofobl:theta-wf}.

\item[Case $\lambda x.M$:]
By definition of $\meta{infer}$, we have $a \, \# \,
\Delta,\Theta$, which implies $a \,\#\,
\ftv(\Gamma)$~\labelNum{num:sound:lam:disjoint-gamma}.  Let $\theta_1 = \theta[a
\rightarrow S]$~\labelNum{num:sound:lam:theta-1}.

Together with \refNum{num:sound:gamma-well-formed} we then have
$\wfctx{\Delta, \Theta, a : \mono}{\Gamma, x : a}$.  By induction, we further have
\begin{align}
& \typ{\Delta, \Theta_1; \theta_1(\Gamma, x : a)}{M : B} \notag\\
\text{equiv.} \quad & \typ{\Delta, \Theta_1; \theta{\Gamma}, x : S}{M : B} \;
      \locallabel{num:sound:lam:typed-M}
\reason{by \refNum{num:sound:lam:disjoint-gamma}, \refNum{num:sound:lam:theta-1}}
\end{align}
as well as $\Delta \vdash \theta_1 : (\Theta, a : \mono) \Rightarrow \Theta_1$,
which implies $\Delta \vdash \theta : \Theta \Rightarrow
\Theta_1$~\ref{proofobl:theta-wf}.

By \refNum{num:sound:lam:typed-M} we have $\wfctx{\Delta, \Theta_1}{\theta\Gamma}$,
which allows us to derive the following:

\[
\inferrule*[right=\freezemlLab{Lam}]
    {\typ
      {\Delta, \Theta_1; \theta\Gamma, x : S}
      {M : B}\;\;(\text{by } \refNum{num:sound:lam:typed-M})}
    {\typ
      {\Delta, \Theta_1; \theta\Gamma}
      {\lambda x.M : S \to B}~\ref{proofobl:typed}}
\]

\item[Case $\lambda (x: A).M$:]

By $\termwf{\Delta}{\lambda (x: A).M}$ we have $\Delta \vdash A$~\labelNum{num:sound:A-wf},
and in particular all free type variables of $A$ in the judgement $\Delta, \Theta
\vdash A$ are monomorphic.  Together with \refNum{num:sound:gamma-well-formed}
this yields $\wfctx{\Delta, \Theta}{\Gamma, x : A}$.
Induction then yields $\typ{\Delta, \Theta_1; \theta(\Gamma, x : A)}{M :
B}$~\labelNum{num:sound:lam-as:M-typed} and $\Delta \vdash \theta : \Theta \Rightarrow
\Theta_1$~\ref{proofobl:theta-wf}.

According to \refNum{num:sound:A-wf} and the latter we further have
 $\theta(A) = A$~\labelNum{num:sound:lam-as:theta-A}.
By \refNum{num:sound:lam-as:M-typed} we have
$\wfctx{\Delta, \Theta_1}{\theta\Gamma}$ and can derive the following:
\[
\inferrule*[right=\freezemlLab{Lam-Ascribe}]
    {\typ
      {\Delta, \Theta_1; \theta\Gamma, x : A}
      {M : B} \;\;(\text{by } \refNum{num:sound:lam-as:M-typed},
                              \refNum{num:sound:lam-as:theta-A})}
    {\typ
      {\Delta, \Theta_1; \theta\Gamma}
      {\lambda (x : A).M : A \to B}~\ref{proofobl:typed}}
\]

\item[Case $M \: N$:]
By definition of $\meta{infer}$, 
 we have:
\begin{alignat}{3} (\Theta_1, \theta_1, A') &= \Infer(\Delta,\Theta, \Gamma,
M)~\locallabel{num:app:infer-M}\\ (\Theta_2, \theta_2, A) &=
\Infer(\Delta,\Theta_1, \theta_1\Gamma, N)~\locallabel{num:app:infer-N}
\end{alignat}
By induction, \refNum{num:app:infer-M} yields $\typ{\Delta,
\Theta_1;\theta_1\Gamma}{M : A'} \;~\labelNum{num:app:type-M}$ and $\Delta
\vdash \theta_1 : \Theta \Rightarrow
\Theta_1~\labelNum{num:app:theta-one-wf}$.

By \refNum{num:app:type-M} we have $\wfctx{\Delta, \Theta_1}{\theta_1\Gamma}$.
Therefore, by induction, \refNum{num:app:infer-N} yields $\typ{\Delta,
\Theta_2;\theta_2\theta_1\Gamma}{N : A} \;~\labelNum{num:app:type-N}$ and
$\Delta \vdash \theta_2 : \Theta_1 \Rightarrow
\Theta_2~\labelNum{num:app:theta-two-wf}$.
By definition of $\meta{infer}$, we have:
\begin{gather} b \mathop{\#} \ftv(A') \;\; b \mathop{\#} \ftv(A) \;\; b \mathop{\#} \Theta
\;\locallabel{num:app:b-free}\\ (\Theta_3, \theta'_3) = \unify(\Delta,
(\Theta_2, b : \poly ), \theta_2 A', A \to b)
\;\locallabel
{num:app:unify-def}\\ \theta'_3 = \theta_3[b \to B]
\;\locallabel{num:app:theta-three-def}
\end{gather}

By \refNum{num:app:type-M} we have $\Delta, \Theta_1 \vdash A'$
and by
\refNum{num:app:theta-two-wf} further
$\wfctx{\Delta, \Theta_2}{\theta_2 A'}$. This implies $\wfctx{\Delta, \Theta_2, b : \poly}
{ \theta_2 A'}$ by ~\refNum{num:app:b-free}.
By \refNum{num:app:type-N}  we have $\wfctx{\Delta, \Theta_2
}{A}$ and therefore also $\wfctx{\Delta, \Theta_2, b : \poly}{A \to b}$.
Together, those properties allow us to apply
Theorem~\ref{thm:unification-sound}, which gives us:

\begin{alignat}{6}
  &&\theta'_3\theta_2(A') &= \; &&\theta'_3(A \to b) \notag\\
\text{implies}\quad
  &&\theta_3\theta_2(A') &= &&\theta_3(A) \to B
\;\locallabel{num:app:unify-equality} &&\reason{by \refNum{num:app:b-free,num:app:theta-three-def}}
\end{alignat}

and
\begin{alignat}{5}
  \Delta &\vdash \theta'_3 : (\Theta_2, b : \poly) \Rightarrow
  \Theta_3 && \notag\\
\text{implies } \quad
  \Delta &\vdash \theta_3 : \Theta_2
\Rightarrow \Theta_3
&&\reason{\text{by $\refNum{num:app:theta-three-def}$}}
\locallabel{num:app:theta-three-wf}
\end{alignat}

By~\refNum{num:app:theta-two-wf}, \refNum{num:app:theta-three-wf}, and
composition, we have $\Delta \vdash \theta_3 \circ \theta_2 : \Theta_1
\Rightarrow \Theta_3$.
By \refNum{num:app:type-M} and Lemma~\ref{lem:subst-gamma-and-A}, we then
have $\typ{\Delta, \Theta_3;\theta_3\theta_2\theta_1\Gamma}{M :
\theta_3\theta_2 A'}$~\labelNum{num:app:type-M-substed}.
Similarly, by \refNum{num:app:theta-three-wf}, \refNum{num:app:type-N}, and
Lemma~\ref{lem:subst-gamma-and-A}, we have $\typ{\Delta,
\Theta_3;\theta_3\theta_2\theta_1\Gamma}{N :
\theta_3A}$~\labelNum{num:app:type-N-substed}

By \refNum{num:app:theta-one-wf}, \refNum{num:app:theta-two-wf},
\refNum{num:app:theta-three-wf}, and Lemma~\ref{lem:stability-wf-env}, we
have $\wfctx{\Delta}{\theta_3\theta_2\theta_1\Gamma}$.
We can then derive:

\[
\inferrule*[right=\freezemlLab{App}]
 {
   \typ
     {\Delta, \Theta_3;\theta_3\theta_2\theta_1\Gamma}
     {M : \theta_3(A) \to B\;
        (\text{by } \refNum{num:app:type-M-substed},\refNum{num:app:unify-equality})} \\
   \typ
     {\Delta, \Theta_3; \theta_3\theta_2\theta_1\Gamma}
     {N : \theta_3A} \;
  (\text{by } \refNum{num:app:type-N-substed})}
  {\typ
     {\Delta, \Theta_3; \theta_3\theta_2\theta_1\Gamma}
     {M\,N : B}\;\ref{proofobl:typed}}
\]

Finally, we show $\Delta \vdash \theta_3 \circ \theta_2 \circ \theta_1 : \Theta
\Rightarrow \Theta_3$.
It follows from \refNum{num:app:theta-one-wf}, \refNum{num:app:theta-two-wf},
\refNum{num:app:theta-three-wf}, and composition~\ref{proofobl:theta-wf}.

\item[Case $\Let \: x  =  M \; \In \; N$:]

By definition of $\meta{infer}$, we have $(\Theta_1, \theta_1, A) =
\Infer(\Delta, \Theta, \Gamma, M)$~\labelNum{num:let:first-infer-call}.
By induction,
this implies
$\typ{\Delta, \Theta_1;\theta_1\Gamma}{M : A}$~\labelNum{num:let:type-M} and
$\Delta \vdash \theta_1 : \Theta \Rightarrow
\Theta_1$~\labelNum{theta-one-wf}.

By definition of $\meta{infer}$ we further have
\begin{equations}
(\Delta'',\Delta''') &= &\mgend[\Delta']{A}{M} \\
         &= &\mgend[(\Delta, (\ftv(\theta_1\Theta)-\Delta))]{A}{M} \\
&&\text{ where $\Delta''' =
               \ftv(A) - (\Delta, (\ftv(\theta_1)-\Delta))
                 \:=\: (\ftv(A) - \Delta) - \ftv(\theta_1)$~\labelNum{num:let:def-delta-two}}
\end{equations}

By applying \cref{lem:inferred-types-are-principal} to
\refNum{num:let:first-infer-call}, we obtain
$\meta{principal}(  (\Delta, \Theta_1 - \Delta'''), \theta_1\Gamma, \Delta''', A )$%
~\labelNum{num:let:A-is-principal}.

We have $\Delta''' \subseteq \Theta_1$ and can therefore rewrite \refNum{num:let:type-M} as
$\typ{\Delta, \Theta_1 - \Delta''', \Delta''';\theta_1\Gamma}{M :
A}$~\labelNum{num:let:type-M-Delta'''-isolated}.

Next, define $\Theta_1' = \dec{demote}(\mono,\Theta_1,\Delta''')$.
Again by
definition of $\dec{infer}$ we have
$(\Theta_2, \theta_2, B) = \Infer(\Delta, \Theta_1' - \Delta'',
(\theta_1(\Gamma), x : \forall \Delta''.A), N)$~\labelNum{num:let:second-infer-use}.

By definition of $\Delta'''$, we have $\Delta''' \fresh \ftv(\theta_1)$ and thus
$\Delta \vdash \theta_1 : \Theta \Rightarrow \Theta_1 -\Delta'''$%
~\labelNum{num:let:theta-one-restricted}.

\newcommand{\freshDeltaThree}{\Delta'''_G}
\newcommand{\freshDeltaTwo}{\Delta''_G}

We distinguish between $M$ being a generalisable value or not. In each case, we
show that there exist $\freshDeltaTwo, \freshDeltaThree, \theta'_2$ and $A'$
such that the following conditions are satisfied:

\begin{align}
&\theta'_2 \circ \theta_1 = \theta_2 \circ \theta_1
  \mathlabelNum{num:let:composition} \\
&\Delta \vdash  \theta'_2 \circ \theta_1 : \Theta \Rightarrow \Theta_2
 \mathlabelNum{num:let:composition-wf} \\
&(\freshDeltaTwo,\freshDeltaThree) = \mgend[(\Delta,\Theta_2)]{\theta'_2 A}{M}
 \mathlabelNum{num:let:gen-eq} \\
&\typ
  {\Delta, \Theta_2, \freshDeltaThree; \theta'_2\theta_1\Gamma}
  {M : \theta'_2 A}
 \mathlabelNum{num:let:type-M-substed} \\
&(\Delta,\Theta_2),\freshDeltaThree,M,\theta'_2 A) \Updownarrow A'
 \mathlabelNum{num:let:updown} \\
&\typ
         {\Delta, \Theta_2; (\theta'_2\theta_1\Gamma,
           x : A')}
         {N : B}
 \mathlabelNum{num:let:type-N} \\
&\meta{principal}( (\Delta, \Theta_2), \theta'_2 \theta_1 \Gamma, M,
      \freshDeltaThree, \theta'_2 A )
 \mathlabelNum{num:let:A-is-principal-substed}
\end{align}

\begin{description}
\item[Sub-Case $M \in \meta{GVal}$:]
By definition of $\meta{gen}$, we  have $\Delta'' = \Delta'''$.
We choose
$\freshDeltaTwo \coloneqq \Delta'''$ and
$\freshDeltaThree \coloneqq \Delta'''$.

In order to apply the induction hypothesis to \refNum{num:let:second-infer-use},
we need to show
$\wfctx
  {\Delta, \Theta_1' - \Delta''}
  {\theta_1(\Gamma), x : \forall \Delta'''.A}$.
First, by \refNum{num:sound:gamma-well-formed,num:let:theta-one-restricted} and
\cref{lem:stability-wf-env}, we have
$\wfctx{\Delta, \Theta_1 - \Delta''}{\theta_1(\Gamma)}$.

Second, by \refNum{num:let:type-M} we have
$\Delta, \Theta_1 \vdash A$ and thus
$\Delta, \Theta_1 - \Delta''' \vdash \forall \Delta'''.A$.
It remains to show that for all $a \in \ftv(A) -\Delta'''$ we have
$\Delta, \Theta_1 \vdash a : \mono$. For $a \in \Delta$, this follows
immediately. Otherwise, we have $a \in \Theta_1 - \Delta'''$ and
$a \in \ftv(\theta_1)$,
which implies that there exists $b \in \Theta$ such that $a \in \ftv(b)$.
If $b \in \ftv(\Gamma)$, then by $\wfctx{\Delta,\Theta}{\Gamma}$ we have
$\Delta, \Theta_1 \vdash \theta(b) : \mono$, which implies
$\Delta, \Theta_1 \vdash a : \mono$. Otherwise, if $b \not\in \ftv(\Gamma)$,
then by \cref{lem:theta-is-id-on-non-Gamma-tyvars} we have
$\theta(b) = b = a$ and $a \not\in \ftv(A)$, contradicting our earlier assumption.
By $\Theta_1 - \Delta'' = \Theta'_1 - \Delta''$ we then have
$\wfctx{\Delta, \Theta'_1 - \Delta''}{\theta_1(\Gamma), x : \forall \Delta''.A}$

In summary, we can apply the induction hypothesis by which we then have
$\typ
  {\Delta, \Theta_2, \theta_2(\theta_1 (\Gamma, x : \forall \Delta''.A)}
  {N : B}$~\labelNum{num:letG:type-N}
and $\Delta \vdash \theta_2 : \Theta_1 - \Delta''' \Rightarrow \Theta_2$%
~\labelNum{num:letG:theta'-two-wf}.
We choose  $\theta'_2 = \theta_2$ and $A' = \forall \Delta'''. \theta'_2 A$,
therefore  satisfying
  \refNum{
          num:let:updown,%
          num:let:composition}.

By \refNum{num:letG:theta'-two-wf,num:let:theta-one-restricted},
 condition \refNum{num:let:composition-wf} is also satisfied.

No type variable in $\Delta'''$ is freely part of the input to $\Infer$ that
resulted in \refNum{num:let:second-infer-use}.  As all newly created variables
are fresh, we then have
$\Delta''' \fresh \Theta_2$~\labelNum{num:letG:Delta''-vs-Theta-two}.

Due to our choice of $\theta'_2$ we have
$\theta'_2(  \theta_1(\Gamma) = \theta_2(\theta_1 \Gamma)$
and by \refNum{num:letG:Delta''-vs-Theta-two,num:letG:theta'-two-wf} also
$\theta_2(\forall \Delta '''. A) =  \forall \Delta'''.\theta'_2(A) $.
Therefore, \refNum{num:letG:type-N} is equivalent to \refNum{num:let:type-N}.

By applying \cref{lem:stability-principality-substitution} to
  \refNum{num:let:A-is-principal,%
          num:letG:theta'-two-wf,%
          num:letG:Delta''-vs-Theta-two} we
show that \refNum{num:let:A-is-principal-substed} is satisfied.

Recall the following relationships:
\begin{align*}
  \ftv(A) &\subseteq \Delta, \Theta_1 \\
  \ftv(\theta) &\subseteq \Delta, \Theta_1 \\
  \Delta''' = \ftv(A) - \Delta - \ftv(\theta) &\subseteq \Theta_1 \\
\end{align*}
Therefore, $\ftv(A) - \Delta - \ftv(\theta)$ (i.e., $\Delta'''$) is equal to
$\ftv(A) - \Delta,(\Theta_1 - \Delta''')$.
This results in $(\Delta'',\Delta''') = \mgend[(\Delta, \Theta_1 - \Delta''')]{A}{M}$%
~\labelNum{num:letG:gen-rewritten}.
Together with \refNum{num:letG:theta'-two-wf} and $\Delta''' \fresh \Theta_2$ we can then
apply \cref{lem:subst-gen-inst}(1) to \refNum{num:letG:gen-rewritten},
yielding satisfaction of \refNum{num:let:gen-eq}.

By applying \cref{lem:subst-gamma-and-A} to
\refNum{num:let:type-M-Delta'''-isolated,num:letG:theta'-two-wf},
we obtain~\refNum{num:let:type-M-substed}.

\item[Sub-Case $M \not\in \meta{GVal}$:]
By definition of $\meta{gen}$, we have $\Delta'' = \cdot$.
Let $\Delta'''$ have the shape $(a_1, \dotsc, a_n)$.
We choose
$\freshDeltaTwo \coloneqq \cdot$ and
$\freshDeltaThree \coloneqq (b_1, \dotsc, b_n)$ for $n$ pairwise different,
 fresh type variables $b_i$.

We show that the induction hypothesis is applicable to
\refNum{num:let:second-infer-use}.
To this end, we show
$\wfctx
   {\Delta, \Theta'_1 - \Delta''}
   { \theta_1\Gamma, x : \forall \Delta''. A }$.
We have $\wfctx{\Delta,\Theta_1}{\theta_1\Gamma}$ and
$\Delta, \Theta_1 \vdash A$ by \refNum{num:let:type-M}.
It remains to show that for all $a \in \ftv(A) $ we have
$(a : \mono) \in \Delta, \Theta'_1$.
If $a \in \Delta'''$, then by definition of $\Theta'_1$ we have
$(a : \mono) \in \Theta'_1$.
Otherwise, if $a \in \Theta_1 - \Delta'''$, we use the same reasoning as in
the case $M \in \dec{GVal}$.

By induction, we then have
$\typ{\Delta, \Theta_2;\theta_2(\theta_1(\Gamma), x : A)}{N : B}$%
~\labelNum{num:letNG:typ-N} and
$\Delta \vdash \theta_2 : \Theta'_1 \Rightarrow \Theta_2$.
By \cref{lem:stability-subst-under-promotion} the latter implies
$\Delta \vdash \theta_2 : \Theta_1 \Rightarrow \Theta_2$%
~\labelNum{num:letNG:theta-two}.

We define $\theta'_2$ such that
\[
\theta'_2(c) =
\begin{cases}
b_i & \text{if } c = a_i \in \Delta''' \\
\theta_2(c) &\text{if } c \in \Theta_1 - \Delta'''
\end{cases}
\]

By \refNum{num:letNG:theta-two} and the definition of $\Delta'''$ we then have
$\Delta \vdash \theta'_2 : \Theta_1 \Rightarrow \Theta_2, \freshDeltaThree$%
~\labelNum{num:letNG:theta'-two-wf}.
Observe that we have
$\theta'_2(a) = \theta_2(a)$ for all $a \in \ftv(\theta_1) - \Delta$
 and therefore
\refNum{num:let:composition} as well as \refNum{num:let:composition-wf} are
satisfied.

Furthermore, we define $\theta''_2$ such that
$\theta''_2(a) = \theta_2(a)$ for all $a \in \Theta_1 - \Delta'''$,
which implies
$\Delta \vdash \theta''_2 : \Theta_1 - \Delta''' \Rightarrow \Theta_2$%
  ~\labelNum{num:letNG:theta''-two-wf}
and $\theta''_2 \circ \theta_1 = \theta_2 \circ \theta_1$%
  ~\labelNum{num:letNG:theta''-two-composition}.

We define the instantiation $\rsubst$ such that $\rsubst(b_i) = \theta_2(a_i)$
for all $a_i \in \Delta'''$. By definition of $\Theta'_1$ and
\refNum{num:letNG:theta-two} we then have
$\Delta, \Theta_2 \vdash \rsubst(b_i) : \mono$ for all
$b_i \in \freshDeltaThree$. This implies
$\Delta \vdash \rsubst : \freshDeltaThree \Rightarrow_\mono \Theta_2$.

We define $A' \coloneqq \rsubst(\theta'_2(A))$, which is identical to
$\theta_2(A)$. Together with
$\theta_2(\theta_1\Gamma) = \theta'_2(\theta_1\Gamma)$, this choice
satisfies \refNum{num:let:updown} and makes \refNum{num:letNG:typ-N} equivalent
 to \refNum{num:let:type-N}.

We have $\ftv(\theta_2 A) \subseteq \Delta, \Theta_2$ and
$\Delta''' \subseteq \ftv(A)$. By $\theta'_2(\Delta''') = \freshDeltaThree$
we have $\freshDeltaThree \subseteq \ftv(\theta'_2 A)$. Together with
$\ftv( \theta'_2(a) ) = \ftv( \theta_2(a) ) \fresh \freshDeltaThree$
holding for all $a \in \ftv(A) - \Delta'''$, we then
have $\ftv( \theta'_2 A  ) - \Delta, \Theta_2 = \freshDeltaThree$.
Therefore, we have
$\mgend[(\Delta, \Theta_2)]{\theta'_2 A}{M} = (\cdot, \freshDeltaThree)$,
 satisfying \refNum{num:let:gen-eq}.

Let $\rsubst_F$ be defined such that $\rsubst(a_i) = b_i$ for all
$1 \le i \le n$, which implies
$\Delta, \Theta \vdash \rsubst_F : \Delta''' \Rightarrow_\mono \freshDeltaThree$
and $ \theta''_2 \rsubst_F (A) = \theta'_2(A)$~\labelNum{num:letNG:rsubstF-vs-theta'}
(by weakening $\theta''_2$ such that
$\Delta, \freshDeltaThree \vdash
  \theta_2'' : \Theta_1 - \Delta''' \Rightarrow \Theta_2$ )
Using \cref{lem:principal:rename-gen-vars}, we then get
$\meta{principal}( (\Delta, \Theta_1 - \Delta'''),  \theta_1 \Gamma, \allowbreak M,
   \freshDeltaThree, \allowbreak \rsubst_F A )$,

We apply \cref{lem:stability-principality-substitution} to this freshened
principality statement and \refNum{num:letNG:theta''-two-wf},
which gives us
$\meta{principal}( (\Delta, \Theta_2), \allowbreak \theta''_2 \theta_1 \Gamma,
    \allowbreak M, \allowbreak \Delta''', \allowbreak \theta''_2 \rsubst_F A )$.

Using \refNum{num:letNG:rsubstF-vs-theta'}, we restate this as
$\meta{principal}( (\Delta, \Theta_2), \theta''_2 \theta_1 \Gamma, \allowbreak M,
   \freshDeltaThree, \allowbreak \theta'_2 A )$,
which by
\refNum{%
        num:letNG:theta''-two-composition,%
        num:let:composition}
is equivalent to \refNum{num:let:A-is-principal-substed}.

By applying \cref{lem:subst-gamma-and-A} to \refNum{num:let:type-M,num:letNG:theta'-two-wf}, we obtain~\refNum{num:let:type-M-substed}.

\end{description}

We have shown that
\refNum{%
num:let:composition,%
num:let:composition-wf,%
num:let:gen-eq,%
num:let:type-M-substed,%
num:let:updown,%
num:let:type-N,%
num:let:A-is-principal-substed,%
}
hold in each case.
We can now derive the following:
\[
\inferrule*[Right=\freezemlLab{Let}]
    {
      (\freshDeltaTwo,\freshDeltaThree) = \mgend[(\Delta,\Theta_2)]{\theta'_2 A}{M}\;
        (\text{by }
          \refNum{num:let:gen-eq})\\\\
      \typ
        {\Delta, \Theta_2, \freshDeltaThree; \theta'_2\theta_1\Gamma}
        {M : \theta'_2 A} \;
        (\text{by } \refNum{num:let:type-M-substed})
\\\\
        ((\Delta,\Theta_2),\freshDeltaThree,M,\theta'_2 A) \Updownarrow A' \;(\text{by } \refNum{num:let:updown})
\\\\
       \typ
         {\Delta, \Theta_2; (\theta'_2\theta_1\Gamma,
           x : A')}
         {N : B} \;
         (\text{by }
           \refNum{num:let:type-N})
\\\\\
       \dec{principal}((\Delta,\Theta_2),\theta'_2\theta_1\Gamma,M,\freshDeltaThree,\theta'_2 A) \; (\text{by } \refNum{num:let:A-is-principal-substed})
    }
    {\typ
      {\Delta, \Theta_2; \theta'_2\theta_1\Gamma}
      {\Let \; x = M\; \In \; N : B}
    }
\]

By
\refNum{num:let:composition,num:let:composition-wf} we have therefore shown
\ref{proofobl:typed} and \ref{proofobl:theta-wf}.

\item[Case $\Let \; (x : A) = M \; \In \; N$]:
Let $A = \forall \Delta''.H$ for appropriate $\Delta''$ and $H$.
By alpha-equivalence, we assume $\Delta'' \fresh \Theta$.
According to \refNum{num:sound:term-well-formed}, we have
$\Delta \vdash A$~\labelNum{let-asc:A-wf}.

We distinguish between whether of not $M$ is a guarded value.
We show that the following conditions hold
for the choice of $A'$ and $\Delta'$ imposed by
$(\Delta', A') = \msplit(A, M)~\labelNum{let-asc:split}$ in each case.
\begin{gather}
\Delta, \Delta' \vdash A' \mathlabelNum{let-asc:A'-wf} \\
\ftv(A) \fresh \Delta' \mathlabelNum{let-asc:Delta'-vs-ftv-A} \\
\Delta' \fresh \Theta \mathlabelNum{let-asc:Delta'-vs-Theta}
\end{gather}

\begin{description}

\item[Sub-Case $M \in \dec{GVal}$:]
We have $\msplit(A, M) = (\Delta'',H)$ (i.e, $\Delta'= \Delta''$ and $A' = H$).

Together with \refNum{let-asc:A-wf}
we have $\Delta, \Delta' \vdash H$ (satisfying \refNum{let-asc:A'-wf}).
Assumption $\Delta'' \fresh \Theta$ satisfies \refNum{let-asc:Delta'-vs-Theta}.
By $A = \forall \Delta'.A'$ we further have $\ftv(A) \fresh \Delta'$.

\item[Sub-Case $M \not\in \dec{GVal}$:]
We have $\msplit(A, M) = (\cdot,A)$ (i.e, $\Delta'= \cdot$ and $A' = A$).
This immediately satisfies
\refNum{let-asc:Delta'-vs-Theta,let-asc:Delta'-vs-ftv-A}.
It further makes \refNum{let-asc:A-wf} equivalent to \refNum{let-asc:A'-wf}.

\end{description}

Moreover, by \refNum{num:sound:term-well-formed}, we  have
$\termwf{\Delta,\Delta'}{M}$~\labelNum{let-asc:termwf} using inversion.

We show that
  $\typ{\Delta, \Theta_1, \Delta'; \theta_1\Gamma}{M : A_1}$%
    ~\labelNum{num:sound:let-asc:type-M}
holds.
By~\refNum{num:sound:gamma-well-formed} and since $\Delta' \fresh \Theta$,
we have $\wfctx{\Delta,\Delta', \Theta}{\Gamma}$.
Together with \refNum{let-asc:termwf}, we then have $\typ{\Delta,
\Delta',\Theta_1;\theta_1' \Gamma}{M : A_1}$ and
$\Delta, \Delta'  \vdash \theta_1 : \Theta \Rightarrow \Theta_1$~\labelNum{let-asc:theta-one-wf}
by
induction.
Further, this indicates $\Delta' \fresh \Theta_1$.

By \refNum{num:sound:let-asc:type-M} we also have $\Delta, \Delta',
\Theta_1 \vdash A_1$.
Recall $\Delta, \Delta' \vdash A'$
 and therefore
$\Delta, \Delta', \Theta_1 \vdash A'$.
Thus, by \cref{thm:unification-sound}, we have $\theta'_2(A_1) =
\theta'_2(A')$~\labelNum{num:let-asc:unify-eq} and $\Delta, \Delta' \vdash
\theta'_2 : \Theta_1 \Rightarrow
\Theta_2$~\labelNum{num:let-asc:theta-two-prime-wf}.

According to the assertion, we have $\ftv(\theta_2'\circ \theta_1) \mathop{\#} \Delta'$%
  ~\labelNum{num:let-asc:theta-two-one-disjoint-delta-prime}.
By definition of $\Infer$, we have $\theta_2 = \theta'_2 \circ
\theta_1$~\labelNum{num:let-asc:def-theta-two}, yielding
$\Delta,\Delta' : \theta_2 : \Theta \Rightarrow \Theta_2$,
which further implies
$\Delta' \fresh \Theta_2 \labelNum{let-asc:Delta'-vs-Theta-two}$.
By \refNum{num:let-asc:theta-two-one-disjoint-delta-prime}, we can strengthen
$\theta_2$ s.t.
$\Delta \vdash \theta_2 : \Theta \Rightarrow
\Theta_2$~\labelNum{num:let-asc:theta-two-one-composed-wf}.

By \refNum{num:sound:let-asc:type-M,num:let-asc:theta-two-prime-wf,%
 let-asc:termwf}, and
\cref{lem:subst-gamma-and-A}, we have $\typ{\Delta, \Delta',
\Theta_2;\theta_2'\theta_1\Gamma}{M : \theta'_2 A_1}$.
By \refNum{num:let-asc:def-theta-two,%
  let-asc:theta-one-wf,%
  num:let-asc:theta-two-prime-wf}, this is equivalent to $\typ{\Delta, \Delta',
\Theta_2;\theta_2\Gamma}{M : \theta'_2
A_1}$~\labelNum{num:let-asc:type-M-substed}.

By \refNum{let-asc:A'-wf} and \refNum{num:let-asc:theta-two-prime-wf}
we have $\theta'_2(A') = A'$.
Together with \refNum{num:let-asc:unify-eq}, this makes
\refNum{num:let-asc:type-M-substed} equivalent to
$\typ{\Delta, \Delta',\Theta_2;\theta_2\Gamma}{M : A'}$%
~\labelNum{let-asc:typ-M-substed}.

By definition of $\meta{infer}$, we have $(\Theta_3, \theta_3, B) =
\Infer(\Delta, \Theta_2, (\theta_2\Gamma, x : A), N)$.
Due to \refNum{num:sound:term-well-formed}, we have $\termwf{\Delta}{N}$.
By \refNum{let-asc:A-wf} and $\Theta_2 \fresh \Delta$, we have $\wfctx{\Delta, \Theta_2}{x : A}$.
Together with \refNum{num:sound:gamma-well-formed,num:let-asc:theta-two-one-composed-wf}
we then have $\wfctx{\Delta, \Theta_2}{(\theta_2\Gamma, x : A)}$.
  Therefore,
by induction, we have
$\typ{\Delta, \Theta_3;\theta_3(\theta_2\Gamma, x : A)}{N :
B}$~\labelNum{num:let-asc:type-N} and $\Delta \vdash \theta_3: \Theta_2
\Rightarrow \Theta_3$~\labelNum{num:let-asc:theta-three}.

According to \refNum{let-asc:Delta'-vs-ftv-A,let-asc:Delta'-vs-Theta-two},
none of the variables in $\Delta'$ are freely part of the input to $\meta{infer}$,
yielding $\Delta' \fresh \Theta_3$.
Together with \refNum{let-asc:Delta'-vs-Theta-two}, we can then weaken
\refNum{num:let-asc:theta-three} to
$\Delta, \Delta' \vdash \theta_3: \Theta_2 \Rightarrow \Theta_3$.
By the latter,
\refNum{let-asc:typ-M-substed},
\refNum{let-asc:termwf}, and
\cref{lem:subst-gamma-and-A}, we have $\typ{\Delta, \Theta_3,
\Delta';\theta_3\theta_2\Gamma}{M : \theta_3 A'}$~\labelNum{num:let-asc:type-M-theta-three}.

Using a similar line of reasoning as before,
we have $\theta_3(A') = A'$~\labelNum{num:let-asc:theta-three-on-A-two} and $\theta_3(A) = A$~\labelNum{let-asc-theta-three-A}.



By \refNum{num:let-asc:theta-two-one-composed-wf},
\refNum{num:let-asc:theta-three}, and composition, we have $\Delta \vdash
\theta_3 \circ \theta_2 : \Theta \Rightarrow \Theta_3$.~\ref{proofobl:theta-wf}.

Together with \refNum{num:sound:gamma-well-formed}, we
obtain $\Delta, \Theta_3 \vdash \theta_3 \theta_2  \Gamma$ and can
derive the following:

\[
\inferrule*[Right=\freezemlLab{Let-Ascribe}]
    {(\Delta', A') = \msplit(A, M) \;
      (\text{by \refNum{let-asc:split}} ) \\\\
      A = \forall \Delta'.A' (\text{by \refNum{let-asc:split}} ) \\\\
     \typ{\Delta,\Theta_3, \Delta'; \theta_3\theta_2\Gamma}{M : A'} \;
       (\text{by
           \refNum{num:let-asc:type-M-theta-three,%
           num:let-asc:theta-three-on-A-two}})\\\\
     \typ
       {\Delta, \Theta_3; \theta_3\theta_2\Gamma, x : A}
       {N : B} \;
       (\text{by  \refNum{num:let-asc:type-N,let-asc-theta-three-A}})
    }
    {\typ
      {\Delta, \Theta_3; \theta_3\theta_2\Gamma}
      {\Let \; (x : A) = M\; \In \; N : B \; \ref{proofobl:typed} }
    }
\]

\end{description}

\end{proof}

\subsection{Completeness of type inference}
\label{subsec:correctness-inference:completeness-mg}

\setcounter{thmbackup}{\thethm}
\setcounterref{thm}{thm:inference-completeness-mg}
\addtocounter{thm}{-1}

\begin{thm}[Type inference is complete and principal]
Let $\termwf{\Delta}{M}$ and $\wfctx{\Delta, \Theta}{\Gamma}$.
If $\Delta \vdash \theta_0 : \Theta \Rightarrow \Theta'$ and $\typ{\Delta,
  \Theta'; \theta_0(\Gamma)}{M : A_0}$, then $\Infer(\Delta, \Theta, \Gamma, M)
\allowbreak = (\Theta'', \theta', A_R)$ where there exists $\theta''$
satisfying $\Delta \vdash \theta'' : \Theta'' \Rightarrow \Theta'$ such that
$\theta_0 = \theta'' \circ \theta'$
and $\theta''( A_R) = A_0$.
\end{thm}

\setcounter{thm}{\thethmbackup}

\begin{proof}
\resetNum
\proofContext{compl-infer}

By induction on the structure of $M$.
In each case, we assume
  $\termwf{\Delta}{M}$~\labelNum{term-wf}, and
  $\wfctx{\Delta, \Theta}{\Gamma}$~\labelNum{gamma-wf}, and
  $\Delta \vdash \theta_0 : \Theta \Rightarrow \Theta'$~\labelNum{theta-well-formed}, and
  $\typ{\Delta,\Theta';\theta_0\Gamma}{M : A_0}$~\labelNum{typed},
  which implies
  $\Delta,\Theta' \vdash \theta_0\Gamma$~\labelNum{theta-gamma-well-formed}, and
  $\Delta,\Theta' \vdash A_0$~\labelNum{res-type-well-formed}.
For each case, we show:
\begin{enumerate}[label=\Roman*., ref=\Roman*]
\item \label{proofobl:compl:infer-def}
  $\Infer(\Delta,\Theta,\Gamma,M) = (\Theta'', \theta', A_R)$
\item \label{proofobl:compl:theta-two-primes-wf}
  $\Delta \vdash \theta'' : \Theta'' \Rightarrow \Theta' $
\item \label{proofobl:compl:theta-composition}
   $\theta_0 = \theta'' \circ \theta'$
\item \label{proofobl:compl:subsumption}
   $\theta''( A_R) = A_0$
\end{enumerate}
%
We reference the proof obligations above to indicate when we have shown them.

\begin{description}
\item[Case $\freeze{x}$:]

By \refNum{typed} and \freezemlLab{Freeze}, we have $(x : A_0) \in
\theta_0\Gamma$.
$\Infer$ succeeds, and we have $\Theta'' =
\Theta$, $\theta' = \idsubst_{\Delta,\Theta}$, and $A_R = \Gamma(x)$.
The latter implies $A_0 = \theta_0(A_R)$.

We have $\Delta \vdash \theta' : \Theta \Rightarrow \Theta$.
Let $\theta'' := \theta_0$.
By \refNum{theta-well-formed} we then have $\Delta \vdash \theta'' :
\Theta \Rightarrow \Theta'$~(\ref{proofobl:compl:theta-two-primes-wf}).
We observe $\theta_0 = \theta'' = \theta'' \circ \idsubst_{\Delta,\Theta} = \theta''
\circ \theta'$~(\ref{proofobl:compl:theta-composition}).

Finally, this yields $\theta''(A_R) = \theta_0(A_R) = A_0$~%
(\ref{proofobl:compl:subsumption}).

\item[Case $x$:]
The derivation for \refNum{typed}
must be of the following form:

\[
  \inferrule*[Lab=\freezemlLab{Var}]
    {{x : \forall \Delta'.H' \in \theta_0 \Gamma} \\
      \Delta, \Theta' \vdash \rsubst : \Delta' \Rightarrow_\poly \cdot
    }
    {\typ{\Delta, \Theta'; \theta_0 \Gamma}{x : \rsubst(H')}}
\]

Therefore, there exists $x : \forall \Delta''.H \in \Gamma$
such that $\forall \Delta'.H' = \theta_0(\forall \Delta''.H)$.
By alpha-equivalence, we assume that $\Delta''$ is fresh,
yielding
$\theta_0(\forall \Delta''.H) = \forall \Delta''. \theta_0 H$.
By \refNum{gamma-wf}, all free type variables in $H$ are monomorphic,
meaning that $\theta_0 H$ cannot have toplevel quantifiers.
Thus, the quantifier structure is preserved by $\theta_0$;
in particular $\Delta' = \Delta''$, $H' = \theta_0 (H)$.

Further, due to our freshness assumption about $\Delta'' = \Delta'$,
we have
$\Delta' \fresh \Delta,\Theta $~\labelNum{var:Delta-prime-Theta-disj}
and $\Delta' \fresh \Delta, \Theta'$.


In total, we have $A_0 =
\rsubst\theta_0 H$~\labelNum{var:def-overline-A} and
$\Gamma(x) = \forall \Delta'.H$~\labelNum{var:x-in-Gamma} and
$\theta_0\Gamma(x) = \theta_0(\forall \Delta' . H ) = \forall
\Delta'.\theta_0 H$~\labelNum{var:Gamma}.

Let $\Delta' = \many{a} = (a_1, \dotsc, a_n)$ with corresponding fresh $\many{b}
= (b_1, \dotsc, b_n)$ for some $n \ge 0$.
Then $\Infer$ succeeds with
  $\Theta'' = (\Theta, \many{b : \poly})$, and
  $\theta' =\idsubst_{\Delta, \Theta}$%
   ~\labelNum{var:def-theta-one-prime},
and
  $A_R =  H[\many{b}/\many{a}]$~\labelNum{var:A-prime-def}.
Due to $\Theta \subseteq \Theta''$ and the freshness of $\many{b}$, we have
$\Delta \vdash \theta' : \Theta \Rightarrow
\Theta''$~\labelNum{var:theta-one-prime-typ}.

We define $\theta''$ such that
\begin{equation}
\theta''(c) =
  \begin{cases}
    \theta_0(c) &\text{if } c \in \Theta \\
    \delta(a_i) &\text{if } c = b_i \text{ for some }b_i \in \many{b}
  \end{cases} \locallabel{var:def-theta-two-primes}
\end{equation}

By~%
  \refNum{theta-well-formed} and~%
  \refNum{var:def-theta-two-primes},
for all $(c : K) \in \Theta$ we have
$\Delta, \Theta' \vdash \theta_0(c) : K$~%
  \labelNum{var:theta-two-primes-on-Theta}.
By $\Delta, \Theta' \vdash \rsubst : \Delta' \Rightarrow_\poly \cdot$, we have
$\Delta, \Theta' \vdash \rsubst(a) : \poly$ for all
$a \in \Delta'$ and thus
$\Delta, \Theta' \vdash \theta''(b) : \poly$ for all $b \in \many{b}$.
Together, we then have $\Delta \vdash \theta'' : \Theta'' \Rightarrow
\Theta'$~(\ref{proofobl:compl:theta-two-primes-wf}).

By \refNum{var:def-theta-one-prime} and
\refNum{var:def-theta-two-primes}, we have $\theta''\theta'(c) =
\theta''(c) = \theta_0(c)$ for all $c \in
\Theta$~(\ref{proofobl:compl:theta-composition}).

It remains to show that $\theta''(H[\many{b}/\many{a}]) = A_0 = \rsubst(\theta_0(H))$.

By \refNum{var:x-in-Gamma,gamma-wf}, we have $\Delta, \Theta \vdash \forall \Delta'.H$
and further $\Delta, \Theta, \Delta' \vdash H$.

We show that for all $c \in \ftv(H) \subseteq \Delta, \Theta, \Delta'$, we have
$\theta''(c[\many{b}/\many{a}]) = \rsubst\theta_0 (c)$~\labelNum{var:ftv-eq}. We
distinguish three cases:

\begin{enumerate}
\item
Let $c = a_i \in \Delta'$.
We then have
\[
\begin{array}{rrcll}
  &a_i &= &\theta_0(a_i)
    \reason{by \refNum{theta-well-formed}}\\
\text{implies} \quad
  &\theta''(a_i[\many{b}/\many{a}]) &= &\rsubst(\theta_0 (a_i))
      &\reason{by \refNum{var:def-theta-two-primes}:
          $\theta''(b_i) = \rsubst(a_i)$} \\
\end{array}
\]

\item
Let $c \in \Theta$.
We then have
\[
\begin{array}{rrcll}
  &\theta_0(c) &= &\theta_0(c) \\
\text{implies} \quad
  &\theta_0(c[\many{b}/\many{a}]) &= &\theta_0(c)
     &\reason{by \refNum{var:Delta-prime-Theta-disj}:
       $c[\many{b}/\many{a}] = c$} \\
\text{implies} \quad
  &\theta''(c[\many{b}/\many{a}]) &= &\rsubst\theta_0 (c)
     &\reason{\text{by} \refNum{var:def-theta-two-primes} :
       $\theta''(c) = \rsubst(c)$ for all $c \in \Theta$}
\end{array}
\]

\item
Let $c \in \Delta$. Then all involved substitutions/instantiations return $c$ unchanged.
\end{enumerate}

By \refNum{var:A-prime-def} and \refNum{var:def-overline-A},
\refNum{var:ftv-eq} then yields $\theta''(A_R) = A_0$
(\ref{proofobl:compl:subsumption}).

\item[Case $\lambda x.M$:]
By \refNum{typed} and \freezemlLab{Lam}, we have $A_0 = S' \to B'$
for some $S', B'$ as well as $\typ{\Delta, \Theta'; \theta_0\Gamma, (x :
S')}{M : B'}$~\labelNum{lam:inversion}.
The latter implies $\Delta, \Theta' \vdash S' : \mono$~\labelNum{lam:S-wf}.

Let $a$ be the fresh variable as in the definition of $\Infer$; in particular $a
\mathop{\#} \Theta$~\labelNum{lam:a-and-theta-disjoint}.
Let $\theta_a$ be defined such that
$\theta_a(b) = \theta_0(b)$ for all $b \in
\Theta$~\labelNum{lam:theta-a-def} and $\theta_a(a) = S'$~\labelNum{lam:def-tilde-theta-a}.
By \refNum{theta-well-formed} and \refNum{lam:S-wf}, we have
$\Delta \vdash \theta_a : (\Theta, a : \mono) \Rightarrow \Theta'$%
  \labelNum{lam:theta-a-wf}.
This definition makes \refNum{lam:inversion} equivalent to
$\typ{\Delta, \Theta';\theta_a(\Gamma, x : a)}{M : B'}$.

By induction, we therefore have that
$\Infer(\Delta, (\Theta, a : \mono), (\Gamma, x : a), M)$
succeeds~\labelNum{lam:rec-infer-success}, returning $(\Theta_1,
\theta'_1, B )$ and there exists $\theta''_1$ s.t.\
\begin{gather}
\Delta \vdash \theta''_1 : \Theta_1 \Rightarrow
\Theta' \mathlabelNum{lam:rec-infer-substitution-wf}
  \\
\theta_a = \theta''_1 \circ
\theta'_1 \mathlabelNum{lam:rec-infer-substitution-composition}
 \\
\theta''_1(B) =
B' \mathlabelNum{lam:rec-infer-subsumption}
\end{gather}

By \cref{thm:inference-sound}, we have
$\Delta \vdash
\theta'_1 : (\Theta, a : \mono) \Rightarrow \Theta_1$~\labelNum{theta'-one-wf}%
\footnote{
Observe that we cannot deduce this from
\refNum{lam:rec-infer-substitution-composition,%
lam:rec-infer-substitution-wf}.
A counter-example would be the following:
$\Theta = (a : \mono)$,
$\Theta''= (b : \poly)$,
$\theta' = (c : \mono)$,
$\theta' = [a \mapsto b]$,
$\theta'' = [b \mapsto c]$.
We have $\vdash (\theta'' \circ \theta') : \Theta \Rightarrow \Theta'$ and
$\vdash \theta'' : \Theta'' \Rightarrow \Theta'$, but not
$\vdash \theta' : \Theta' \Rightarrow \Theta''$.
}.
By preservation of kinds under substitution, we have
$\Delta, \Theta_1 \vdash \theta'_1(a) : \mono$.
This implies that $\theta'_1(a)$ is a syntactic monotype.
Thus, $\theta'_1 = \theta[a \mapsto S]$~\labelNum{lam:theta-def} is
well-defined, yielding a substitution $\Delta \vdash \theta : \Theta \Rightarrow
\Theta_1$.
Hence, all steps of $\Infer$
succeed.

According to the return values of $\Infer$, we have $A_R = S \to B$, $\Theta'' =
\Theta_1$, and $\theta' = \theta$~\labelNum{lam:def-theta-prime}.

Let $\theta''$ be defined as
$\theta''_1$~\labelNum{lam:def-theta-two-primes}.
By \refNum{lam:rec-infer-substitution-wf}, this choice immediately
satisfies~(\ref{proofobl:compl:theta-two-primes-wf}).

We show (\ref{proofobl:compl:theta-composition}) as follows:
Let $b \in \Theta$. We then have
\[
\begin{array}{cll}
   &\theta_0(b) \\
=  &\theta_a(b)
      &\reason{by \refNum{lam:theta-a-def,%
                  lam:a-and-theta-disjoint}}\\
=  &\theta''_1 \theta'_1 (b)
      &\reason{by \refNum{
        lam:rec-infer-substitution-composition,%
        lam:theta-a-wf,lam:rec-infer-substitution-wf,%
        theta'-one-wf%
        }}\\
=  &\theta''_1 \theta (b)
      &\reason{by \refNum{lam:a-and-theta-disjoint},
                  \refNum{lam:theta-def}} \\
=  &\theta'' \theta'(b)
      &\reason{by \refNum{lam:def-theta-prime},
                  \refNum{lam:def-theta-two-primes}}
\end{array}
\]

By~\refNum{lam:theta-def}, we have $\theta'_1(a) = S$.
By~\refNum{lam:def-tilde-theta-a}, we have $\theta_a(a) = S'$.
By~\refNum{lam:rec-infer-substitution-composition} we  therefore have
$\theta''_1(S) = \theta_a(a) = S'$.
Together with~\refNum{lam:rec-infer-subsumption}, $A_0 = S' \to B'$, and
$A_R = S \to B$ we have shown~(~\ref{proofobl:compl:subsumption}).

\item[Case $\lambda (x: A). M$:]

This case is analogous to the previous one;
the only difference is as follows:

By \refNum{typed} and \freezemlLab{Lam-Ascribe}, we have $A_0 = A \to B'$
for some $ B'$ as well as $\typ{\Delta, \Theta'; \theta_0\Gamma, (x :
A)}{M : B'}$.
However, by \refNum{term-wf}, we have $\Delta \vdash A$
and therefore $\theta_0(A) = A$.

Hence, we can apply the induction hypothesis directly to the typing
judgement above, rather than having to construct $\theta_a$.

\item[Case $M \: N$:]
By \refNum{typed} and \freezemlLab{App}, we have
$\typ{\Delta, \Theta'; \theta_0\Gamma }{M : A_N \to A_0}$ and
$\typ{\Delta, \Theta'; \theta_0\Gamma }{N :
A_N}$~\labelNum{app:type-N} for some type $A_N$.
The former implies $\Delta, \Theta' \vdash A_0$~\labelNum{app:A-zero-wf}

By induction, $\Infer(\Delta, \Theta, \Gamma, M)$ succeeds, returning
$(\Theta_1, \theta_1, A')$ and there exists $\theta''_1$ such that the
following conditions hold:
\begin{align}
&\Delta \vdash \theta''_1 : \Theta_1 \Rightarrow
  \Theta' \mathlabelNum{app:rec-call-one-theta-two-primes-wf}\\
&\theta_0 = \theta''_1 \circ
  \theta_1 \mathlabelNum{app:rec-call-one-composition} \\
&\Delta, \Theta' \vdash \theta''_1(A') = A_N \to
A_0 \mathlabelNum{app:rec-call-one-subsumption}
\end{align}

By \refNum{app:rec-call-one-subsumption}, $A'$ must not have
toplevel quantifiers.
Let $B_N$ and $B_M$ such that $A' = B_N \to
B_M$~\labelNum{app:def-B-M-and-B-N}.
This yields $\theta''_1(B_N) =
A_N$~\labelNum{app:B-N-vs-A-N} and $\theta''_1(B_M) =
A_0$~\labelNum{app:B-M-vs-overline-A}.

By \cref{thm:inference-sound}, we have $\Delta \vdash \theta_1
: \Theta \Rightarrow \Theta_1$~\labelNum{app:theta-one-wf}
and $\typ{\Delta, \Theta_1;\theta_1(\Gamma)}{M : A'}$,
which implies $\Delta, \Theta_1 \vdash A'$.
By choosing $b$ as fresh, we have $b \mathop{\#} \Delta$, and $b \mathop{\#}
\Theta_1$, and $b \mathop{\#} \Theta_2$ and $b \mathop{\#}
\Theta'$~\labelNum{app:b-disjointness}

By \refNum{app:rec-call-one-composition}, we can rewrite
\refNum{app:type-N} as $\typ{\Delta, \Theta'; \theta''_1\theta_1
\Gamma }{N : A_N}$.
By induction (using \refNum{app:rec-call-one-theta-two-primes-wf}),
we then have that $\Infer(\Delta, \Theta_1, \theta_1\Gamma, N)$
succeeds, returning $(\Theta_2, \theta_2, A)$ and
there exists $\theta''_2$ such that
\begin{align}
&
\Delta \vdash \theta''_2 : \Theta_2 \Rightarrow
\Theta' \mathlabelNum{app:rec-call-two-theta-two-primes-wf}\\
&
\theta''_1 = \theta''_2 \circ
\theta_2 \mathlabelNum{app:rec-call-two-composition} \\
&
\Delta,\Theta' \vdash \theta''_2(A) =
A_N \mathlabelNum{app:rec-call-two-subsumption}
\end{align}

By \cref{thm:inference-sound}, $\Delta \vdash \theta_2 :
\Theta_1 \Rightarrow \Theta_2$~\labelNum{app:theta-two-wf} as well as
$\typ{\Delta, \Theta_2;\theta_2 \theta_1\Gamma}{N : A}$,
which implies $\Delta, \Theta_2 \vdash A$~\labelNum{app:A-wf}.

Let $\theta_b$ be defined such that
\begin{equation}
\theta_b(c) =
  \begin{cases}
    \theta''_2(c) &\text{if } c \in \Theta_2 \\
    \theta''_2\theta_2(B_M) &\text{if } c = b
  \end{cases} \mathlabelNum{app:def-theta-b}
\end{equation}

We have $\theta_b(b) = \theta''_2\theta_2(B_M) = \theta''_1(B_M) = A_0$%
  \labelNum{app:theta-b-on-b}.
By \refNum{app:rec-call-two-theta-two-primes-wf,app:A-zero-wf}
we thus have $\Delta \vdash \theta_b : (\Theta_2, b : \poly)
\Rightarrow \Theta'$.
Due to \refNum{app:A-wf}, we further have $\theta_b(A) = \theta''_2(A)$%
  ~\labelNum{app:theta-b-on-A}.

We show applicability of the completeness of  unification theorem:
\[
\begin{array}{lrcll}
&\multicolumn{3}{c}{\theta_b\theta_2(A')} \\
=\qquad
  &\theta_b\theta_2(B_N) &\to &\theta_b\theta_2(B_M)
  \qquad&\reason{by \refNum{app:def-B-M-and-B-N}} \\
=
  &\theta''_2\theta_2(B_N) &\to &\theta''_2\theta_2(B_M)
         &\reason{by \refNum{app:def-theta-b,app:b-disjointness}} \\
=
&\theta''_1(B_N) &\to &\theta''_2\theta_2(B_M)
      &\reason{by \refNum{app:rec-call-two-composition}}\\
=
&A_N &\to &\theta''_2\theta_2(B_M)
    &\reason{by \refNum{app:B-N-vs-A-N}} \\
=
  &\theta''_2(A) &\to &\theta''_2\theta_2(B_M)
     &\reason{by \refNum{app:rec-call-two-subsumption}} \\
=

     &\theta_b(A) &\to &\theta_b(b)
     &\reason{by \refNum{app:def-theta-b,app:theta-b-on-A}} \\
=
     &\theta_b(A &\to &b) \;
     &\reason{by \refNum{app:def-B-M-and-B-N}}
 \end{array}
\]

By
the equality above as well as
$\Delta, \Theta_2 \vdash \theta_2(A')$ and
$\Delta, \Theta_2, b : \poly \vdash (A \to b)$,
 Theorem~\ref{thm:thmunifycomplete} states that
$\unify(\Delta, (\Theta_2, b : \poly), \theta_2(H), A \to b)$ succeeds,
returning $(\Theta_3, \theta'_3)$, and there exists $\theta''_3$ such that
$\Delta \vdash \theta''_3 : \Theta_3 \Rightarrow
\Theta'$~\labelNum{app:theta-U-wf} and $\theta_b = \theta''_3 \circ
\theta'_3$~\labelNum{app:theta-U-wf-composition}.
The latter implies $\Delta \vdash \theta'_3 : (\Theta_2, b : \poly) \Rightarrow
\Theta_3$.
This makes defining $\theta'_3 = \theta_3[b \to
B]$~\labelNum{app:def-theta-three-from-theta-three-prime}
succeed, resulting in $\Delta \vdash \theta_3 : \Theta_2 \Rightarrow
\Theta_3$~\labelNum{app:theta-three-wf}.

Observe that $\theta''_2$ arises from $\theta_b$ in the same
way as $\theta_3$ arises from $\theta_3'$ by removing $b$
from its domain.
Therefore, \refNum{app:theta-U-wf-composition}
yields $\theta''_2 = \theta''_3 \circ \theta_3$%
~\labelNum{app:theta-U-wf-composition-simpl}.

By \refNum{app:theta-one-wf}, \refNum{app:theta-two-wf},
\refNum{app:theta-three-wf}, and composition, we have $\Delta \vdash
\theta_3 \circ \theta_2 \circ \theta_1 : \Theta \Rightarrow \Theta_3$.

We have shown that all steps of the algorithm
succeed and it returns $(\Theta'', \theta',
A_R) = (\Theta_3, \theta_3 \circ \theta_2 \circ \theta_1,
B)$~\labelNum{app:def-return-vals}.

Let $\theta''$ be defined as $\theta''_3$, satisfying
(\ref{proofobl:compl:theta-two-primes-wf}), by
\refNum{app:theta-U-wf}.

We show satisfaction of (\ref{proofobl:compl:theta-composition})
as follows:
\begin{alignat*}{3}
  &\theta_0 \\
=\;\; &\theta''_1 \circ \theta_1
    &\reason{by \refNum{app:rec-call-one-composition}} \\
=\;\; &(\theta''_2 \circ \theta_2) \circ \theta_1
    &\reason{by \refNum{app:rec-call-two-composition}} \\
=\;\; &((\theta''_3 \circ \theta_3) \circ \theta_2) \circ \theta_1
    &\reason{by \refNum{app:theta-U-wf-composition-simpl}} \\
=\;\; &\theta'' \circ \theta'
\end{alignat*}


We show~(\ref{proofobl:compl:subsumption}):
\begin{alignat*}{3}
  &\theta''(A_R) \\
=\;\; &\theta''_3(B) \qquad\quad
     &&\reason{by $\theta'' = \theta''_3$, $A_R = B$} \\
=\;\; &\theta''_3\theta_3'(b)
     &&\reason{by \refNum{app:def-theta-three-from-theta-three-prime}}\\
=\;\; &\theta_b(b)
     &&\reason{by \refNum{app:theta-U-wf-composition}}\\
=\;\; &A_0 &&\reason{by \refNum{app:theta-b-on-b}}
\end{alignat*}

\item[Case $\Let \; x = \; M \; \In \; N$:]
By \refNum{typed} and \freezemlLab{Let}, there exist $A'$, $A_x$, and $\Delta_G$ such that
\begin{align}
 &\Delta_G = \ftv(A') - (\Delta, \Theta') \mathlabelNum{def-Delta-G} \\
&\typ{\Delta, \Theta', \Delta_G;\theta_0\Gamma}{M : A'}%
  \mathlabelNum{let:type-M} \\
&((\Delta, \Theta'), \Delta_G,M,A') \Updownarrow A_x~\mathlabelNum{let:updownarrow} \\
&\typ{\Delta, \Theta' ;\theta_0\Gamma, x : A_x}{N : A_0}
  \mathlabelNum{let:type-N} \\
&\meta{principal}((\Delta, \Theta'), \theta_0\Gamma,\Delta_G,A')%
  \mathlabelNum{letA-prime-is-principal}
\end{align}

We assume without loss of generality that $\Delta_G$ is fresh,
in particular $\Delta_G \fresh \Theta$.
This is justified, as we may otherwise apply \cref{lem:bijective-renaming} to
\refNum{let:type-M} using a substitution that does the necessary freshening.
This would yield corresponding judgements for deriving $\typ{\Delta, \Theta';\theta_0 \Gamma}{\Let \; x = \; M \; \In \; N : A_0}$.

By~\refNum{theta-well-formed} and weakening,
we have $\Delta \vdash \theta_0 : \Theta \Rightarrow
\Theta', \Delta_G$.
Together with \refNum{let:type-M} we then have that $\Infer(\Delta,
\Theta, \Gamma, M)$ succeeds, returning $(\Theta_1, \theta_1, A)$, and there
exists $\theta''_1$ such that
\begin{align}
&\Delta \vdash \theta''_1 : \Theta_1 \Rightarrow (\Theta', \Delta_G)
  \mathlabelNum{let:theta-two-primes-one-wf} \\
&\theta_0 = \theta''_1 \circ \theta_1
  \mathlabelNum{let:tilde-theta-composition} \\
&\theta''_1(A) = A'
  \mathlabelNum{let:A-prime-is-instance-of-A}
\end{align}

By \refNum{term-wf,gamma-wf}, \cref{thm:inference-sound} yields
  $\Delta \vdash \theta_1 : \Theta \Rightarrow \Theta_1$~\labelNum{let:theta-one-wf}
and
  $\typ{\Delta, \Theta_1; \theta_1\Gamma}{M : A}$,
which implies $\Delta, \Theta_1 \vdash A$~\labelNum{let:A-wf}.

Note that $\Delta_G$ does not appear as part of the input to $\meta{infer}$,
and we therefore have $\Delta_G \fresh \Theta_1$.

Let $\Theta_{\theta_1} = \ftv(\theta_1) - \Delta$,
which implies $\Theta_{\theta_1} \subseteq \Theta_1$
and $\Delta''' \mathbin{\#} \Theta_{\theta_1}$ and $\Delta'' \mathbin{\#} \Theta_{\theta_1}$.
By \refNum{let:tilde-theta-composition,let:theta-two-primes-one-wf,theta-well-formed}
we have $\Delta, \Theta' \vdash \theta''_1(a) :K $ for all $(a : K) \in \Theta_{\theta_1}$~\labelNum{let:theta-two-primes-one-restriction}.


By \refNum{%
  let:tilde-theta-composition,%
  theta-well-formed,%
  let:theta-two-primes-one-wf,%
  def-Delta-G,%
  let:A-prime-is-instance-of-A,%
  let:A-wf,%
  letA-prime-is-principal%
  },
we can apply \cref{lem:bijection-on-principal-types-vars}, yielding $\theta_1''(\Delta''') = \Delta_G$~\labelNum{let:Delta-three-primes-vs-Delta-G}.

We have $\Delta'' \mathop{\#} \Theta_{\theta_1}$ and can therefore strengthen \refNum{let:theta-one-wf} to
$\Delta \vdash \theta_1 : \Theta \Rightarrow \Theta_1 - \Delta''$~\labelNum{let:theta-one-wf-strengthened}.

We distinguish two cases based on the shape of $M$.
In each case we show that
 there exists $\theta''_N$ such that
\begin{align}
 &\Delta \vdash \theta''_N :  (\Theta_1' - \Delta'') \Rightarrow \Theta'
   \mathlabelNum{let:theta-two-primes-N-wf} \\
 &\typ{\Delta, \Theta'; \theta''_N(\theta_1(\Gamma), x : \forall \Delta''.A)}{N : A_0}
    \mathlabelNum{let:type-N-substed} \\
 &\theta_0 = \theta''_N \circ \theta_1
  \mathlabelNum{let:tilde-theta-composition-theta-two-primes-N}
\end{align}

\begin{description}
\item[Subcase 1, $M \in \dec{GVal}$:]

We have $\Delta'' = \Delta'''$.
By \refNum{let:updownarrow}, we have that $A_x = \forall \Delta_G.A'$ holds.

According to $\Delta'' = \Delta'''$ and $\Theta'_1 = \meta{demote}(\mono, \Theta_1, \Delta''')$ we have that
  $\Theta'_1 - \Delta'' = \Theta_1 - \Delta''$.



Let $\theta''_N$ be defined as follows for all $c \in \Theta_1 - \Delta'' = \Theta'_1 - \Delta''$:
\[
\theta''_N(c) =
\begin{cases}
  \theta''_1(c) &\text{if $c \in \Theta_{\theta_1}$} \\
  A_D         &\text{if $c \in \Theta_1 - \Delta'' - \Theta_{\theta_1}$}
\end{cases}
\]
Where $A_D$ is some arbitrary type with $\Delta, \Theta' \vdash A_D : \mono$ (e.g., $\Int$).
By $\Theta_{\theta_1} \subseteq \Theta_1$, this definition is well-formed.

By $\Delta'' = \Delta''' = \ftv(A) - \Delta - \Theta_{\theta_1}$ we have
$\theta''_N(c) = \theta''(c) $ for all $c \in \ftv(A) - \Delta''$~\labelNum{let:theta-two-primes-N-on-ftv-A}.

Together with \refNum{let:theta-two-primes-one-restriction} and $\Delta, \Theta' \vdash A_D : \mono$, we then have
$\Delta, \Theta' \vdash \theta''_N(c) : K $ for all $(c : K) \in \Theta_1 - \Delta''$
and therefore $\Delta \vdash \theta''_N : \Theta'_1 - \Delta'' \Rightarrow \Theta'$.

By \refNum{let:theta-one-wf-strengthened,let:tilde-theta-composition} and $\theta''_N(c) = \theta''_1(c)$ for all $c \in \Theta_{\theta_1}$ we
also have $\theta_0 = \theta''_N \circ \theta_1$.

We have
\[
\begin{array}{cll}
= &\theta''_N(\forall \Delta''.A) \\
= &\theta''_N(\forall \Delta_G. A [\Delta_G/\Delta'']) \\
= &\forall \Delta_G. \theta''_N(A [\Delta_G/\Delta''])
    &\reason{by $\ftv(\theta''_N) \subseteq \Delta, \Theta'$ and $\Delta, \Theta' \mathop{\#} \Delta_G \fresh \Theta_1$} \\
= &\forall \Delta_G. \theta''_1(A) &\reason{by $\Delta'' = \Delta'''$ and \refNum{let:Delta-three-primes-vs-Delta-G} and \refNum{let:theta-two-primes-N-on-ftv-A} } \\
=  &A_x &\reason{by $A_x = \forall \Delta_G.A'$ and \refNum{let:A-prime-is-instance-of-A}}\\
\end{array}
\]

Thus, \refNum{let:type-N} is equivalent to
$\typ{\Delta, \Theta'; \theta''_N((\theta_1\Gamma), x : \forall \Delta''.A)}{N : A_0}$.

\item[Subcase 2, $M \not\in \dec{GVal}$:]
We have $\Delta'' = \cdot$.
By \refNum{let:updownarrow}, we have $A_x = \rsubst(A')$ for some $\rsubst$ with
$\Delta, \Theta' \vdash \rsubst : \Delta_G \Rightarrow_\mono \cdot$~\labelNum{let:rsubst-wf}.

Let $\theta''_N$ be defined as follows for all $c \in \Theta_1 - \Delta'' = \Theta_1$:
\begin{equation}
\theta''_N(c) =
\begin{cases}
  \theta''_1(c) &\text{if $c \in \Theta_{\theta_1}$} \\
  A_D &\text{if $c \in \Theta_1 - \Delta''' - \Theta_{\theta_1}$} \\
  \rsubst(\theta''_1(c)) &\text{if $c \in \Delta'''$} \\
\end{cases}
  \mathlabelNum{let:def-theta-two-primes-N-non-value-case}
\end{equation}
Here, $A_D$ is defined as before.

By $\Delta'' \mathbin{\#} \Theta_{\theta_1}$ and $\Delta''' \subseteq \Theta_1$ and $\Theta_{\theta_1} \subseteq \Theta_1$,
the three cases are non-overlapping and exhaustive for $\Theta_1$.

Using \refNum{let:theta-two-primes-one-restriction}, we have that $\Delta, \Theta' \vdash \theta''_N(c) : K $ for all $(c : K) \in \Theta_{\theta_1}$.
Note that by $\Delta''' \fresh \Theta_{\theta_1}$ we have $\Theta_1(c) = \Theta'_1(c)$ for all $c \in \Theta_{\theta_1}$.

By \refNum{let:rsubst-wf}, we have $\Delta, \Theta' \vdash \rsubst(c) : \mono $ for all $c \in \Delta_G$ and therefore $\Delta, \Theta' \vdash \theta''_N(c')~:~\mono$ for all $(c' : K) \in \Delta'''$.

Together with $\Delta, \Theta' \vdash A_D : \mono$, we then have $\Delta \vdash \theta''_N : \Theta'_1 \Rightarrow  \Theta'$.
By \cref{lem:stability-subst-under-promotion}, we also have $\Delta \vdash \theta''_N : \Theta_1 \Rightarrow  \Theta'$.
We have $\theta''_N(c) = \theta''(c)$ for all $c \in \Theta_{\theta_1}$ and together with \refNum{let:theta-one-wf}, \refNum{let:tilde-theta-composition}, and $\Delta'' = \cdot$ we
then have $\theta_0 = \theta''_N \circ \theta_1$.


We have
\[
\begin{array}{cll}
  &\theta''_N(\forall \Delta''.A) \\
= &\theta''_N(A) &\reason{by $\Delta'' = \cdot$}\\
= &\theta''_1(A)[\rsubst(\Delta_G)/\Delta_G]
 &\reason{by \refNum{let:Delta-three-primes-vs-Delta-G,%
                     let:def-theta-two-primes-N-non-value-case}} \\
= &A'[\rsubst(\Delta_G)/\Delta_G] &\reason{by \refNum{let:A-prime-is-instance-of-A}} \\
= &\rsubst(A') \\
= &A_x
\end{array}
\]

\end{description}

We have shown that in each case,
 \refNum{let:theta-two-primes-N-wf},
 \refNum{let:type-N-substed}, and
 \refNum{let:tilde-theta-composition-theta-two-primes-N}
hold.
Using the same reasoning as in the case for unannotated $\Let$ in the proof of
\cref{thm:inference-sound}, we obtain $\wfctx{\Delta, \Theta_1 - \Delta''}{\theta_1 : \Gamma}$.

Thus, by induction,
we have that $\Infer(\Delta, \Theta'_1 - \Delta'', \theta_1\Gamma, N)$ succeeds,
returning $(\Theta_2, \theta_2, B)$,
and there exists $\theta''_2$ such that
\begin{align}
&\Delta \vdash \theta''_2 : \Theta_2 \Rightarrow \Theta'
  \mathlabelNum{let:theta-two-primes-two-wf} \\
&\theta''_N = \theta''_2 \circ \theta_2 \mathlabelNum{let:theta''-N-composition} \\
&\theta''_2(B) = A_0
  \mathlabelNum{let:B-vs-tilde-A}
\end{align}

By the return values of $\Infer$, we have $\Theta' := \Theta_2$, and $\theta' := \theta_2 \circ \theta_1$
and $A_R := B$.

Let $\theta'' = \theta''_2$.
By \refNum{let:theta-two-primes-two-wf}, this choice immediately satisfies~(\ref{proofobl:compl:theta-two-primes-wf}).

We have
\[
\begin{array}{cll}
  &\theta_0 \\
= &\theta''_N \circ \theta_1 &\reason{by \refNum{let:tilde-theta-composition-theta-two-primes-N}} \\
= &\theta''_2 \circ \theta_2 \circ \theta_1 &\reason{by \refNum{let:theta''-N-composition}}
\end{array}
\]
and therefore $\theta_0 = \theta'' \circ \theta'$~(\ref{proofobl:compl:theta-composition}).

We show satisfaction of~(\ref{proofobl:compl:subsumption}) as follows:
\[
\begin{array}{cll}
  &A_0 \\
= & \theta''_2(B) &\reason{by \refNum{let:B-vs-tilde-A}} \\
= &\theta''(B) &\reason{by $\theta'' := \theta''_2$}
\end{array}
\]

\item[Case $\Let \: (x : A) = \: M \: \In \: N$:]

By \refNum{typed} and \freezemlLab{Let-Ascribe}, there exist $\Delta_G$ and $A_M$ such that we have
\begin{align}
&\Delta_G, A_M = \msplit(A, M) \\
&\typ{\Delta, \Theta', \Delta_G;\theta_0\Gamma}{M :
  A_M} \mathlabelNum{let-asc:type-M} \\
&A = \forall \Delta_G.A_M \mathlabelNum{let-asc:A-vs-A-prime} \\
& \typ{\Delta, \Theta'; \theta_0(\Gamma), (x : A)}{N : A_0} \mathlabelNum{let-asc:type-N}
\end{align}
By alpha-equivalence, we assume $\Delta_G \fresh \Theta$.

Note that by definition of $\Infer$ and $\meta{split}$, we have $\Delta_G = \Delta'$
and $A' = A_M$~\labelNum{let-asc:A-prime-vs-A-two}.
By \refNum{let-asc:type-M}, we have $\Delta' \mathop{\#}
\Theta'$~\labelNum{let-asc:Delta-prime-Theta-prime-disjoint}.
We weaken \refNum{theta-well-formed} to
$\Delta, \Delta' \vdash \theta_0 :
\Theta \Rightarrow \Theta'$.

By inversion on \refNum{term-wf}, we have $\termwf{\Delta, \Delta'}{M}$ and $\termwf{\Delta}{N}$ and
and $\Delta \vdash A$~\labelNum{let-asc:plain-A-wf}, which implies $\Delta, \Delta' \vdash A'$~\labelNum{let-asc:A-wf}.

Together with~\refNum{let-asc:type-M} we then have the following by
induction:
$\Infer((\Delta, \Delta'), \Theta, \Gamma, M)$ succeeds, returning $(\Theta_1,
\theta_1, A_1)$ and there exists $\theta''_1$ such that
\begin{align}
&\Delta, \Delta' \vdash \theta''_1 : \Theta_1 \Rightarrow
  \Theta'\mathlabelNum{let-asc:theta-one-two-primes-wf} \\
&\theta_0 = \theta''_1 \circ
  \theta_1 \mathlabelNum{let-asc:theta-tilde-composition} \\
&\theta''_1 (A_1) = A_M \mathlabelNum{let-asc:A-one-vs-A-prime}
\end{align}



\Cref{thm:inference-sound} yields
$\Delta, \Delta' \vdash \theta_1 : \Theta \Rightarrow \Theta_1$~\labelNum{let-asc:theta-one-wf} and
$\typ{\Delta, \Delta', \Theta_1; \theta_1(\Gamma)}{M : A_1}$,
which implies $\Delta, \Delta', \Theta_1 \vdash A_1$~\labelNum{let-asc:A-one-wf}.

We then have
\[
\begin{array}{cll}
   &\theta''_1 (A_1) \\
=  &A_M
      &\reason{by \refNum{let-asc:A-one-vs-A-prime}}\\
=  &A'
      &\reason{by \refNum{let-asc:A-prime-vs-A-two}} \\
=  &\theta''_1 (A')
     &\reason{by \refNum{let-asc:theta-one-two-primes-wf,let-asc:A-wf}} \\
\end{array}
\]
In addition to above equality and \refNum{let-asc:theta-one-two-primes-wf}
as well as \refNum{let-asc:A-one-wf}, we have
$\Delta, \Delta', \Theta_1 \vdash A'$ by weakening \refNum {let-asc:A-wf}.
Hence, \cref{thm:thmunifycomplete} yields the following:
$\unify((\Delta,\Delta'), \Theta_1, A', A_1)$ succeeds, returning $(\Theta_2,
\theta'_2)$, and there exists $\theta''_2$ such that
\begin{align}
&\Delta, \Delta' \vdash \theta''_2 : \Theta_2 \Rightarrow
  \Theta' \mathlabelNum{let-asc:theta-two-primes-two-wf} \\
&\theta''_1 = \theta''_2 \circ
  \theta_2' \mathlabelNum{let-asc:theta-two-primes-one-comp}
\end{align}
%
By \cref{thm:unification-sound}, we have
$\Delta, \Delta' \vdash \theta'_2 : \Theta_1 \Rightarrow
\Theta_2$.
Together with \refNum{let-asc:theta-one-wf} and composition,
we then have $\Delta, \Delta' \vdash \theta_2 : \Theta \Rightarrow \Theta_2$%
  ~\labelNum{let-asc:theta-two-wf}.


By~\refNum{let-asc:theta-tilde-composition}
and~\refNum{let-asc:theta-two-primes-one-comp}, we have
$\theta_0 = \theta''_2 \circ \theta'_2 \circ
\theta_1 = \theta''_2 \circ \theta_2$~\labelNum{let-asc:tilde-theta-comnposition-two}.
We show $\ftv(\theta_2) \subseteq \Delta, \Theta_2$:
Otherwise, if $a \in \Theta$ and $b \in \Delta'$ such that $b \in \ftv(\theta_2(a))$, then
by \refNum{let-asc:theta-two-primes-two-wf}, $\theta''_2(b) = b$ and
$b \in \ftv(\theta''_2(\theta_2(a))) = \ftv(\theta_0(a))$, violating
\refNum{theta-well-formed}.

Therefore, the assertion $\ftv(\theta_2) \fresh \Delta'$ succeeds,
allowing us to strengthen \refNum{let-asc:theta-two-wf} to
$\Delta \vdash \theta_2 : \Theta \Rightarrow
\Theta_2$~\labelNum{let-asc:theta-two-wf-stronger}.

By \refNum{let-asc:plain-A-wf} we have $\ftv(A) \subseteq \Delta$,
and together with \refNum{let-asc:theta-two-primes-two-wf} this yields $\theta''_2(A) =
A$~\labelNum{let-asc:theta-two-primes-two-applied-to-A}.


We have
\begin{alignat}{3}
 &&\typ{\Delta, \Theta' ;\theta_0\Gamma, x : A&}{N : A_0}
      &&\qquad\reason{by \refNum{let-asc:type-N}} \notag \\
\text{implies} \qquad
  &&\typ{\Delta, \Theta' ;\theta''_2\theta_2\Gamma, x : A&}{N : A_0}
      &&\qquad\reason{by \refNum{let-asc:tilde-theta-comnposition-two,%
                  let-asc:theta-two-wf,let-asc:theta-two-primes-two-wf}} \notag\\
\text{implies} \qquad
 &&\typ{\Delta, \Theta' ;\theta''_2(\theta_2(\Gamma), x : A)&}{N : A_0}
     &&\qquad\reason{by \refNum{let-asc:theta-two-primes-two-applied-to-A}}
     \mathlabelNum{let-asc:typ-N-substed}
\end{alignat}

By \refNum{gamma-wf} and \refNum{let-asc:theta-two-wf-stronger},
we have $\wfctx{\Delta, \Theta_2}{\theta_2 (\Gamma)}$.
Together with \refNum{let-asc:plain-A-wf}, we then have
$\wfctx{\Delta, \Theta_2}{\theta_2 (\Gamma), x : A}$.

Hence, induction on \refNum{let-asc:typ-N-substed,%
let-asc:theta-two-wf-stronger} shows that $\Infer(\Delta,
\Theta_2 , (\theta_2 \Gamma, x : A), N)$ succeeds, returning
$(\Theta_3,\theta_3,B)$ and there exists $\theta''_3$ such that
\begin{align}
&\Delta \vdash \theta''_3 : \Theta_3 \Rightarrow
  \Theta' \mathlabelNum{let-asc:theta-two-primes-three-wf} \\
&\theta''_2 = \theta''_3 \circ
  \theta_3 \mathlabelNum{let-asc:theta-two-primes-two-vs-composition} \\
&\theta''_3(B) =
  A_0 \mathlabelNum{let-asc:theta-two-primes-three-on-B}
\end{align}


We have shown that all steps of the algorithm
succeed.
According to the return values of $\Infer$, we have $\Theta'' = \Theta_3$,
$\theta' = \theta_3 \circ \theta_2$, and $A_R = B$.
Let $\theta'' = \theta''_3$.
By \refNum{let-asc:theta-two-primes-three-wf}, this choice immediately
satisfies~(\ref{proofobl:compl:theta-two-primes-wf}).

We show (\ref{proofobl:compl:theta-composition}):
\[
\begin{array}{cll}
  &\theta_0 \\
= &\theta''_2 \circ \theta_2
    &\reason{by \refNum{let-asc:tilde-theta-comnposition-two}} \\
= &\theta''_3 \circ \theta_3 \circ \theta_2
    &\reason{by \refNum{let-asc:theta-two-primes-two-vs-composition}} \\
= &\theta'' \circ \theta'
    &\reason{by $\theta' := \theta_3 \circ \theta_2$, $\theta'' := \theta''_3$ }
\end{array}
\]

By $\theta'' = \theta''_3$,
\refNum{let-asc:theta-two-primes-three-on-B} yields
(\ref{proofobl:compl:subsumption}).




\end{description}

\end{proof}

\fi

\end{document}

%% file: figures/example_table.tex
\begin{figure*}[htb]
\small
\setlength{\tabcolsep}{3pt}

\begin{tikzpicture}
\node (A-B-C) {

\begin{tabular}{l r c l}
\hline
A  & \multicolumn{3}{l}{POLYMORPHIC INSTANTIATION} \\
\hline
A1          & $\lambda x\ y.y$                                  &:& $a \to b \to b$\\
A1\altty    & $\gen(\lambda x\ y.y)$                            &:& $\forall a\ b.a \to b \to b$\\
A2          & $\dec{choose}\ \dec{id}$                          &:& $(a \to a) \to (a \to a)$ \\
A2\altty    & $\dec{choose}\ \freeze{\dec{id}}$                 &:& $(\forall a.a \to a) \to (\forall a.a \to a)$ \\
A3          & $\dec{choose}\ [] \ \dec{ids}$                    &:& $\lstp{\forall a.a \to a}$ \\
A4          & $\lambda(x : \forall a.a \to a). x\ x$            &:& $(\forall a.a \to a) \to (b \to b)$\\
A4\altty    & $\lambda(x : \forall a.a \to a). x\ \freeze{x}$   &:& $(\forall a.a \to a) \to (\forall a.a \to a)$\\
A5          & $\dec{id}\ \dec{auto}$                            &:& $(\forall a.a \to a) \to (\forall a.a \to a)$\\
A6          & $\dec{id}\ \dec{auto'}$                           &:& $(\forall a.a \to a) \to (b \to b)$ \\
A6\altty    & $\dec{id}\ \freeze{\dec{auto'}}$                  &:& $\forall b.(\forall a.a \to a) \to (b \to b)$ \\
A7          & $\dec{choose}\ \dec{id}\ \dec{auto}$              &:& $(\forall a.a \to a) \to (\forall a.a \to a)$\\
A8          & $\dec{choose}\ \dec{id}\ \dec{auto'}$             &:& \xmark \\
A9\newty    & $f\ (\dec{choose}\ \freeze{\dec{id}})\ \dec{ids}$ &:& $\forall a.a \to a$ \\
            & \multicolumn{3}{r}{where $f : \forall a.(a \to a) \to \lst{a} \to a$} \\
A10\newty   & $\dec{poly}\ \freeze{\dec{id}}$                   &:& $\Int \times \Bool$ \\
A11\newty   & $\dec{poly}\ \gen(\lambda x.x)$                   &:& $\Int \times \Bool$ \\
A12\newty   & $\dec{id}\ \dec{poly}\ \gen(\lambda x.x)$         &:& $\Int \times \Bool$\\
\hline
C  & \multicolumn{3}{l}{FUNCTIONS ON POLYMORPHIC LISTS} \\
\hline
C1        & $\dec{length}\ \dec{ids}$            &:& $\Int$ \\
C2        & $\dec{tail}\ \dec{ids}$              &:& $\lstp{\forall a.a \to a}$\\
C3        & $\dec{head}\ \dec{ids}$              &:& $\forall a.a \to a$\\
C4        & $\dec{single}\ \dec{id}$             &:& $\lstp{a \to a}$\\
C4\altty  & $\dec{single}\ \freeze{\dec{id}}$    &:& $\lstp{\forall a.a \to a}$\\
C5\newty  & $\freeze{\dec{id}} \cons \dec{ids}$  &:& $\lstp{\forall a.a \to a}$\\
C6\newty  & $\gen(\lambda x.x) \cons \dec{ids}$&:& $\lstp{\forall a.a \to a}$\\
C7        & $(\dec{single}\ \dec{inc}) \append (\dec{single}\ \dec{id})$ &:& $\lstp{\Int \to \Int}$\\
C8\newty  & $\dec{g}\ (\dec{single}\ \freeze{\dec{id}})\ \dec{ids}$               &:& $\forall a.a \to a$ \\
          & \multicolumn{3}{r}{where $\dec{g} : \forall a.\lst{a} \to \lst{a} \to a$} \\
C9\newty  & $\dec{map}\ \dec{poly}\ (\dec{single}\ \freeze{\dec{id}})$            &:& $\lstp{\Int \times \Bool}$\\
C10       & $\dec{map}\ \dec{head}\ (\dec{single}\ \dec{ids})$                    &:& $\lstp{\forall a.a \to a}$\\
\hline
\end{tabular}
};
\node[right= 0.1cm of A-B-C.north east, anchor=north west] (D-E-F){
\begin{tabular}{l r cl}
\hline
B  & \multicolumn{3}{l}{INFERENCE WITH POLYMORPHIC ARGUMENTS} \\
\hline
B1\newty    & $\lambda (f : \forall a.a \to a).$ \\
           & $\qquad(f\ 1, f\ \dec{True})$  &:& $(\forall a.a \to a) \to \Int \times \Bool$ \\
B2\newty    & $\lambda (\var{xs} : \lstp{\forall a.a \to a}).$ \\
             & $\qquad\qquad \dec{poly}\ (\dec{head}\ \var{xs})$ &:& $\lstp{\forall a.a \to a} \to \Int \times \Bool$ \\
\hline
D  & \multicolumn{3}{l}{APPLICATION FUNCTIONS \phantom{jhgasjdhagskdjhagskdjhg}} \\
\hline
D1\newty & $\dec{app}\ \dec{poly}\ \freeze{\dec{id}}$        &:& $\Int \times \Bool$\\
D2\newty & $\dec{revapp}\ \freeze{\dec{id}}\ \dec{poly}$     &:& $\Int \times \Bool$\\
D3\newty & $\dec{runST}\ \freeze{\dec{argST}}$               &:& $\Int$ \\
D4\newty & $\dec{app}\ \dec{runST}\ \freeze{\dec{argST}}$    &:& $\Int$ \\
D5\newty & $\dec{revapp}\ \freeze{\dec{argST}}\ \dec{runST}$ &:& $\Int$ \\
\hline
E  & \multicolumn{3}{l}{$\eta$-EXPANSION} \\
\hline
E1        & $k\ h\ l$                                   &:& \xmark \\
E2\newty  & $k\ \gen(\lambda x.(h\ x)\inst)\ l$ &:& $\forall a.\Int \to a \to a$ \\
          & \multicolumn{3}{r}{where
                $\ba[t]{@{~}l@{~}c@{~}l@{}}
                   k &:& \forall a.a \to \lst{a} \to a \\
                   h &:& \Int \to \forall a.a \to a \\
                   l &:& \lstp{\forall a.\Int \to a \to a} \\
                \ea$} \\
E3        & $r\ (\lambda x\ y.y)$                             &:& \xmark \\
E3\altty  & $r\ \gen(\lambda x.\gen(\lambda y.y))$          &:& $\Int$ \\
          & \multicolumn{3}{r}{where $r : (\forall a.a \to \forall b.b \to b) \to \Int$} \\
\hline
F  & \multicolumn{3}{l}{\freezeml PROGRAMS} \\
\hline
F1        & $\dec{id}  = \gen(\lambda \var{x}. \var{x})$                     &:& $\forall a.a \to a$\\
F2        & $\dec{ids} = [\freeze{\dec{id}}]$                 &:& $\lstp{\forall a.a \to a}$\\
F3        & $\dec{auto} = \lambda (\var{x} : \forall a.a \to a). \var{x}\ \freeze{\var{x}}$ &:& $(\forall a.a \to a) \to (\forall a.a \to a)$\\
F4        & $\dec{auto'} = \lambda (\var{x} : \forall a.a \to a). \var{x}\ \var{x}$         &:& $\forall b.(\forall a.a \to a) \to b \to b$\\
F5\newty  & $\dec{auto}\ \freeze{\dec{id}}$                   &:& $\forall a.a \to a$ \\
F6        & $(\dec{head}\ \dec{ids}) \cons \dec{ids}$               &:& $\lstp{\forall a.a \to a}$ \\
F7\newty  & $(\dec{head}\ \dec{ids})\inst\ 3$                 &:& $\Int$ \\
F8        & $\dec{choose}\ (\dec{head}\ \dec{ids})$           &:& $(\forall a.a \to a) \to (\forall a.a \to a)$\\
F8\altty  & $\dec{choose}\ (\dec{head}\ \dec{ids})\inst$      &:& $(a \to a) \to (a \to a)$\\
 \multicolumn{4}{l}{
\hspace{-5pt}$
\begin{array}{ll}
\text{F9}
& \Let\; f = \dec{revapp}\ \freeze{\dec{id}}\;\In\;  f\ \dec{poly} \\ &\qquad: \Int \times \Bool \\
 \text{F10}\dagger       & \dec{choose}\ \dec{id}\
                (\lambda
                  (x : \forall a.a \to a).\gen(\dec{auto'}\ x)) \\
   & \qquad : (\forall a.a \to a) \to (\forall a.a \to a)\\
\end{array}
$
} \\
\hline
\end{tabular}
};
\end{tikzpicture}
\caption{Example \freezeml Terms and Types}
\label{fig:freezeml-examples}
\end{figure*}


%% file: figures/example_signatures.tex
\begin{figure*}[ht]
\begin{minipage}{\textwidth}
\begin{align*}
\dec{head}   &: \forall a.\lst{a} \to a                   & \dec{id}     &: \forall a.a \to a                             & \dec{map}    &: \forall a\ b.(a \to b) \to \lst{a} \to \lst{b}  \\
\dec{tail}   &: \forall a.\lst{a} \to \lst{a}             & \dec{ids}    &: [\forall a.a \to a]                           & \dec{app}    &: \forall a\ b.(a \to b) \to a \to b             \\
[\,]         &: \forall a.\lst{a}                         & \dec{inc}    &: \Int \to \Int                                 & \dec{revapp} &: \forall a\ b.a \to (a \to b) \to b              \\
(\cons)      &: \forall a.a \to \lst{a} \to \lst{a}                         & \dec{choose} &: \forall a.a \to a \to a     & \dec{runST}  &: \forall a.(\forall s.ST\ s\ a) \to a           \\
\dec{single} &: \forall a.a \to \lst{a}                   & \dec{poly}   &: (\forall a.a \to a) \to \Int \times \Bool     & \dec{argST}  &: \forall s.ST\ s\ \Int                       \\
(\append)    &: \forall a.\lst{a} \to \lst{a} \to \lst{a} & \dec{auto}   &: (\forall a.a \to a) \to (\forall a.a \to a)   & \dec{pair}   &: \forall a\ b.a \to b \to a \times b \\
\dec{length} &: \forall a.\lst{a} \to \Int                & \dec{auto'}  &: \forall b.(\forall a.a \to a) \to (b \to b)   & \dec{pair'}  &: \forall b\ a.a \to b \to a \times b
\end{align*}
\end{minipage}
\caption{Type signatures for functions used in the text; adapted from~\cite{SerranoHVJ18}.}
\label{fig:function-signatures}
\end{figure*}


%% file: figures/freezeml_syntax.tex
\begin{figure}
\[
\bl
\ba[t]{@{}l@{}r@{~}c@{~}l@{}}
\slab{Type Variables}     &\dec{TVar} \ni a, b, c \\
\slab{Type Constructors}  &\dec{Con} \ni \tc
                                &::= & \Int \mid \List \mid \mathord{\to} \mid \mathord{\times} \mid \dots \\
\slab{Types}              &\dec{Type} \ni A, B
                                &::= & a \mid \tc\,\many{A} \mid \forall a.A \\
\slab{Monotypes}          &\dec{MType} \ni S, T
                                &::= & a \mid \tc\,\many{S} \\
\slab{Guarded Types}      &\dec{GType} \ni H
                                &::= & a \mid \tc\,\many{A} \\
\slab{Type Instantiation} &\dec{Subst} \ni \rsubst  &::= & \emptyset \mid \rsubst[a \mapsto A] \\
\slab{Term Variables}     &\dec{Var} \ni x, y, z \\
\slab{Terms}              &\dec{Term} \ni M, N
                                &::= & x \mid \freeze{x} \mid \lambda x.M  \\
                           &     &\mid& \lambda (x : A).M \mid M\,N \\
                           &     &\mid& \Let\; x = M \;\In\; N \\
                           &     &\mid& \Let\; (x : A) = M \;\In\; N \\
\slab{Values}             &\dec{Val} \ni V, W
                                &::= & x \mid \freeze{x} \mid \lambda x.M \\
                           &     &\mid& \lambda (x : A). M \\
                           &     &\mid& \Let\; x = V \;\In\; W \\
                           &     &\mid& \Let\; (x : A) = V \;\In\; W \\
\slab{Guarded Values}     &\dec{GVal} \ni U
                                &::= & x \mid \lambda x.M \mid \lambda (x\!:\! A).M\\
                                &     &\mid& \Let\; x = V \;\In\; U \\
                                &     &\mid& \Let\; (x : A) = V \;\In\; U \\
\slab{Kinds}              &\dec{Kind} \ni K
                                &::= & \mono \mid \poly \\
\slab{Kind Environments}  &\dec{PEnv} \ni \Delta
                                &::= & \cdot \mid \Delta, a \\
\slab{Type Environments}  &\dec{TEnv} \ni \Gamma
                                &::= & \cdot \mid \Gamma, x:A \\
\ea \\
\el
\]

\caption{\freezeml Syntax}
\label{fig:freezeml-syntax}
\end{figure}


%% file: figures/freezeml_kinding_instantiation.tex
\begin{figure}
\raggedright
$\boxed{\Delta \vdash A : K}$
\vspace{-0.6cm}
\begin{mathpar}
  \inferrule
    {a \in \Delta}
    {\Delta \vdash a : \mono}

  \inferrule
    {\arity(D) = n \\\\
     \Delta \vdash A_1 : K \\\\
     \cdots \\\\
     \Delta \vdash A_n : K}
    {\Delta \vdash \tc\,\many{A} : K}

  \inferrule
    {\Delta, a \vdash A : \poly}
    {\Delta \vdash \forall a.A : \poly}

  \inferrule
    {\Delta \vdash A : \mono}
    {\Delta \vdash A : \poly}
\end{mathpar}
\caption{\freezeml Kinding Rules}
\label{fig:freezeml-kinding}
\end{figure}
\begin{figure}
\raggedright
$\boxed{\Delta \vdash \rsubst : \Delta' \Rightarrow_K \Delta''}$
\vspace{-0.4cm}
\begin{mathpar}
\inferrule
  { }
  {\Delta \vdash \emptyset : \cdot \Rightarrow_K \Delta'}

\inferrule
  {\Delta \vdash \rsubst : \Delta' \Rightarrow_K \Delta'' \\
   \Delta, \Delta'' \vdash A : K}
  {\Delta \vdash \rsubst[a \mapsto A] : (\Delta', a) \Rightarrow_K \Delta''}
\end{mathpar}

\caption{\freezeml Instantiation Rules}
\label{fig:freezeml-inst}
\end{figure}

%% file: figures/type_instantiation.tex
\begin{figure}[tb]
\begin{mathpar}
\ba{@{}r@{~}c@{~}l@{~}c@{~}r@{~}c@{~}l@{}}
\emptyset(A)            &=& A                     &\hspace{20pt}& \rsubst[a \mapsto A](a) &=& A\\
\rsubst(D~\many{A})     &=& D~(\many{\rsubst(A)}) &\hspace{20pt}& \rsubst[a \mapsto A](b) &=& \rsubst(b)
\ea
\vspace{-11pt}

\begin{eqs}
\rsubst(\forall a.A) &=& \forall c.\rsubst[a \mapsto c](A), \text{where}\ c \not\in \ftv(\rsubst(b)) \text{ for all } b \neq  c\\
\end{eqs}
\end{mathpar}
\caption{Application of a Type Instantiation in \freezeml}
\label{fig:freezeml-substitution-application}
\end{figure}


%% file: figures/freezeml_typing_instantiation.tex
\begin{figure}
\raggedright
$\boxed{\typ{\Delta; \Gamma}{M : A}}$
\vspace{-0.4cm}
\begin{mathpar}
  \inferrule*[Lab=\freezemlLab{Freeze}]
    {x : A \in \Gamma}
    {\typ{\Delta; \Gamma}{\freeze{x} : A}}

  \inferrule*[Lab=\freezemlLab{Var}]
    {{x : \forall \Delta'.H \in \Gamma} \\\\
      \Delta \vdash \rsubst : \Delta' \Rightarrow_\poly \cdot
    }
    {\typ{\Delta; \Gamma}{x : \rsubst(H)}}

  \inferrule*[Lab=\freezemlLab{App}]
    {\typ{\Delta; \Gamma}{M : A \to B} \\\\
     \typ{\Delta; \Gamma}{N : A}
    }
    {\typ{\Delta; \Gamma}{M\,N : B}}

  \inferrule*[Lab=\freezemlLab{Lam}]
    {\typ{\Delta; \Gamma, x : S}{M : B}}
    {\typ{\Delta; \Gamma}{\lambda x.M : S \to B}}

  \inferrule*[Lab=\freezemlLab{Lam-Ascribe}]
    {\typ{\Delta; \Gamma, x : A}{M : B}}
    {\typ{\Delta; \Gamma}{\lambda (x : A).M : A \to B}} \\

  \inferrule*[Lab=\freezemlLab{Let}]
    {(\Delta', \Delta'') = \mgend{A'}{M} \\
     (\Delta, \Delta'', M, A') \Updownarrow A \\
     \Delta, \Delta''; \Gamma \vdash M : A' \\
     \typ{\Delta; \Gamma, x : A}{N : B} \\\\
     \meta{principal}(\Delta, \Gamma, M, \Delta'', A')
    }
    {\typ{\Delta; \Gamma}{\Let \; x = M\; \In \; N : B}}

  \inferrule*[Lab=\freezemlLab{Let-Ascribe}]
    {(\Delta', A') = \msplit(A, M) \\
     \typ{\Delta, \Delta'; \Gamma}{M : A'} \\
     \typ{\Delta; \Gamma, x : A}{N : B}
    }
    {\typ{\Delta; \Gamma}{\Let \; (x : A) = M\; \In \; N : B}}

\end{mathpar}

\caption{\freezeml typing rules}
\label{fig:freezeml-typing}
\end{figure}


%% file: figures/freezeml_aux_judgments.tex
\begin{figure}
\raggedright
$\boxed{(\Delta, \Delta', M, A') \Updownarrow A}$
\begin{mathpar}
\inferrule
  {M \in \dec{GVal}}
  {(\Delta, \Delta', M, A') \Updownarrow \forall \Delta'. A'}

\inferrule
  {\Delta \vdash \rsubst : \Delta' \Rightarrow_\mono \cdot \\ M \not\in \dec{GVal}}
  {(\Delta, \Delta', M, A') \Updownarrow \rsubst(A')}
\end{mathpar}

\begin{mathpar}
\mgend{A}{M} =
\bl
\left\{
\ba{@{~}l@{~}l@{}@{\quad}l}
  (\Delta', &\Delta') &\text{if $M \in \dec{GVal}$} \\
  (\cdot,   &\Delta') &\text{otherwise} \\
\ea
\right. \\
\text{where }\Delta' = \ftv(A) - \Delta \\
\el

\bl
\meta{split}(\forall \Delta.H, M) =
\left\{
\ba{@{~}l@{}@{\quad}l}
  (\Delta, H)              & \text{if $M \in \dec{GVal}$} \\
  (\cdot,\forall \Delta.H) & \text{otherwise}
\ea
\right.
\el

\end{mathpar}
\begin{mathpar}
\bl
\meta{principal}(\Delta, \Gamma, M, \Delta', A') = \\
\hspace{0.75cm}\Delta' = \ftv(A') - \Delta \;\text{ and }\;
\Delta, \Delta'; \Gamma \vdash M : A' \;\text{ and } \\
\hspace{0.75cm}\bl
(\text{for all }\Delta'', A'' \mid
       \bl
       \text{if }
       \Delta'' = \ftv(A'') - \Delta \text{ and } \\
       \Delta, \Delta''; \Gamma \vdash M : A'' \\
       \text{then there exists }
       \rsubst
       \text{ such that } \\
       \;\Delta  \vdash \rsubst : \Delta' \Rightarrow_\poly \Delta''
       \text{ and }
       \rsubst(A') = A'' )\\
       \el \\
\el
\el
\end{mathpar}
\caption{\freezeml auxiliary definitions}
\label{fig:freezeml-aux-judgments}
\end{figure}


%% file: figures/freezeml_well_scopedness.tex
\begin{figure}
\begin{mathpar}
  \inferrule
    {\strut}
    {\termwf{\Delta}{\freeze{x}}}

  \inferrule
    {\strut}
    {\termwf{\Delta}{x}}

  \inferrule
    {\termwf{\Delta}{M}}
    {\termwf{\Delta}{\lambda x.M}}
\\
  \inferrule
    {\Delta \vdash A :\poly\\\\
     \termwf{\Delta}{M}
    }
    {\termwf{\Delta}{\lambda (x : A).M }}

  \inferrule
    {\termwf{\Delta}{M} \\\\
     \termwf{\Delta}{N}
    }
    {\termwf{\Delta}{M\,N}}

  \inferrule
    {\termwf{\Delta}{M } \\
     \termwf{\Delta}{N}
    }
    {\termwf{\Delta}{\Let \; x = M\; \In \; N }}

  \inferrule
    {\Delta \vdash A : \poly \\
     (\Delta', A') = \msplit(A, M) \\
     \termwf{\Delta, \Delta'}{M} \\
     \termwf{\Delta}{N}
    }
    {\termwf{\Delta}{\Let \; (x : A) = M\; \In \; N}}
\end{mathpar}

\caption{Well-Scopedness of \freezeml Terms}
\label{fig:terms-well-scopedness}
\end{figure}


%% file: figures/translation_freezeml_to_systemf.tex
\begin{figure*}[htb]
\begin{mathpar}
  \transfreezemlfof[\Bigg]
    {\inferrule
      {x : A \in \Gamma}
      {\typ{\Delta; \Gamma}{\freeze{x} : A}}}
  = x

\transfreezemlfof[\Bigg]
    {\inferrule
      {\typ{\Delta; \Gamma, x : S}{M : B}}
      {\typ{\Delta; \Gamma}{\lambda x.M : S \to B}}
    }
  = \lambda x^S.\transfreezemlfof{M}

  \transfreezemlfof[\Bigg]
    {\inferrule
      {\typ{\Delta; \Gamma, x : A}{M : B}}
      {\typ{\Delta; \Gamma}{\lambda (x : A).M : A \to B}}
    }
  = \lambda x^A.\transfreezemlfof{M}

  \hspace{-10pt}\transfreezemlfof[\Bigg]{\raisebox{-0pt}
    {$\inferrule
      {{x : \forall \Delta'.H \in \Gamma} \\
        \Delta \vdash \rsubst : \Delta' \Rightarrow_\poly \cdot
      }
      {\typ{\Delta; \Gamma}{x : \rsubst(H)}}
    $}} = x \: \rsubst(\Delta')

\hspace{-10pt}\transfreezemlfof[\Bigg]{\raisebox{-0pt}
    {$\inferrule
      {\typ{\Delta; \Gamma}{M : A \to B} \\
        \typ{\Delta; \Gamma}{N : A}
      }
      {\typ{\Delta; \Gamma}{M\,N : B}}
    $}}
  = \transfreezemlfof{M} \; \transfreezemlfof{N}

  \transfreezemlfof[\mediumginormous]
    {\raisebox{-22pt}{$\inferrule*
      { (\Delta', \Delta'') = \mgend{A'}{M} \\
        (\Delta, \Delta'', M, A') \Updownarrow A \\\\
        \Delta, \Delta''; \Gamma \vdash M : A' \\
        \typ{\Delta; \Gamma, x : A}{N : B} \\\\
        {\color{gray} \meta{principal}(\Delta, \Gamma, M, \Delta'', A')}
     }
      {\typ{\Delta; \Gamma}{\Let \; x = M\; \In \; N : B}}
    $}}
  =
  \begin{array}{l}
    \Let \: x^A = \Lambda \, \Delta'. \transfreezemlfof{M}\\
    \In\; \transfreezemlfof{N}
  \end{array}
  =
  \transfreezemlfof[\mediumginormous]
    {\raisebox{-22pt}{$\inferrule*
      { (\Delta', A') = \msplit(A, M) \\\\
        \typ{\Delta, \Delta'; \Gamma}{M : A'} \\\\
        \typ{\Delta; \Gamma, x : A}{N : B}
      }
      {\typ{\Delta; \Gamma}{\Let \; (x : A) = M\; \In \; N : B}}
    $}}

\end{mathpar}
\caption{Translation from \freezeml to System F
}
\label{fig:trans-freezeml-to-f}
\end{figure*}


%% file: figures/freezeml_inference_algorithm.tex
\begin{figure}[t]
\raggedright
\[
\bl
\scalebox{0.9}{$\Infer : (\meta{PEnv} \times \meta{KEnv} \times \meta{TEnv} \times \meta{Term}) \rightharpoonup (\meta{KEnv} \times \meta{Subst} \times \meta{Type}) \medskip$} \\

\Infer(\Delta, \Theta, \Gamma, \freeze{x}) = \\
\quad\mreturn\; (\Theta, \idsubst, \Gamma(x)) \medskip \\

\Infer(\Delta, \Theta, \Gamma, x) = \\
\quad\bl
     \mlet\; \forall \many{a}.H = \Gamma(x) \\
     \mfresh\; \many{b} \\
     \mreturn\; ((\Theta,\many{b : \poly}), \idsubst, H[\many{b}/\many{a}])) \\
     \el \medskip \\

\Infer(\Delta, \Theta, \Gamma, \lambda x.M) =\\
\quad\bl
     \mfresh\; a \\
     \mlet\; (\Theta_1, \fsubst[a \mapsto S], B) = \Infer(\Delta, (\Theta, a:\mono), (\Gamma,x : a), M) \\
     \mreturn\; (\Theta_1, \fsubst, S \to B) \\
     \el \medskip \\

\Infer(\Delta, \Theta, \Gamma, \lambda (x : A).M) = \\
\quad\bl
     \mlet\; (\Theta_1, \fsubst, B) = \Infer(\Delta, \Theta, (\Gamma,x : A), M) \\
     \mreturn\; (\Theta_1, \fsubst, A \to B) \\
     \el \medskip \\

\Infer(\Delta, \Theta, \Gamma, M\ N) = \\
\quad\bl
     \mlet\; (\Theta_1, \fsubst_1, A') = \Infer(\Delta, \Theta, \Gamma, M) \\
     \mlet\; (\Theta_2, \fsubst_2, A) = \Infer(\Delta, \Theta_1, \fsubst_1(\Gamma), N) \\
     \mfresh\; b \\
     \mlet\; (\Theta_3, \fsubst_3[b \mapsto B]) = \unify(\Delta, (\Theta_2,b:\poly), \fsubst_2(A'), A \to b) \\
     \mreturn\; (\Theta_3, \fsubst_3 \circ \fsubst_2 \circ \fsubst_1, B) \\
     \el \medskip \\

\Infer(\Delta, \Theta, \Gamma, \Let\; x = M \;\In\; N) = \\
\quad\bl
     \mlet\; (\Theta_1, \fsubst_1, A) = \Infer(\Delta, \Theta, \Gamma, M) \\
     \mlet\; \Delta' = \ftv(\fsubst_1) - \Delta \\
     \mlet\; (\Delta'', \Delta''') = \mgend[(\Delta, \Delta')]{A}{M}  \\
     \mlet\; \Theta_1' = \meta{demote}(\mono, \Theta_1, \Delta''') \\
     \mlet\; (\Theta_2, \fsubst_2, B) = \Infer(\Delta, \Theta_1' - \Delta'', \fsubst_1(\Gamma),x:\forall \Delta''.A, N) \\
     \mreturn\; (\Theta_2, \fsubst_2 \circ \fsubst_1, B) \\
     \el \medskip \\

\Infer(\Delta, \Theta, \Gamma, \Let\; (x : A) = M \;\In\; N) = \\
\quad\bl
     \mlet\; (\Delta', A') = \msplit(A, M) \\
     \mlet\; (\Theta_1, \fsubst_1, A_1) = \Infer((\Delta, \Delta'), \Theta, \Gamma, M) \\
     \mlet\; (\Theta_2, \fsubst_2') =
       \unify((\Delta, \Delta'), \Theta_1, A', A_1) \\

     \mlet\; \fsubst_2 = (\fsubst_2' \circ \fsubst_1) \\
     \massert\; \ftv(\fsubst_2) \mathbin{\#} \Delta' \\
     \mlet\; (\Theta_3, \fsubst_3, B) =
       \Infer(\Delta, \Theta_2, (\fsubst_2(\Gamma), x: A), N) \\
     \mreturn\; (\Theta_3, \fsubst_3 \circ \fsubst_2, B) \\
     \el \medskip
\el
\]
\caption{Type Inference Algorithm}
\label{fig:inference}
\end{figure}

%% file: figures/system_f_syntax.tex
\begin{figure}
\begin{syntax}
\slab{Type Variables}    &a, b, c \\
\slab{Type Constructors} &\tc     &::= & \scalebox{0.9}{$\Int \mid \Bool \mid \List \mid \mathord{\to} \mid \mathord{\times} \mid \dots$} \\
\slab{Types}             &A, B    &::=& a \mid \tc\,\many{A} \mid \forall a.A \\
\slab{Term Variables}    &x, y, z \\
\slab{Terms}             &M, N    &::= & \scalebox{0.9}{$x \mid \lambda x^A.M \mid M\,N \mid \Lambda a.V \mid M\,A$} \\
\slab{Values}            &V, W    &::= & I \mid \lambda x^A.M \mid \Lambda a.V \\
\slab{Instantiations}    &I       &::= & x \mid I\,A \\
\slab{Kind Environments} &\Delta  &::=& \cdot \mid \Delta, a \\
\slab{Type Environments} &\Gamma  &::=& \cdot \mid \Gamma, x:A \\
\end{syntax}
\caption{System F Syntax}
\label{fig:system-f-syntax}
\end{figure}


%% file: figures/system_f_kinding_typing.tex
\begin{figure}
\raggedright
$\boxed{\Delta \vdash A : \poly}$
\vspace{-0.5cm}
\begin{mathpar}
  \inferrule
    {a \in \Delta}
    {\Delta \vdash a : \poly}

  \inferrule
    {\arity(D) = n \\\\
     \Delta \vdash A_1 : \poly \, \cdots \, \Delta \vdash A_n : \poly}
    {\Delta \vdash \tc\,\many{A} : \poly}

  \inferrule
    {\Delta, a \vdash A : \poly}
    {\Delta \vdash \forall a.A : \poly}
\end{mathpar}

\raggedright
$\boxed{\typ{\Delta; \Gamma}{M : A}}$
\vspace{-0.5cm}
\begin{mathpar}
  \inferrule*[Lab=\fLab{Var}]
    {x : A \in \Gamma}
    {\typ{\Delta; \Gamma}{x : A}}

  \inferrule*[Lab=\fLab{App}]
    {\typ{\Delta;\Gamma}{M : A \to B} \\\\
     \typ{\Delta;\Gamma}{N : A}
    }
    {\typ{\Delta;\Gamma}{M\,N : B}}

  \inferrule*[Lab=\fLab{PolyLam}]
    {\typ{\Delta, a;\Gamma}{V : A}}
    {\typ{\Delta;\Gamma}{\Lambda a . V : \forall a . A}}

  \inferrule*[Lab=\fLab{Lam}]
    {\typ{\Delta;\Gamma, x : A}{M : B}}
    {\typ{\Delta;\Gamma}{\lambda x^A.M : A \to B}}

  \inferrule*[Lab=\fLab{PolyApp}]
    {\typ{\Delta;\Gamma}{M : \forall a . B}\\
    \Delta \vdash A : \poly}
    {\typ{\Delta;\Gamma}{M\,A : B[A/a]}}
\end{mathpar}

\caption{System F Kinding and Typing Rules}
\label{fig:system-f-typing}
\end{figure}


%% file: figures/system_f_equational_rules.tex
\begin{figure}
\raggedright

\begin{minipage}{0.45\columnwidth}
$\beta$-rules
\begin{equations}
(\lambda x^A. V)\,W &\conv& V[W/x] \\
(\Lambda a. V)\,A   &\conv& V[A/a] \medskip \\
\end{equations}
\end{minipage}
\hspace{10pt}
\begin{minipage}{0.4\columnwidth}
$\eta$-rules
\begin{equations}
\lambda x^A. M\, x &\conv& M \\
\Lambda a. V\,a    &\conv& V \medskip \\
\end{equations}
\end{minipage}

\caption{System F Equational Rules}
\label{fig:system-f-equations}
\end{figure}


%% file: figures/ml_syntax.tex
\begin{figure}
\begin{syntax}
\slab{Type Variables}      &a, b, c \\
\slab{Type Constructors}   &\tc     &::= & \scalebox{0.9}{$\Int \mid \Bool \mid \List \mid \mathord{\to} \mid \times \mid \dots$} \\
\slab{Monotypes}           &S, T    &::= & a \mid \tc\,\many{S} \\
\slab{Type Schemes}        &P, Q    &::= & \forall \many{a}.S \\
\slab{Type Instantiations} &\rsubst &::= & \emptyset \mid \rsubst[a \mapsto S] \\
\slab{Term Variables}      &x, y, z \\
\slab{Terms}               &M, N    &::= & \scalebox{0.85}{$x \mid \lambda x.M \mid M\,N \mid \Let\; x = M \;\In\; N$} \\
\slab{Values}              &V, W    &::= & \scalebox{0.85}{$x \mid \lambda x.M \mid \Let\; x = V \;\In\; W$} \\
\slab{Kinds}               &K       &::= & \mono \mid \poly \\
\slab{Kind Environments}   &\Delta  &::= & \cdot \mid \Delta, a \\
\slab{Type Environments}   &\Gamma  &::= & \cdot \mid \Gamma, x:P \\
\end{syntax}

\caption{ML Syntax}
\label{fig:ml-syntax}
\end{figure}


%% file: figures/ml_kinding_typing.tex
\begin{figure}
\raggedright
$\boxed{\Delta \vdash S : \mono}$
$\boxed{\Delta \vdash P : \poly}$
\vspace{-0.5cm}
\begin{mathpar}
  \inferrule
    {a \in \Delta}
    {\Delta \vdash a : \mono}

  \inferrule
    {\Delta, \Delta' \vdash S : \mono}
    {\Delta \vdash \forall \Delta'.S : \poly}

  \inferrule
    {\arity(D) = n \\\\
     \Delta \vdash S_1 : \mono \quad \cdots \quad \Delta \vdash S_n : \mono}
    {\Delta \vdash \tc\,\many{S} : \mono}
\end{mathpar}

\raggedright
$\boxed{\typ{\Delta; \Gamma}{M : S}}$
\begin{mathpar}
  \inferrule*[Lab=\mlLab{Var}]
    {x : \forall \Delta'.S \in \Gamma \\
     \Delta \vdash \rsubst : \Delta' \Rightarrow \cdot
    }
    {\typ{\Delta; \Gamma}{x : \rsubst(S)}}

  \inferrule*[Lab=\mlLab{Lam}]
    {\typ{\Delta; \Gamma, x : S}{M : T}}
    {\typ{\Delta; \Gamma}{\lambda x.M : S \to T}}
\\
  \inferrule*[Lab=\mlLab{App}]
    {\typ{\Delta; \Gamma}{M : S \to T} \\\\
     \typ{\Delta; \Gamma}{N : S}
    }
    {\typ{\Delta; \Gamma}{M\,N : T}}

  \inferrule*[Lab=\mlLab{Let}]
    {\Delta' = \mgend{S}{M} \quad
     \typ{\Delta, \Delta'; \Gamma}{M : S} \\\\
     P = \forall \Delta'.S \quad
     \typ{\Delta; \Gamma, x : P}{N : T}
    }
    {\typ{\Delta; \Gamma}{\Let \; x = M\; \In \; N : T}}
\end{mathpar}

\raggedright
$\boxed{\Delta \vdash \rsubst : \Delta' \Rightarrow \Delta''}$
\begin{mathpar}
\inferrule
  { }
  {\Delta \vdash \emptyset : \cdot \Rightarrow \Delta'}

\inferrule
  {\Delta \vdash \rsubst : \Delta' \Rightarrow \Delta'' \\
   \Delta, \Delta'' \vdash S : \mono}
  {\Delta \vdash \rsubst[a \mapsto S] : (\Delta', a) \Rightarrow \Delta''}
\end{mathpar}

\begin{mathpar}
\mgend{S}{M} =
\begin{cases}
  \, \ftv(S) - \Delta    &\text{if $M$ is a value} \\
  \, \cdot               &\text{otherwise}
\end{cases}
\end{mathpar}

\caption{ML Kinding and Typing Rules}
\label{fig:ml-typing}
\end{figure}


%% file: figures/ml_to_system_f.tex
\begin{figure}
\begin{align*}
  \transmltof[\ginormous]{\raisebox{-7pt}
    {$\inferrule
      {{x : \forall \Delta'.S \in \Gamma} \\\\
        \Delta \vdash \rsubst : \Delta' \Rightarrow \cdot
      }
      {\typ{\Delta; \Gamma}{x : \rsubst(S)}}
    $}} &= x \: \rsubst(\Delta') \\
  \transfreezemlfof[\Ginormous]
    {\raisebox{-27pt}{$\inferrule*
      {\Delta' = \mgen(\Delta, S, M) \\\\
       \typ{\Delta, \Delta'; \Gamma}{M : S} \\\\
       P = \forall \Delta'.S \\\\
       \typ{\Delta; \Gamma, x : P}{N : T}
      }
      {\typ{\Delta; \Gamma}{\Let \; x = M\; \In \; N : T}}
    $}}
  &=
    \begin{aligned}[t]
    &\Let \: x^{\forall \Delta'.S} = \Lambda \, \Delta'. \transfreezemlfof{M} \\
    &\In\; \transfreezemlfof{N} \\
    \end{aligned} \\
  \transmltof[\Bigg]
    {\inferrule
      {\typ{\Delta; \Gamma, x : S}{M : T}}
      {\typ{\Delta; \Gamma}{\lambda x.M : S \to T}}
    }
  &= \lambda x^S.\transmltof{M} \\
  \transmltof[\ginormous]{\raisebox{-7pt}
    {$\inferrule
      {\typ{\Delta; \Gamma}{M : S \to T} \\\\
        \typ{\Delta; \Gamma}{N : S}
      }
      {\typ{\Delta; \Gamma}{M\,N : T}}
    $}}
  &= \transmltof{M} \; \transmltof{N}
\end{align*}
\caption{Translation from ML to System F}
\label{fig:ml-to-f}
\end{figure}
